\documentclass[a4paper,onecolumn,10pt,accepted=2024-02-14]{quantumarticle}
\pdfoutput=1

\usepackage{amsmath}
\usepackage{amssymb}
\usepackage{amsthm}
\usepackage{graphicx}
\usepackage{xcolor}
\usepackage[english]{babel}
\usepackage{csquotes}
\usepackage{braket}
\usepackage{hyperref}
\usepackage{mathtools}
\usepackage{enumitem}
\usepackage{IEEEtrantools}
\usepackage{relsize}
\usepackage{placeins}
\usepackage{comment}
\usepackage[normalem]{ulem}
\usepackage{systeme}
\usepackage{physics}

\usepackage[caption=false]{subfig}

\usepackage{tikz-cd} 
\usepackage{adjustbox}

\usepackage{tikz}
\usepackage{pgffor}

\usepackage{pgfplots}
\usetikzlibrary{pgfplots.groupplots}
\pgfplotsset{compat=1.11}
\usepgfplotslibrary{fillbetween}

\usepackage{tikz}
\usetikzlibrary{
  shapes,
  shapes.geometric,
	trees,
	matrix,
  positioning,
    pgfplots.groupplots,
  }
\newcommand{\ipic}[3][-0.5]{\raisebox{#1\height}{\scalebox{#3}{\includegraphics{#2}}}}

\PassOptionsToPackage{pdftex,hyperfootnotes=false,pdfpagelabels}{hyperref}
\usepackage{hyperref}
\usepackage[normalem]{ulem}
\usepackage[backend=bibtex,style=phys,biblabel=brackets]{biblatex}
\AtEveryBibitem{\clearfield{note}}

\usepackage[allcolors=quantumviolet]{hyperref}

\newcommand{\Id}{\ensuremath{\mathbb{I}}}

\providecommand{\myvec}[1]{\ensuremath{\boldsymbol{#1}}}

\providecommand{\tt}{\ensuremath{\myvec{t}}}

\newtheorem{theorem}{Theorem}
\newtheorem{lemma}[theorem]{Lemma}
\newtheorem{proposition}[theorem]{Proposition}
\newtheorem{definition}[theorem]{Definition}
\newtheorem{corollary}[theorem]{Corollary}

\newtheorem{observation}[theorem]{Observation}
\newtheorem{example}[theorem]{Example}

\newcommand{\Flip}{\mathbb{F}}
\newcommand{\Pmath}{\mathrm{P}}
\newcommand{\Idd}{\mathbb{I}}

\newcommand{\CNOT}{\ensuremath{\operatorname{CNOT}}}
\newcommand{\Hadamard}{\ensuremath{\operatorname{H}}}
\newcommand{\PhaseS}{\ensuremath{\operatorname{S}}}

\newcommand{\ExU}{\underset{U\sim\mu_H}{\mathbb{E}}}

\newcommand{\Expsi}{\underset{\ket{\psi}\sim\mu_H}{\mathbb{E}}}
\newcommand{\MatC}[1]{\mathcal{L}\left(\mathbb{C}^{#1}\right)}

\newcommand{\MomOp}[2][k]{\mathcal{M}^{(#1)}_{\mu_H}\!\left(#2\right)}
\newcommand{\hs}[2]{\langle #1, #2\rangle_{HS}}

\renewcommand{\Id}{I}
\newcommand{\Ug}{\mathrm{U}}
\newcommand{\kket}[1]{|#1\rangle\!\rangle}

\newcommand{\bbra}[1]{\langle\!\langle#1|}
\newcommand{\bbrakket}[2]{\langle\!\langle#1 \rvert #2\rangle\!\rangle}

\renewcommand{\O}{\ensuremath{\operatorname{O}}}

\newcommand{\fu}{Dahlem Center for Complex Quantum Systems, Freie Universit\"{a}t Berlin, 14195 Berlin, Germany}

\begin{document}

\title{Introduction to Haar Measure Tools in Quantum Information: A Beginner's Tutorial}

\author{Antonio Anna Mele}
\email{a.mele@fu-berlin.de}
\affiliation{\fu}

\begin{abstract}
The Haar measure plays a vital role in quantum information, but its study often requires a deep understanding of representation theory, posing a challenge for beginners. This tutorial aims to provide a basic introduction to Haar measure tools in quantum information, utilizing only basic knowledge of linear algebra and thus aiming to make this topic more accessible.
The tutorial begins by introducing the Haar measure with a specific emphasis on characterizing the moment operator, an essential element for computing integrals over the Haar measure. It also covers properties of the symmetric subspace and introduces helpful tools like tensor network diagrammatic notation, which aid in visualizing and simplifying calculations.
Next, the tutorial explores the concept of unitary designs, providing equivalent definitions, and subsequently explores approximate notions of unitary designs, shedding light on the relationships between these different notions.
Practical examples of Haar measure calculations are illustrated, including the derivation of well-known formulas such as the twirling of a quantum channel. Lastly, the tutorial showcases the applications of Haar measure calculations in quantum machine learning and classical shadow tomography. 

\end{abstract}

\maketitle
\tableofcontents

\clearpage
\pagestyle{plain}
\pagenumbering{arabic}

\section{Introduction}
The notion of Haar measure formalizes the fundamental concept of drawing unitary matrices uniformly at random.
In the context of quantum information, unitary matrices often represent quantum evolutions that are commonly described using quantum circuits. The applications of the Haar measure are widespread across various domains within quantum information, such as quantum tomography~\cite{Huang_2020,Haah_2017, Elben_2022,odonnell2015efficient,kueng2014low,PhysRevLett.125.200501,Guţă_2020,cieśliński2023analysing},~computational advantage of random quantum circuits sampling \cite{harrow2023approximate,Bouland_2018,movassagh2020quantum,dalzell2021random,Hangleiter_2018,Bouland18,Pashayan_2020,hangleiter2023computational},~benchmarking of quantum devices \cite{Popescu_2006, Emerson_2005,Helsen_2022,RBMagesan},~quantum foundations and communication~\cite{DiVincenzo_2002,Harrow_2004,Ambainis_2009,PhysRevA.60.1888,Abeyesinghe_2009,horodecki1998reduction,batlevallespir2006characterization,_yczkowski_1998,Zyczkowski_2001}, quantum machine learning~\cite{McClean_2018, Holmes_2022, Cerezo_2021,napp2022quantifying}, quantum many-body and high energy physics~\cite{Skinner,Fisher_2023,ChanPhysics,Hayden_2007,Haferkamp_2022,Brand_o_2021,Roberts_2017,Liu_2018,Onorati_2017,Oles_PhysRevLett.121.126803,Movassagh2011}.

Although the Haar measure serves as a fundamental tool in quantum information, its study can be challenging for beginners due to its reliance on advanced concepts from representation theory. 
In order to make this topic more accessible, this tutorial provides an introduction to Haar measure tools using only concepts from linear algebra. Throughout the tutorial, we intentionally avoid delving into representation theory to enhance accessibility.
For further reading, one can refer, for example, to~\cite{KuengNotes_2019,christandl2006structure,Kliesch_2021,Haferkamp2022,low2010pseudorandomness,Brand_o_2021,watrous_2018,Collins_2006,Roberts_2017}.

The tutorial is structured as follows.
Section~\ref{sec:notation} offers an overview of the notation and preliminary concepts that will be used throughout the tutorial.
Section~\ref{sec:Haar} focuses on introducing the Haar measure, with a specific emphasis on characterizing the moment operator—a crucial quantity for computing integrals over the Haar measure.
In Section~\ref{sec:sym}, the symmetric and antisymmetric subspaces are introduced, highlighting their properties and their connection to the Haar measure over pure states.
Sections~\ref{sec:vecform} and~\ref{sec:diagram} introduce tools that facilitate calculations in Haar measure applications. Section \ref{sec:vecform} presents the vectorization formalism, while Section~\ref{sec:diagram} introduces the tensor networks diagrammatic notation, providing visual representations that enhance comprehension and streamline computations.
Section~\ref{sec:design} introduces the concept of unitary designs—a method employed to mimic some properties of the Haar measure, facilitating efficient protocols in quantum computing. Building upon this foundation, Section~\ref{sec:apprdesign} explores approximate notions of unitary designs, elucidating the relationships between the different introduced notions.
In Section~\ref{sec:applications}, the tutorial showcases practical examples of Haar measure calculations. We begin by deriving well-known formulas such as the twirling of a quantum channel, average gate fidelity, average purity of reduced states in bipartite systems, and Haar averages of observables expectation values. Furthermore, we demonstrate how these expected value calculations can be translated into probability statements using concentration inequalities. Finally, we delve into two in-depth applications: Barren Plateaus phenomena in Quantum Machine Learning and Classical Shadow Tomography, both of which rely on the theory of $k$-designs.

\section{Notation and Preliminaries}
\label{sec:notation}
We use the following notation throughout this tutorial. $\mathcal{L}(\mathbb{C}^d)$ is the set of linear operators that act on the $d$-dimensional complex vector space $\mathbb{C}^d$, while $\mathrm{Herm}(\mathbb{C}^d)$ is the set of Hermitian operators on $\mathbb{C}^d$. The identity operator is denoted by $I$, and we define the operator $\mathbb{I}\coloneqq I\otimes I$ as the tensor product of two identity operators. The unitary group is denoted by $\Ug(d)$ and is defined as the set of operators $U \in \mathcal{L}(\mathbb{C}^d)$ such that $U^\dagger U = \Id$. Furthermore, we use the notation $[d]$ to denote the set of integers from 1 to $d$, i.e., $[d] \coloneqq \{1,\dots, d\}$.

Let $v \in \mathbb{C}^d$ be a vector, and let $p \in \left[1,\infty\right]$. The $p$-norm of $v$ is denoted by $\norm{v}_p$, and is defined as $\norm{v}_p \coloneqq  (\sum_{i=1}^d |v_i|^p)^{1/p}$. 
The Schatten $p$-norm of a matrix $A\in \MatC{d}$ is given by $\norm{A}_p\coloneqq \Tr\small(\small(\sqrt{A^\dagger A}\small)^p\small)^{1/p}$, which corresponds to the $p$-norm of the vector of singular values of $A$.
The trace norm and the Hilbert-Schmidt norm are important instances of Schatten $p$-norms and are respectively denoted as $\|\cdot\|_1$ and $\|\cdot\|_2$. The Hilbert-Schmidt norm is induced by the Hilbert-Schmidt scalar product $\hs{A}{B}\coloneqq \Tr(A^\dagger B)$ for $A,B \in \mathcal{L}\left(\mathbb{C}^d\right)$. 
The infinite norm, denoted as $\norm{\cdot}_\infty$, of a matrix is defined as its largest singular value. This norm can be understood as the limit of the Schatten $p$-norm of the matrix as $p$ approaches infinity.
Important facts about the Schatten $p$-norms that will be used in the tutorial are the following. 
For all matrices $A$ and $1\le p\le q$, we have $\norm{A}_{q}\le \norm{A}_{p}$ and $\norm{A}_{p}\le \mathrm{rank}(A)^{(p^{-1}-q^{-1})}\norm{A}_{q}$. 
Additionally, for all unitaries $U$ and $V$, and matrix $A$, we have the unitary invariance property $\norm{UAV}_p=\norm{A}_p$. Furthermore, we also have the tensor product property  $\norm{A\otimes B}_p=\norm{A}_p \norm{B}_p$ and the submultiplicativity property $\norm{A B}_p\le \norm{A}_p \norm{B}_p$, for all matrices $A$ and $B$.

We use the bra-ket notation, where we denote a vector $v \in \mathbb{C}^d$ using the ket notation $\ket{v}$ and its adjoint using the bra notation $\bra{v}$.
We refer to a vector $\ket{\psi}\in \mathbb{C}^d$ as a (pure) state if $\norm{\ket{\psi}}_2=1$. 
The canonical basis of $\mathbb{C}^d$ is denoted by $\{\ket{i}\}_{i=1}^{d}$. We define the non-normalized maximally entangled state $\ket{\Omega}$ as $\ket{\Omega}\coloneqq \sum^d_{i=1}\ket{i}\otimes\ket{i}=\sum^d_{i=1}\ket{i,i}$.
We define the set of density matrices (or quantum states) as $\mathcal{S}\left(\mathbb{C}^d\right)\coloneqq \{\rho \in \mathcal{L}\left(\mathbb{C}^d\right) \,:\,\rho \ge 0,\,\Tr(\rho)=1\}$.
A quantum channel $\Phi:\mathcal{L}(\mathbb{C}^d)\rightarrow\mathcal{L}(\mathbb{C}^d)$ is a linear map that is completely positive and trace-preserving.  In particular, the completely positive condition means that for all positive operators $\sigma \in \mathcal{L}(\mathbb{C}^d\otimes \mathbb{C}^D)$, for all $D\in \mathbb{N}$, the operator $\Phi \otimes \mathcal{I}(\sigma)$ is also positive. Here, $\mathcal{I}:\MatC{D}\rightarrow \MatC{D}$ denotes the identity map, which simply maps any $A\in \MatC{D}$ to itself.
Additionally, every quantum channel $\Phi$ can be expressed in terms of $d^2$ Kraus operators, i.e., there exist $\{K_i\}^{d^2}_{i=1}$ operators such that $\Phi\left(\cdot\right)=\sum^{d^2}_{i=1} K_i \left(\cdot\right)K^\dagger_i$, subject to the condition that $\sum^{d^2}_{i=1} K^\dagger_i K_i= I$ in order to satisfy the trace-preserving property.

\section{Haar measure and moment operator}
\label{sec:Haar}
\begin{definition}[Haar measure]
The Haar measure on the unitary group $\Ug(d)$ is the unique probability measure $\mu_H$ that is both left and right invariant over the group $\Ug(d)$, i.e., for all integrable functions $f$ and for all $V \in \Ug(d)$, we have:
\begin{align}
\int_{\Ug(d)} f\left(U\right)d\mu_H(U)=\int_{\Ug(d)} f\left(UV\right)d\mu_H(U)=\int_{\Ug(d)} f\left(VU\right)d\mu_H(U).
\end{align}
\end{definition}

In fact, for compact groups such as the unitary group, there exists a unique probability measure that is both left and right invariant under group multiplication \cite{Simon1995RepresentationsOF}. However, we will not delve deeper into the theory of compact groups and measures in this tutorial, as our focus is on applications of the tools we describe. For a more comprehensive treatment of this topic, we recommend referring to \cite{Simon1995RepresentationsOF,Collins_2006}.

The Haar measure is a probability measure, satisfying the properties $\int_{S}1\, d\mu_H(U) \ge 0$ for all sets $S \subseteq \Ug(d)$ and $\int_{\Ug(d)}1 \, d\mu_H(U) = 1$. Consequently, we can denote the integral of any function $f(U)$ over the Haar measure as the expected value of $f(U)$ with respect to the probability measure $\mu_H$, denoted as $\ExU[f(U)]$.
\begin{align}
\ExU\left[f\left(U\right)\right]\coloneqq \int_{\Ug(d)} f\left(U\right)d\mu_H(U).
\end{align}
When $f(U)$ is a matrix function, the expected value is understood to be the expected value of each of its entries.

We can prove the following proposition, which shows that any integral involving an \emph{unbalanced} product of matrix entries of $U$ and $U^*$ must vanish. In the following, we use the notation $U^{\otimes k}\coloneqq \underbrace{U\otimes\cdots \otimes U}_{k \,\text{times}}$.\footnote{The significance of investigating tensor powers of unitaries and their adjoints, i.e. $U^{\otimes{k}} \otimes U^{*\otimes{k}}$, becomes evident when considering applications (Section \ref{sec:applications}). 
In brief, within various scenarios, the computation of integrals over the Haar measure of the relevant quantity reduces to evaluating integrals of homogeneous polynomials of degree $k$ in the matrix elements of $U$ and $U^*$, for some $k\in \mathbb{N}$.  Remarkably, any such polynomial can be expressed as $\Tr(A U^{\otimes{k}} \otimes U^{*\otimes{k}})$, where $A$ is a matrix that contains the coefficients of the polynomial.}
\begin{proposition}\
Let $k_1,k_2 \in \mathbb{N}$. If $k_1\neq k_2$, then we have $\ExU\left[U^{\otimes k_1} \otimes U^{* \otimes k_2}\right]=0$.
\end{proposition}
\begin{proof}
We can use the right invariance multiplying $U$ by the unitary $\exp(i\frac{\pi}{k_1-k_2})\Id$:
\begin{align}
\ExU\left[U^{\otimes k_1} \otimes U^{* \otimes k_2}\right]=-\ExU\left[U^{\otimes k_1} \otimes U^{* \otimes k_2}\right],   
\end{align}
from which the claim follows.
\end{proof}
This next proposition will be useful in our subsequent calculations. It states that in Haar integrals, we are free to change the variable $U$ with $U^\dagger$. 
\begin{proposition}
\label{prop:Haar}
    For all integrable functions $f$ defined on $\Ug(d)$, we have that:
    \begin{align}
    \ExU\left[f\left(U^\dagger\right)\right]=\ExU\left[f\left(U\right)\right].
    \end{align}
\end{proposition}
\begin{proof}
    
Let $\mu_{\dagger}$ be a probability measure defined as
\begin{align}
 \int_{\Ug(d)} f\left(U\right)d\mu_{\dagger}(U)\coloneqq \int_{\Ug(d)} f\left(U^{\dagger}\right)d\mu_{H}(U).
\end{align}
We will now show that $\mu_{\dagger}$ is right and left invariant, which implies that it coincides with $\mu_H$ because of the uniqueness of the Haar measure. 
To this end, let $V$ be a fixed unitary matrix:
\begin{align}
 \int_{\Ug(d)} f\left(UV\right)d\mu_{\dagger}(U)&=\int_{\Ug(d)} f\left(U^{\dagger}V\right)d\mu_{H}(U)=\int_{\Ug(d)} f\left(U^{\dagger} V^\dagger V\right)d\mu_{H}(U)\\&=\int_{\Ug(d)} f\left(U^{\dagger}\right)d\mu_{H}(U)=\int_{\Ug(d)} f\left(U\right)d\mu_{\dagger}(U).
\end{align}
This shows that $\mu_{\dagger}$ is right-invariant, as claimed. Left-invariance similarly follows.
\end{proof}
A quantity that plays a crucial role in our analysis is the $k$-th moment operator, where $k$ is a natural number. 
\begin{definition}[$k$-th Moment operator]
The $k$-th moment operator, with respect to the probability measure $\mu_H$, is defined as
$\mathcal{M}_{\mu_H}^{(k)}:\mathcal{L}\left((\mathbb{C}^d)^{\otimes k}\right) \rightarrow\mathcal{L}\left((\mathbb{C}^d)^{\otimes k}\right)$ :
   \begin{align}
       \mathcal{M}_{\mu_H}^{(k)}\!\left(O\right)\coloneqq  \ExU \left[U^{\otimes k} O U^{\dagger \otimes k}\right],
   \end{align}
   for all operators $O\in\mathcal{L}((\mathbb{C}^d)^{\otimes k})$.
\label{def:moment}
\end{definition}
By characterizing the moment operator and computing its matrix elements, we can explicitly evaluate integrals over the Haar measure. 
For example, let $O=\ketbra{i_1,\dots,i_k}{j_1,\dots,j_k}$ with $i_a,j_a\in [d]$ for all $a\in [k]$, we have:
\begin{align}
    \bra{l_1,\dots,l_k}\ExU \left[U^{\otimes k} O U^{\dagger \otimes k}\right]\ket{m_1,\dots,m_k}=\ExU \left[U_{l_1, i_1}\cdots U_{l_k, i_k} U^*_{m_1, j_1}\cdots U^*_{m_k, j_k}\right]
\end{align}
where $l_a,m_a\in [d]$  for all $a\in [k]$.
In order to characterize the moment operator, we need to define the $k$-th order commutant of a set of matrices $S$.

\begin{definition}[Commutant]
\label{def:comm}
Given $S \subseteq \mathcal{L}\left(\mathbb{C}^d\right)$, we define its $k$-th order commutant as \begin{align}
  \mathrm{Comm}(S,k)\coloneqq \{ A \in\mathcal{L}\left((\mathbb{C}^d)^{\otimes k}\right)\colon\left[A,B^{\otimes k}\right]=0 \,\, \forall\, B\in S \}.
\end{align} 
\end{definition}
It is worth noting that $\mathrm{Comm}(S,k)$ is a vector subspace. In the following we will show that the moment operator is the orthogonal projector onto the commutant of the unitary group $\mathrm{Comm}(\Ug(d),k)$ with respect to the Hilbert-Schmidt inner product. In order to do so, we first prove the following Lemma.

\begin{lemma}[Properties of the moment operator]\
\label{le:PropMom}

The moment operator $\mathcal{M}_{\mu_H}^{(k)}\!\left(\cdot\right)\coloneqq \ExU \left[U^{\otimes k} \left(\cdot\right) U^{\dagger \otimes k}\right]$ has the following properties:
\begin{enumerate}
    \item  It is linear, trace-preserving, and self-adjoint with respect to the Hilbert-Schmidt inner product.
    \label{propr:MomHerm}
    \item For all $A \in \mathcal{L}\left((\mathbb{C}^d)^{\otimes k}\right)$, $\mathcal{M}_{\mu_H}^{(k)}\!\left(A\right)\in \mathrm{Comm}(\Ug(d),k)$.
    \label{propr:CoMom}
    \item If $A\in \mathrm{Comm}(\Ug(d),k)$, then $\mathcal{M}_{\mu_H}^{(k)}\!\left(A\right)=A$.
    \label{propr:EigMom}
\end{enumerate}
\end{lemma}
\begin{proof}\
\begin{enumerate}
\item Linearity and the trace preserving property follow easily from Definition~\ref{def:moment}. 
To show that the moment operator is self-adjoint, we need to prove that:
\begin{align}
    \langle \mathcal{M}_{\mu_H}^{(k)}\!\left(A\right),B\rangle_{HS}=\langle A,\mathcal{M}_{\mu_H}^{(k)}\!\left(B\right)\rangle_{HS}.
\end{align}
for all $A,B \in \mathcal{L}\left((\mathbb{C}^d)^{\otimes k}\right)$.
This follows because: 
\begin{align}
\Tr\!\left(\mathcal{M}^\dagger_{k}\!\left(A\right)B\right)&=\ExU\Tr\!\left(U^{\otimes k} A^\dagger U^{\dagger\otimes k}B\right)\\&=\Tr\!\left( A^\dagger \ExU\left[U^{\dagger\otimes k}BU^{\otimes k}\right]\right)\\&=\Tr\!\left(A^\dagger\mathcal{M}_{\mu_H}^{(k)}\!\left(B\right)\right),
\end{align}
where in the last step we used Proposition~\ref{prop:Haar}.
 \item For all $V \in \Ug(d)$, we have that:
\begin{align}
V^{\otimes k}\mathcal{M}_{\mu_H}^{(k)}\!\left(A\right)&=\ExU \left[\left(VU\right)^{\otimes k} A U^{\dagger \otimes k}\right] =\ExU \left[U^{\otimes k} A \left(V^\dagger U\right)^{\dagger \otimes k}\right]\\&=\ExU \left[U^{\otimes k} A U^{\dagger \otimes k}\right]V^{\otimes k}=\mathcal{M}_{\mu_H}^{(k)}\!\left(A\right)V^{\otimes k},
\end{align}
where we used the left invariance of the Haar measure.
        \item Since  $A\in \mathrm{Comm}(\Ug(d),k)$, we have:
\begin{align}
    \mathcal{M}_{\mu_H}^{(k)}\!\left(A\right)=\ExU \left[U^{\otimes k} A U^{\dagger \otimes k}\right]= \ExU \left[A U^{\otimes k} U^{\dagger \otimes k}\right]= A.
\end{align}
\end{enumerate}
\end{proof}

\begin{theorem}[Projector onto the commutant]
\label{prop:projComm}
The moment operator $ \mathcal{M}_{\mu_H}^{(k)}\!\left(\cdot\right)=\ExU \left[U^{\otimes k} \left(\cdot\right) U^{\dagger \otimes k}\right]$ is the orthogonal projector onto the commutant $\mathrm{Comm}\coloneqq \mathrm{Comm}(\Ug(d),k)$ with respect to the Hilbert-Schmidt inner product.
In particular, let $P_1,\ldots,P_{\mathrm{dim}(\mathrm{Comm})}$ be an orthonormal basis of $\mathrm{Comm}$ and let $O \in \mathcal{L}((\mathbb{C}^d)^{\otimes k})$. Then, we have:
\begin{align}
\MomOp{O}=\sum^{\mathrm{dim}\left(\mathrm{Comm}\right)}_{i=1} \langle P_i ,O\rangle_{HS} P_i \,.
\end{align}

\end{theorem}
\begin{proof}
Let us extend the orthonormal basis of the commutant with the orthonormal operators $P_i$ for $i\in \{\mathrm{dim}\left(\mathrm{Comm}\right)+1,\dots,\mathrm{dim}\left(V\right)\}$, where $V\coloneqq \mathcal{L}\left((\mathbb{C}^d)^{\otimes k}\right)$. This extended basis forms an orthonormal basis for $V$. Therefore we have: 
\begin{align}
    \MomOp{O}&=\sum^{\mathrm{dim}\left(V\right)}_{i=1} \langle P_i ,\MomOp{O}\rangle_{HS} P_i 
    \\&=\sum^{\mathrm{dim}\left(\mathrm{Comm}\right)}_{i=1} \langle \MomOp{P_i} ,O\rangle_{HS} P_i + \sum^{\mathrm{dim}\left(V\right)}_{i=\mathrm{dim}\left(\mathrm{Comm}\right)+1} \langle P_i ,\MomOp{O}\rangle_{HS} P_i  \\
    &=\sum^{\mathrm{dim}\left(\mathrm{Comm}\right)}_{i=1} \langle \MomOp{P_i} ,O\rangle_{HS} P_i + 0\\
    &=\sum^{\mathrm{dim}\left(\mathrm{Comm}\right)}_{i=1} \langle P_i ,O\rangle_{HS} P_i ,
    \end{align}
where in the second line we used the fact that the moment operator is self-adjoint (Lemma~\ref{le:PropMom}.\ref{propr:MomHerm}), in the third line that $\MomOp[k]{O}\in \mathrm{Comm}$ (Lemma~\ref{le:PropMom}.\ref{propr:CoMom}) and that $P_i$ with $i\in \{\mathrm{dim}\left(\mathrm{Comm}\right)+1,\dots,\mathrm{dim}\left(V\right)\}$ are in its orthogonal complement, in the fourth line that $P_i \in \mathrm{Comm}$ for $i\in \{1,\dots,\mathrm{dim}\left(\mathrm{Comm}\right)\}$ and Lemma~\ref{le:PropMom}.\ref{propr:EigMom}.
\end{proof}
We have just shown that the moment operator is intimately related to the $k$-order commutant of the unitary group, i.e. all the matrices that commute with $U^{\otimes k}$ for any unitary $U$. A set of operations that surely commutes with $U^{\otimes k}$ is if we exchange some tensor factors. For this purpose, we will now define the operators that implement such transformations, namely the permutation operators.
\begin{definition}[Permutation operators]
    Given $\pi \in S_k$ an element of the symmetric group $S_k$, we define the permutation matrix $V_d(\pi)$ to be the unitary matrix that satisfies:
    \begin{align}
    V_d(\pi) \ket{\psi_1}\otimes \cdots \otimes \ket{\psi_k}=\ket{\psi_{\pi^{-1}(1)}}\otimes \cdots \otimes \ket{\psi_{\pi^{-1}(k)}},
    \label{def:permutation}
    \end{align}
    for all $\ket{\psi_1},\dots,\ket{\psi_k} \in \mathbb{C}^{d}$.
\end{definition}
Note that from this definition it follows that $ V_d(\sigma) V_d(\pi) = V_d(\sigma \pi )$ and  $V_d(\pi^{-1})=V^\dagger_d(\pi)$. 
Equivalently, we can write the permutation matrix as:
\begin{align}
V_d(\pi)=\sum_{i_1,\ldots,i_k \in [d]^{k}} \ketbra{i_{\pi^{-1}(1)},\ldots,i_{\pi^{-1}(k)}}{i_1,\cdots,i_k}.
\end{align}
Thus, we have the following property:
\begin{align}
V_d(\pi) \left(A_1\otimes \cdots \otimes A_k\right)V^\dagger_d(\pi)=A_{\pi^{-1}(1)}\otimes \cdots \otimes A_{\pi^{-1}(k)},
\end{align}
for $A_1,\dots,A_k \in \MatC{d}$.

A crucial and much celebrated result of representation theory is now that the permutation operators characterize \emph{all} possible matrices in the commutant -- this is the Schur-Weyl duality. 
\begin{theorem}[Schur-Weyl duality \cite{goodman2000representations}]
\label{th:SchurW}
The $k$-th order commutant of the unitary group is the span of the permutation operators associated to $S_k$: 
\begin{align}
\mathrm{Comm}(\Ug(d),k)=\operatorname{span} \Big(V_d(\pi):\, \pi \in S_k\Big).
\end{align}
\end{theorem}
We will omit the proof of this theorem here, except for the cases where $k=1$ and $k=2$, which we will show later in this section. Interested readers can refer to Refs.~\cite{goodman2000representations,christandl2006structure,zhang2015matrix} for a detailed exposition.
However we can easily check that $\operatorname{span} \Big(V_d(\pi):\, \pi \in S_k\Big)\subseteq \mathrm{Comm}(\Ug(d),k)$.
To see why this is true, consider an arbitrary permutation $V_d(\pi)$ with $\pi \in S_k$ and $U\in \Ug \left(d\right)$. 
    We have: 
    \begin{align}
        V_d(\pi) U^{\otimes k} \ket{\psi_1}\otimes\cdots\otimes\ket{\psi_k}&=V_d(\pi)  (U\ket{\psi_1})\otimes\cdots\otimes (U\ket{\psi_k})\\
        &=U\ket{\psi_{\pi^{-1}(1)}}\otimes \cdots \otimes U\ket{\psi_{\pi^{-1}(k)}}\\
        &= U^{\otimes k} V_d(\pi)\ket{\psi_1}\otimes\cdots\otimes\ket{\psi_k},
    \end{align}
    for all $\ket{\psi_1},\dots,\ket{\psi_k} \in \mathbb{C}^{d}$. Hence we have that $\left[V_d(\pi) ,U^{\otimes k}\right]=0$ for all $\pi \in S_k$, from which the claim follows.

The permutation matrices form a basis for the $k$-th order commutant of the unitary group, capturing its essential structure. However, it is important to note that they are not orthonormal with respect to the Hilbert-Schmidt inner product. Therefore, we cannot directly apply Theorem~\ref{prop:projComm}. Nevertheless, we have an alternative approach that allows us to evaluate the moment operator and, consequently, compute integrals over the Haar measure. The following theorem presents a recipe for accomplishing this task.
\begin{theorem}(Computing moments)
\label{th:eq:MomPerm}
Let $O\in \mathcal{L}\left((\mathbb{C}^d)^{\otimes k}\right)$. The moment operator can then be expressed as a linear combination of permutation operators:
\begin{align}
    \ExU \left[U^{\otimes k} O U^{\dagger \otimes k}\right]=\sum_{\pi \in S_k} c_\pi(O)  V_d(\pi),
    \label{eq:MomPerm}
\end{align}
where the coefficients $c_\pi(O)$ can be determined by solving the following linear system of $k!$ equations:
\begin{align}
    \Tr\!\left(V^\dagger_d(\sigma)O\right) = \sum_{\pi \in S_k} c_\pi(O) \Tr\!\left(V^\dagger_d(\sigma)V_d(\pi)\right) \quad \text{for all $\sigma \in S_k$.}
\end{align}
This system always has at least one solution.
\end{theorem}
\begin{proof}
Equation~\eqref{eq:MomPerm} follows from Lemma~\ref{le:PropMom}.\ref{propr:CoMom}, which states that $\MomOp{O}=\ExU \left[U^{\otimes k} O U^{\dagger \otimes k}\right] \in \mathrm{Comm}(\Ug(d),k)$, and from Schur-Weyl duality (Theorem~\ref{th:SchurW}).

To obtain the linear system of equations, we begin by multiplying both sides of Eq.~\eqref{eq:MomPerm} by $V^\dagger_d(\sigma)$ and taking the trace for all $\sigma \in S_k$. This yields:
\begin{align}
    \sum_{\pi \in S_k} c_\pi(O) \Tr\!\left(V^\dagger_d(\sigma)V_d(\pi)\right)&=\Tr\!\left(V^\dagger_d(\sigma)\MomOp[k]{O}\right)\\&=\Tr\!\left(\MomOp[k]{V^\dagger_d(\sigma)O}\right)\\&=\Tr\!\left(V^\dagger_d(\sigma)O\right),
\end{align}
where in the first equality we used Eq.~\eqref{eq:MomPerm}, in the second equality we used that $V^\dagger_d(\sigma)$ commutes with $U^{\otimes k}$, and in the last equality we used the fact that the moment operator is trace preserving (Lemma~\ref{le:PropMom}.\ref{propr:MomHerm}).
Since $\MomOp{O} \in \operatorname{span} \left(V_d(\pi):\, \pi \in S_k\right)$, a solution to this linear system of equations always exists.
\end{proof}

The previous theorem provides an explicit expression for the coefficients of the moment operator in the permutation basis as $c_\pi(O)=\sum_{\sigma \in S_k} \left(\mathrm{G}^{+}\right)_{\pi,\sigma}\Tr\!\left(V_\sigma^\dagger O\right)$. Here, $\mathrm{G}$ is the Gram matrix, i.e. that matrix with coefficients $\mathrm{G}_{\pi,\sigma}=\Tr\!\left(V^\dagger_d(\pi)V_d(\sigma)\right)$, and $\mathrm{G}^+$ is its pseudo-inverse. This result allows us to express the moment operator in terms of the so-called Weingarten coefficients $\mathrm{Wg}(\pi^{-1}\sigma,d)\coloneqq \left(\mathrm{G}^+\right)_{\pi,\sigma}$ \cite{Collins_2006,collins2002moments,Weingarten:1977ya,collins2021weingarten}\footnote{The calculation of moments over the Haar measure is often referred to as the Weingarten calculus \cite{collins2021weingarten,köstenberger2021weingarten}.}, as follows:
\begin{align}
     \ExU \left[U^{\otimes k} O U^{\dagger \otimes k}\right]= \sum_{\pi,\sigma \in S_k} \mathrm{Wg}(\pi^{-1}\sigma,d)\Tr
     (V^\dagger_d(\sigma) O) V_d(\pi).
     \label{eq:momOpweing}
\end{align}
The Weingarten coefficients $\mathrm{Wg}(\pi^{-1}\sigma,d)$ can be written in terms of characters of the symmetric group. However, we will not explore this aspect here. For reference, see  \cite{Collins_2006,collins2002moments,Brand_o_2016,garcíamartín2023deep}.

The Gram matrix $\mathrm{G}$ has a simple expression in terms of the cycle structure of the permutations, given by:
\begin{align}
\mathrm{G}_{\pi,\sigma}=\Tr\!\left(V^\dagger_d(\pi)V_d(\sigma)\right)=\Tr\!\left(V_d(\pi^{-1})V_d(\sigma)\right)=\Tr\!\left(V_d(\pi^{-1}\sigma)\right)=d^{\#\text{cycles}(\pi^{-1}\sigma)},
\end{align}
where the last equality follows from the fact that $\Tr\!\left(V_d(\pi)\right)=\sum_{i_1,\dots,i_k \in [d]^k } \braket{i_{1}}{i_{\pi^{-1}(1)}}\cdots \braket{i_{k}}{i_{\pi^{-1}(k)}}$ and observing that this sum has $d^{\# \text{cycles}(\pi)}$ nonvanishing terms and they are all equal to $1$.  
This fact is also evident in the tensor networks notation that we will introduce in section~\ref{sec:diagram}.
\begin{proposition}
\label{prop:permInd}
For $\pi \in S_k$, the permutation matrices $V_d(\pi)$ are linearly independent if $k\le d$, but linearly dependent if $k>d$.
\end{proposition}

\begin{proof}
    Let us consider first the case where $k\le d$. We assume that there exist complex coefficients $\alpha_\pi \in \mathbb{C}$ for all $\pi\in S_k$ such that the following equation holds:
    \begin{align}
        \sum_{\pi \in S_k} \alpha_\pi V_d(\pi)=0.
    \end{align}
    Since, $k\le d$, we can choose $k$ distinct elements $i_1,\dots,i_k \in [d]$ and apply to the right of both sides of the previous equation the state $\ket{i_1}\otimes\cdots\otimes\ket{i_k}$, obtaining:
    \begin{align}
        \sum_{\pi \in S_k} \alpha_\pi \ket{i_{\pi^{-1}(1)}}\otimes\cdots\otimes\ket{i_{\pi^{-1}(k)}}=0.
    \end{align}
    For all $\sigma \in S_k$, we multiply to the left by $\bra{i_{\sigma^{-1}(1)}}\otimes\cdots\otimes\bra{i_{\sigma^{-1}(k)}}$ and deduce $\alpha_{\sigma}=0$. Thus, we have shown linear independence.
    If $k>d$, then consider the operator:
    \begin{align}
        A=\sum_{\pi\in S_k} \mathrm{sgn}(\pi)V_d(\pi).
    \end{align}
    Now we show this linear combination is the zero operator, proving linear dependence.

We consider the action of $A$ on an arbitrary product basis state $\ket{i_1}\otimes\cdots\otimes\ket{i_k}$. Since $k>d$, at least two tensor factors must have matching entry, i.e., there exist $l\neq m \in [k]$ such that $i_l = i_m$. Due to anti-symmetrization, the output vector has to be the zero vector. In fact, we have:
\begin{align}
    A \ket{i_1}\otimes\cdots\otimes\ket{i_k}&=A V_d(\tau_{l,m})\ket{i_1}\otimes\cdots\otimes\ket{i_k}\\
    &=\sum_{\pi \in S_k} \mathrm{sgn}( \pi )V_d(
 \pi)V_d(\tau_{l,m})\ket{i_1}\otimes\cdots\otimes\ket{i_k}\\
 &=\sum_{\pi \in S_k} \mathrm{sgn}( \pi )V_d(
 \pi\tau_{l,m})\ket{i_1}\otimes\cdots\otimes\ket{i_k}\\
 &=\sum_{\pi \in S_k} \mathrm{sgn}( \pi\tau^{-1}_{l,m} )V_d(\pi)\ket{i_1}\otimes\cdots\otimes\ket{i_k}\\
 &=\mathrm{sgn}\left(\tau^{-1}_{l,m}\right)\sum_{\pi \in S_k} \mathrm{sgn}(\pi )V_d(
 \pi)\ket{i_1}\otimes\cdots\otimes\ket{i_k}\\
 &=- A\ket{i_1}\otimes\cdots\otimes\ket{i_k}.
\end{align}
Therefore $A \ket{i_1}\otimes\cdots\otimes\ket{i_k}=0$ and so $A=0$.
\end{proof}

In the following we define the identity permutation operator $\mathbb{I}$ and the Flip operator $\mathbb{F}$ (also often referred to as  $\mathrm{SWAP}$ operator) which are the permutation operators corresponding to the elements of the permutation group $S_2$. 
\begin{definition}[Identity and Flip operators] The identity permutation operator $\mathbb{I}$ is defined as the linear operator that leaves any tensor product state ${\ket{\psi}\otimes\ket{\phi}}$ unchanged, that is: 
    \begin{align}
\mathbb{I}\left(\ket{\psi}\otimes\ket{\phi}\right)=\ket{\psi}\otimes\ket{\phi}, \quad \quad  \text{for all $\ket{\psi},\ket{\phi} \in \mathbb{C}^d$}.
    \end{align} 
The Flip operator $\mathbb{F}$ is defined as the linear operator that interchanges the order of the tensor factors of any product state $\ket{\psi}\otimes\ket{\phi}$, that is:
\begin{align}
\mathbb{F}\left(\ket{\psi}\otimes\ket{\phi}\right)=\ket{\phi}\otimes\ket{\psi}, \quad \quad  \text{for all $\ket{\psi},\ket{\phi} \in \mathbb{C}^d$}.
\end{align}
\end{definition}
Writing the Identity and the Flip in the computational basis, we have:
\begin{align}
    \Idd=\sum^d_{i,j=1} \ketbra{i,j}{i,j}, \quad\quad\quad \Flip =\sum^d_{i,j=1} \ketbra{i,j}{j,i}.
\end{align}
From this, it is evident that the Flip operator is Hermitian. Another key property of the Flip operator is the \emph{swap-trick}, which states that for all operators $A,B \in \mathcal{L}(\mathbb{C}^d)$, we have:
\begin{align}
\Tr\!\left(\left(A\otimes B\right) \Flip \right)=\Tr\!\left(AB\right).
\end{align}

This property can be easily verified using the definition of the Flip operator. The \emph{swap-trick} is particularly useful as it allows us to simplify calculations involving tensor products and permutations, and it will be used extensively in the subsequent sections.

We now present a corollary of Theorem~\ref{th:eq:MomPerm}, which plays a crucial role in many calculations involving Haar integrals in the context of quantum information, as we will see in section~\ref{sec:applications}. In the following corollary we assume $d>1$ (which makes sense for an $n$-qubit system, since $d=2^n> 1$).
\begin{corollary}[First and second moment]\
\label{ex:SecondMoment}
Given $O \in \MatC{d}$, we have:
\begin{align}
            \ExU \left[U O U^{\dagger}\right]&=\frac{\Tr\!\left(O\right)}{d} I.
            \label{eq:1momHaar}
\end{align}
Given $O \in \mathcal{L}((\mathbb{C}^d)^{\otimes 2})$, we have:
    \begin{align}
        \ExU \left[U^{\otimes 2} O U^{\dagger \otimes 2}\right]=c_{\Idd,O}\mathbb{I} + c_{\Flip,O} \mathbb{F},
        \label{eq:2momHaar}
    \end{align}
    where:
    \begin{align}
        c_{\Idd,O}=\frac{\Tr\!\left(O\right)-d^{-1}\Tr\!\left(\mathbb{F}O\right)}{d^2-1}\quad \text{and} \quad c_{\Flip,O}=\frac{\Tr\!\left(\mathbb{F}O\right)-d^{-1}\Tr\!\left(O\right)}{d^2-1}.
    \end{align}
\end{corollary}
\begin{proof}
    According to Theorem~\eqref{th:eq:MomPerm}, we have that the first-order moment operator is proportional to the identity operator (the only permutation operator associated with the permutation group of one element $S_1$), i.e. 
    \begin{align}
        \ExU \left[U O U^{\dagger}\right]=c_I I,
    \end{align}
    with $c_I\in\mathbb{C}$. Taking the trace to both sides, we deduce that $c_I=\frac{\Tr(O)}{d}$. 

    Moreover, using Theorem~\eqref{th:eq:MomPerm}, we can see that that the second-order moment operator is a linear combination of the two permutation operators associated with the permutation group $S_2$ which are the identity $\Idd$ and the flip $\Flip$ operator:
    \begin{align}
        \ExU \left[U^{\otimes 2} O U^{\dagger \otimes 2}\right]=c_{\mathbb{I},O} \mathbb{I} + c_{\mathbb{F},O} \mathbb{F},
    \end{align}
    with $c_{\mathbb{I},O}, c_{\mathbb{F},O} \in \mathbb{C}$. To find these numbers, we left-multiply both sides by $\mathbb{I}$ and take the trace, which gives us:
    \begin{align}
        \Tr\!\left(O\right)=c_{\mathbb{I},O} d^2 + c_{\mathbb{F},O} d,
    \end{align}
    where we used the fact that $\Tr\!\left(\mathbb{I}\right)=d^2$ and $\Tr\!\left(\mathbb{F}\right)=d$.
    Similarly, by left-multiplying by $\mathbb{F}$ and taking the trace, we obtain the linear system of equation:
\begin{align}
    \systeme*{\Tr\!\left(O\right)=c_{\mathbb{I},O} d^2 + c_{\mathbb{F},O} d, \Tr\!\left(\mathbb{F}O\right)=c_{\mathbb{I},O} d + c_{\mathbb{F},O} d^2}.
\end{align}
Solving this system, we obtain $c_{\mathbb{I},O}=\frac{\Tr\!\left(O\right)-d^{-1}\Tr\!\left(O\mathbb{F}\right)}{d^2-1}$ and $c_{\mathbb{F},O}=\frac{\Tr\!\left(\mathbb{F}O\right)-d^{-1}\Tr\!\left(O\right)}{d^2-1}$.
    
\end{proof}
We have shown how to obtain the formulas for the first and second moment operator by utilizing the property that the moment operator lies in the commutant (Lemma~\ref{le:PropMom}.\ref{propr:CoMom}) and that the commutant is the span of the permutation matrices $\mathrm{Comm}(\Ug(d),k)=\operatorname{span} \Big(V_d(\pi):\, \pi \in S_k\Big)$ for $k=1,2$, which are particular instances of Schur-Weyl duality. Now, we will provide explicit proofs for these cases.

\begin{example}
\label{ex:comm2}
    \begin{align}
        &\mathrm{Comm}(\Ug(d),k=1)=\operatorname{span} \Big(I\Big),\\
        &\mathrm{Comm}(\Ug(d),k=2)=\operatorname{span} \Big(\Idd,\Flip\Big).
    \end{align}
\end{example}
\begin{proof}
Consider an element $A\in \mathrm{Comm}(\Ug(d),k=1)$. In the canonical basis decomposition, $A$ can be written as:
\begin{align}
A=\sum^{d}_{i,j=1} A_{i,j}\ketbra{i}{j}=\sum^{d}_{i=1} A_{i,i}\ketbra{i}{i}+ \sum^{d}_{\substack{i,k= 1\\j\neq i}} A_{i,j}\ketbra{i}{j},
\end{align}
where $A_{i,j}:= \bra{i}A\ket{j}$.
Since $A$ is in the commutant, then $UAU^\dagger=A$ holds for all unitary matrices $U$. By choosing a unitary $U$ that sends vector $\ket{k}$ to its negative $\ket{k}\rightarrow -\ket{k}$ while leaving the other basis vectors unchanged, it follows that, for all $k\neq l \in [d]$, the term $A_{k,l}\ketbra{k}{l}$ is equal to its negative, and hence, $A_{k,l}=0$. 
Therefore, we have that if $A$ is in the commutant, then $A=\sum^{d}_{i=1} A_{i,i}\ketbra{i}{i}$. 
Now, for all $k\neq l \in [d]$, by choosing a unitary matrix $U$ that exchanges the vectors $\ket{k}$ and $\ket{l}$ while leaving the other basis states unchanged, we derive that all diagonal elements $A_{k,k}$ should be equal each other for all $k\in [d]$.
Thus, we have $
A=c \sum^{d}_{i=1} \ketbra{i}{i} =c \Id$,
where $c:=A_{1,1}$. Therefore an operator can only be in the commutant if and only if it is proportional to the Identity $I$. Thus, we have shown the first formula we aimed to prove.

Consider now $Q\in \mathrm{Comm}(\Ug(d),k=2)$. $Q$ can be decomposed as:
    \begin{align}
        Q=\sum^d_{i,j,k,l=1}Q_{i,j;k,l}\ketbra{i,j}{k,l},
    \end{align}
where $Q_{i,j;k,l}:=\bra{i,j}Q\ket{k,l}$.
Since $Q$ is in the commutant, then $U^{\otimes 2}QU^{\dagger \otimes 2}=Q$ holds for all unitary matrices $U$. For all sets of indices $\{i,j\} \neq \{k,l\}$, let $m$ be a index such that $m \in \{i,j\}$ but $m \notin \{k,l\}$ or vice versa (there is at least one). By choosing unitaries that send a basis vector to its negative, while leaving unchanged the other basis vectors, it follows that, for all $\{i,j\} \neq \{k,l\} \in [d]$, the terms $Q_{i,j;k,l}\ketbra{i,j}{k,l}$ are equal to their negatives, and hence, $Q_{i,j;k,l}=0$. 
Therefore, if $Q$ is in the commutant, it must have the following form:
    \begin{align}
         Q=\sum^d_{i=1}Q_{i,i;i,i}\ketbra{i,i}{i,i}+\sum^d_{\substack{i,j=1\\i\neq j}}Q_{i,j;i,j}\ketbra{i,j}{i,j}+\sum^d_{\substack{i,j=1\\i\neq j}}Q_{i,j;j,i}\ketbra{i,j}{j,i}.
    \end{align}
Again, by choosing a unitary matrix $U$ that exchanges the vectors $\ket{i,j}$ with a different vector $\ket{k,l}$ while leaving the other basis states unchanged, and using $U^{\otimes 2}QU^{\dagger \otimes 2}=Q$, we derive that:
    \begin{align}
         Q&=a\sum^d_{i=1}\ketbra{i,i}{i,i}+b\sum^d_{\substack{i,j=1\\i\neq j}}\ketbra{i,j}{i,j}+c\sum^d_{\substack{i,j=1\\i\neq j}}\ketbra{i,j}{j,i}
    \end{align}
where $a:=Q_{1,1;1,1}$, $b:=Q_{1,2;1,2}$ and $c:=Q_{1,2;2,1}$.
Note that we can also write $Q$ as:
\begin{align}
    Q=(a-b-c)\sum^d_{i=1}\ketbra{i,i}{i,i}+b\,\Idd+c\,\Flip.
\end{align}
Since $\Idd$ and $\Flip$ commute with $U^{\otimes 2}$ for all $U\in \mathrm{U}(d)$, then $Q$ is in the commutant if and only if $(a-b-c)\sum^d_{i=1}\ketbra{i,i}{i,i}$ is also in the commutant.

We now show that $\sum^d_{i=1}\ketbra{i,i}{i,i}$ is not in the commutant, which implies that $Q$ is in the commutant if and only if $(a-b-c)=0$, i.e., if and only if $Q$ is a linear combination of only $\Idd$ and $\Flip$. To see this, suppose that $\sum^d_{i=1}\ketbra{i,i}{i,i}$ is in the commutant. Then, it should commute with $U^{\otimes2}$ for all unitaries $U$. We can choose the (Discrete Fourier Transform) unitary $U$ such that $U\ket{k}=\frac{1}{\sqrt{d}}\sum^d_{j=1} \omega^{(j-1)(k-1)}\ket{j}$, where $\omega:=\exp(i\frac{2\pi }{d})$, and obtain the following relation:
\begin{align}
\sum^d_{i=1}\ketbra{i,i}{i,i}=\sum^d_{i=1}U^{\otimes 2}\ketbra{i,i}{i,i}U^{\otimes 2\dagger}.
\end{align}
However, computing the expectation value on both sides of the previous equation over the state $\ket{1,1}$ leads to $1=\frac{1}{d}$, which is false for every $d>1$ (for $d=1$, the relation to prove immediately follows). Therefore, we conclude that $\sum^d_{i=1}\ketbra{i,i}{i,i}$ is not in the commutant, and hence $Q$ is in the commutant if and only if $Q$ is a linear combination of $\Idd$ and $\Flip$.
\end{proof}
\section{Symmetric subspace}
\label{sec:sym}
In this section, we introduce the symmetric subspace, which plays a crucial role when analyzing Haar random states. For a more in-depth analysis, see \cite{harrow2013church}. The symmetric subspace can be defined as the set of states $\ket{\psi}$ in $(\mathbb{C}^d)^{\otimes k}$ that are invariant under permutations of their constituent subsystems. Formally, we define the symmetric subspace as follows:
\begin{definition}[Symmetric subspace]
    \begin{align}
        \mathrm{Sym}_k(\mathbb{C}^d)\coloneqq \Big\{\ket{\psi} \in (\mathbb{C}^d)^{\otimes k}\colon V_d(\pi)\ket{\psi}=\ket{\psi} \,\, \forall \, \pi \in S_k\Big\}.
    \end{align}
\end{definition}
To facilitate our analysis, we also define the operator $P^{(d,k)}_{\mathrm{sym}}$ as follows: 
\begin{align}
    P^{(d,k)}_{\mathrm{sym}}\coloneqq \frac{1}{k!}\sum_{\pi\in S_k} V_d(\pi).
\end{align}
\begin{theorem}[Projector on $\mathrm{Sym}_k(\mathbb{C}^d)$]
\label{th:Psymprojector}
    ${P^{(d,k)}_{\mathrm{sym}}}$ is the orthogonal projector on the symmetric subspace ${\mathrm{Sym}_k(\mathbb{C}^d)}$.
\end{theorem}
\begin{proof}
We start by observing that $V_d(\pi)P^{(d,k)}_{\mathrm{sym}}=P^{(d,k)}_{\mathrm{sym}}$:
\begin{align}
    V_d(\pi)P^{(d,k)}_{\mathrm{sym}}=\frac{1}{k!}\sum_{\sigma\in S_k}  V_d(\pi)V_d( \sigma)=\frac{1}{k!}\sum_{\sigma\in S_k}  V_d( \pi \sigma )=\frac{1}{k!}\sum_{\sigma\in S_k}  V_d( \sigma)=P^{(d,k)}_{\mathrm{sym}}.
\end{align}
Using this, we can show that $P^{(d,k)2}_{\mathrm{sym}}=P^{(d,k)}_{\mathrm{sym}}$:
    \begin{align}
        P^{(d,k)2}_{\mathrm{sym}}=\frac{1}{k!}\sum_{\pi\in S_k} V_d(\pi) P^{(d,k)}_{\mathrm{sym}}=\frac{1}{k!}\sum_{\pi\in S_k} P^{(d,k)}_{\mathrm{sym}}=P^{(d,k)}_{\mathrm{sym}}.
    \end{align}
Furthermore, we have $P^{(d,k)\dagger}_{\mathrm{sym}}=P^{(d,k)}_{\mathrm{sym}}$:
\begin{align}
    P^{(d,k)\dagger}_{\mathrm{sym}}=\frac{1}{k!}\sum_{\pi\in S_k} V^{\dagger}_d(\pi)=\frac{1}{k!}\sum_{\pi\in S_k} V_d(\pi^{-1})=\frac{1}{k!}\sum_{\pi\in S_k} V_d(\pi)=P^{(d,k)}_{\mathrm{sym}}.
\end{align}
Therefore, $P^{(d,k)}_{\mathrm{sym}}$ is an orthogonal projector.
Moreover $\mathrm{Im}\left(P^{(d,k)}_{\mathrm{sym}}\right)\subseteq \mathrm{Sym}_k(\mathbb{C}^d)$: in fact $P^{(d,k)}_{\mathrm{sym}}\ket{\psi} \in \mathrm{Sym}_k(\mathbb{C}^d)$ for all $\ket{\psi} \in (\mathbb{C}^d)^{\otimes k}$, since  $ V_d(\pi)P^{(d,k)}_{\mathrm{sym}}\ket{\psi}=P^{(d,k)}_{\mathrm{sym}}\ket{\psi}$ for all $\pi\in S_k$, where we used again that $V_d(\pi)P^{(d,k)}_{\mathrm{sym}}=P^{(d,k)}_{\mathrm{sym}}$. Moreover if $\ket{\psi} \in \mathrm{Sym}_k(\mathbb{C}^d)$, then \begin{align}
   P^{(d,k)}_{\mathrm{sym}}\ket{\psi}=\frac{1}{k!}\sum_{\pi\in S_k} V_d(\pi)\ket{\psi}=\frac{1}{k!}\sum_{\pi\in S_k}\ket{\psi}=\ket{\psi}.
\end{align}
Therefore we also have $\mathrm{Sym}_k(\mathbb{C}^d) \subseteq \mathrm{Im}\left(P^{(d,k)}_{\mathrm{sym}}\right)$ and so $\mathrm{Sym}_k(\mathbb{C}^d) = \mathrm{Im}\left(P^{(d,k)}_{\mathrm{sym}}\right)$.
\end{proof}
We now turn to computing the dimension of the symmetric subspace.
\begin{theorem}[Dimension of the symmetric subspace]
\label{th:dimsym}
  \begin{align}
\Tr\!\left(P^{(d,k)}_{\mathrm{sym}}\right)=\mathrm{dim}\left(\mathrm{Sym}_k(\mathbb{C}^d)\right)= 
    \binom{k+d-1}{k}.
  \end{align}
\end{theorem}
\begin{proof}
The first equality follows by the fact that $P^{(d,k)}_{\mathrm{sym}}$ is the orthogonal projector onto $\mathrm{Sym}_k(\mathbb{C}^d)$.
Now we observe that:
\begin{align}
    \mathrm{Sym}_k(\mathbb{C}^d)= \mathrm{Im}\left(P^{(d,k)}_{\mathrm{sym}}\right)=\operatorname{span}\left(P^{(d,k)}_{\mathrm{sym}}\ket{i_1}\otimes\cdots\otimes\ket{i_k}: \,\forall \,\, i_1,\dots,i_k \in [d] \right).
\end{align}

To count the number of vectors in this set that are linearly independent, we make the following observation. Suppose $\ket{i_1}\otimes\cdots\otimes\ket{i_k}$ and $\ket{j_1}\otimes\cdots\otimes\ket{j_k}$ are two computational basis vectors with $i_1,j_1,\dots,i_k,j_k\in[d]$ such that the sets $\{i_1,\dots,i_k\}$ and $\{j_1,\dots,j_k\}$ have the same number $n_m$ of elements that are equal to $m$ for every $m\in[d]$. Then there exists a unitary operator $V_d(\pi)$ such that $V_d(\pi)\ket{i_1}\otimes\cdots\otimes\ket{i_k}=\ket{j_1}\otimes\cdots\otimes\ket{j_k}$. Furthermore, since $P^{(d,k)}_{\mathrm{sym}}V_d(\pi)=P^{(d,k)}_{\mathrm{sym}}$, it follows that $P^{(d,k)}_{\mathrm{sym}}\ket{i_1}\otimes\cdots\otimes\ket{i_k}=P^{(d,k)}_{\mathrm{sym}}\ket{j_1}\otimes\cdots\otimes\ket{j_k}$.

Therefore, we can define the set of vectors ${\ket{n_1,\dots,n_d}}$, where $n_1,\dots,n_d\in[k]$ with $n_1+\cdots+n_d=k$, as follows:
\begin{align}
\ket{n_1,\dots,n_d}\coloneqq P^{(d,k)}_{\mathrm{sym}}\ket{i_1}\otimes \cdots \otimes \ket{i_k},
\end{align}
where $\ket{i_1}\otimes \cdots \otimes \ket{i_k}$ is any computational basis vector such that, for each $m\in[d]$, there are $n_m$ distinct $j\in[k]$ such that $i_j=m$. It is easy to see that these vectors are orthogonal, and hence linearly independent.
Therefore, the dimension of the symmetric subspace is equal to the number of linearly independent vectors in the set ${\ket{n_1,\dots,n_d}}$ defined earlier. This is equal to the number of ways of deciding how to assign $k$ indices between $d$ possible labels, which is $\binom{k+d-1}{k}$.
\end{proof}

Next, we introduce the anti-symmetric subspace. 
\begin{definition}[Anti-symmetric subspace]
   The anti-symmetric subspace is the set:
    \begin{align}
        \mathrm{ASym}_k(\mathbb{C}^d)\coloneqq \Big\{\ket{\psi} \in (\mathbb{C}^d)^{\otimes k}\colon V_d(\pi)\ket{\psi}=\mathrm{sgn}\left(\pi\right)\ket{\psi} \,\, \forall \, \pi \in S_k\Big\}, 
    \end{align}
where $\mathrm{sgn}\left(\sigma\right)$ denotes the sign of a permutation $\sigma \in S_k$.
\end{definition}
Similarly as before, we can define the operator: 
\begin{align}
    P^{(d,k)}_{\mathrm{asym}}\coloneqq \frac{1}{k!}\sum_{\pi\in S_k} \mathrm{sgn}\left(\pi\right) V_d(\pi).
\end{align}
and prove the following theorem:
\begin{theorem}
    The operator $P^{(d,k)}_{\mathrm{asym}}$ is the orthogonal projector on the anti-symmetric subspace $\mathrm{ASym}_k(\mathbb{C}^d)$.
    \label{th:Pasym}
\end{theorem}
\begin{proof}
    We have: 
    \begin{align}
        V_d(\sigma) P^{(d,k)}_{\mathrm{asym}}&=\frac{1}{k!}\sum_{\pi\in S_k} \mathrm{sgn}\left(\pi\right) V_d(\sigma) V_d(\pi)=\frac{1}{k!}\sum_{\pi\in S_k} \mathrm{sgn}\left(\pi\right) V_d( \sigma\pi)\\&=\frac{1}{k!}\sum_{\pi\in S_k} \mathrm{sgn}\left(\sigma^{-1}\pi \right) V_d(\pi )=\mathrm{sgn}\left(\sigma^{-1}\right)P^{(d,k)}_{\mathrm{asym}}\\&=\mathrm{sgn}\left(\sigma\right)P^{(d,k)}_{\mathrm{asym}}.
    \end{align}
Similarly, $P^{(d,k)}_{\mathrm{asym}} V_d(\sigma) = \mathrm{sgn}\left(\sigma\right)P^{(d,k)}_{\mathrm{asym}}$.
Using this, we can show that $P^{(d,k)2}_{\mathrm{asym}}=P^{(d,k)}_{\mathrm{asym}}$:
\begin{align}
    P^2_{\mathrm{asym}}=\frac{1}{k!}\sum_{\pi\in S_k} \mathrm{sgn}\left(\pi\right) V_d(\pi)P^{(d,k)}_{\mathrm{asym}}=\frac{1}{k!}\sum_{\pi\in S_k} \left(\mathrm{sgn}\left(\pi\right)\right)^2 P^{(d,k)}_{\mathrm{asym}}=P^{(d,k)}_{\mathrm{asym}}.
\end{align}
We also have $P^{(d,k)\dagger}_{\mathrm{asym}}=P^{(d,k)}_{\mathrm{asym}}$:
\begin{align}
    P^{(d,k)\dagger}_{\mathrm{asym}}&=\frac{1}{k!}\sum_{\pi\in S_k} \mathrm{sgn}\left(\pi\right) V^{\dagger}_d(\pi)=\frac{1}{k!}\sum_{\pi\in S_k} \mathrm{sgn}\left(\pi\right) V_d(\pi^{-1})=\frac{1}{k!}\sum_{\pi\in S_k} \mathrm{sgn}\left(\pi^{-1}\right) V_d(\pi^{-1})\\&=\frac{1}{k!}\sum_{\pi\in S_k} \mathrm{sgn}\left(\pi\right) V_d(\pi)=P^{(d,k)}_{\mathrm{asym}}.
\end{align}
We can show that $\mathrm{Im}\left(P^{(d,k)}_{\mathrm{asym}}\right)\subseteq \mathrm{ASym}_k(\mathbb{C}^d)$ as follows:
for all $\ket{\psi} \in (\mathbb{C}^d)^{\otimes k}$, we have $P^{(d,k)}_{\mathrm{asym}}\ket{\psi} \in \mathrm{ASym}_k(\mathbb{C}^d)$, since $V_d(\pi)P^{(d,k)}_{\mathrm{asym}}\ket{\psi}=\mathrm{sgn}\left(\pi\right)P^{(d,k)}_{\mathrm{asym}}\ket{\psi}$ for all $\pi\in S_k$, where we used again that $V_d(\pi)P^{(d,k)}_{\mathrm{asym}}=\mathrm{sgn}\left(\pi\right)P^{(d,k)}_{\mathrm{asym}}$. Moreover, we can also show that $\mathrm{ASym}_k(\mathbb{C}^d) \subseteq \mathrm{Im}\left(P^{(d,k)}_{\mathrm{asym}}\right)$. In fact, if $\ket{\psi} \in \mathrm{ASym}_k(\mathbb{C}^d)$, then we have:

\begin{align}
   P^{(d,k)}_{\mathrm{asym}}\ket{\psi}=\frac{1}{k!}\sum_{\pi\in S_k} \mathrm{sgn}\left(\pi\right)V_d(\pi)\ket{\psi}=\frac{1}{k!}\sum_{\pi\in S_k}\left(\mathrm{sgn}\left(\pi\right)\right)^2\ket{\psi}=\ket{\psi}.
\end{align}

Therefore, we have shown that $\mathrm{Im}\left(P^{(d,k)}_{\mathrm{asym}}\right)= \mathrm{ASym}_k(\mathbb{C}^d)$, which implies that $P^{(d,k)}_{\mathrm{asym}}$ is the orthogonal projector on the anti-symmetric subspace $\mathrm{ASym}_k(\mathbb{C}^d)$.
\end{proof}
Next, we compute the dimension of the anti-symmetric subspace.
\begin{proposition}[Dimension of the anti-symmetric subspace] If $d\ge k$, we have:
\begin{align}
\Tr\!\left(P^{(d,k)}_{\mathrm{asym}}\right)=\mathrm{dim}\left(\mathrm{ASym}_k(\mathbb{C}^d)\right)= 
    \binom{d}{k},
  \end{align}
  otherwise $\Tr\!\left(P^{(d,k)}_{\mathrm{asym}}\right)=0$.

\end{proposition}
\begin{proof}
First, we have:
\begin{align}
    \mathrm{ASym}_k(\mathbb{C}^d)= \mathrm{Im}\left(P^{(d,k)}_{\mathrm{asym}}\right)=\operatorname{span}\left(P^{(d,k)}_{\mathrm{asym}}\ket{i_1}\otimes\cdots\otimes\ket{i_k}: \,\forall \,\, i_1,\dots,i_k \in [d] \right).
\end{align}

To count the number of linearly independent vectors $P^{(d,k)}_{\mathrm{asym}}\ket{i_1}\otimes\cdots\otimes\ket{i_k}$, we observe that if there are at least two tensor factors of $\ket{i_1}\otimes\cdots\otimes\ket{i_k}$ with matching entry i.e. there exist $l\neq m \in [k]$ such that $i_l = i_m$, then: 
\begin{align}
    P^{(d,k)}_{\mathrm{asym}}\ket{i_1}\otimes\cdots\otimes\ket{i_k}=0.
\end{align}
This because: 
\begin{align}
P^{(d,k)}_{\mathrm{asym}}\ket{i_1}\otimes\cdots\otimes\ket{i_k}&=P^{(d,k)}_{\mathrm{asym}}V_d(\tau_{l,m})\ket{i_1}\otimes\cdots\otimes\ket{i_k}\\&=\mathrm{sgn}\left(\tau_{l,m}\right)P^{(d,k)}_{\mathrm{asym}}\ket{i_1}\otimes\cdots\otimes\ket{i_k}\\&=-P^{(d,k)}_{\mathrm{asym}}\ket{i_1}\otimes\cdots\otimes\ket{i_k},
\end{align}
where in the second equality we used $P^{(d,k)}_{\mathrm{asym}}V_d(\tau_{l,m})=\mathrm{sgn}\left(\tau_{l,m}\right)P^{(d,k)}_{\mathrm{asym}}$ (as shown in the proof of Theorem~\ref{th:Pasym}).
Therefore, $i_1,\dots,i_k$ must all be distinct for $P^{(d,k)}_{\mathrm{asym}}\ket{i_1}\otimes\cdots\otimes\ket{i_k}$ to be nonzero. This also implies that if $d<k$, then $\Tr\!\left(P^{(d,k)}_{\mathrm{asym}}\right)=0$. Therefore, we now focus on the case $d\ge k$.
If for all $m\in [d]$, there exist two vectors $\ket{i_1}\otimes\cdots\otimes\ket{i_k}$ and $\ket{j_1}\otimes\cdots\otimes\ket{j_k}$ with $i_1,j_1\dots,i_k,j_k \in [d]$ such that both sets, $\{i_1,\dots,i_k\}$ and $\{j_1,\dots,j_k\}$, contain $n_m\in\{0,1\}$ elements equal to $m$, then there exists a permutation $V_d(\pi)$ such that $V_d(\pi)\ket{i_1}\otimes\cdots\otimes\ket{i_k}=\ket{j_1}\otimes\cdots\otimes\ket{j_k}$. 
Since $P^{(d,k)}_{\mathrm{asym}}V_d(\pi)=\mathrm{sgn}\left(\pi\right)P^{(d,k)}_{\mathrm{asym}}$, we have $P^{(d,k)}_{\mathrm{asym}}\ket{i_1}\otimes\cdots\otimes\ket{i_k}=\mathrm{sgn}\left(\pi\right)P^{(d,k)}_{\mathrm{asym}}\ket{j_1}\otimes\cdots\otimes\ket{j_k}$, and hence, $P^{(d,k)}_{\mathrm{asym}}\ket{i_1}\otimes\cdots\otimes\ket{i_k}$ and $P^{(d,k)}_{\mathrm{asym}}\ket{j_1}\otimes\cdots\otimes\ket{j_k}$ are linearly dependent.

Thus, taking $n_1,\dots,n_d \in \{0,1\}$ with $n_1+\cdots+n_d=k$, we define $\ket{n_1,\dots,n_d}$ as $P^{(d,k)}_{\mathrm{asym}}\ket{i_1}\otimes \cdots \otimes \ket{i_k}$, where $\ket{i_1}\otimes \cdots \otimes \ket{i_k}$ is any computational basis vector such that, for each $m\in [d]$, there exist $n_m$ elements $j\in [k]$ such that $i_j=m$. It is easy to see that such vectors are orthogonal and hence independent. Therefore, the dimension is given by the number of such independent vectors, which is equal to the number of ways to choose an unordered subset of $k$ elements from a set of $d$ elements, which is $\binom{d}{k}$.
\end{proof}    
We will now present a proposition that establishes a relationship between the symmetric and anti-symmetric subspaces.
 \begin{proposition}
     We have $P^{(d,k)\dagger}_{\mathrm{asym}}P^{(d,k)}_{\mathrm{sym}}=0$. In particular $P^{(d,k)}_{\mathrm{asym}}$ and $P^{(d,k)}_{\mathrm{sym}}$ are orthogonal with respect to the Hilbert-Schmidt inner product.
 \end{proposition}
\begin{proof}
We have:
\begin{align}
    P^{(d,k)\dagger}_{\mathrm{asym}}P^{(d,k)}_{\mathrm{sym}}=P^{(d,k)}_{\mathrm{asym}}P^{(d,k)}_{\mathrm{sym}}=\frac{1}{k!}\sum_{\pi\in S_k} \mathrm{sgn}\left(\pi\right)V_d(\pi)P^{(d,k)}_{\mathrm{sym}}=\left(\frac{1}{k!}\sum_{\pi\in S_k} \mathrm{sgn}\left(\pi\right)\right) P^{(d,k)}_{\mathrm{sym}}=0,
\end{align}
where we used that $V_d(\pi)P^{(d,k)}_{\mathrm{sym}}=P^{(d,k)}_{\mathrm{sym}}$ and that $\frac{1}{k!}\sum_{\pi\in S_k} \mathrm{sgn}\left(\pi\right)=0$ since: \begin{align}
\frac{1}{k!}\sum_{\pi\in S_k} \mathrm{sgn}\left(\pi\right)=\frac{1}{k!}\sum_{\pi\in S_k} \mathrm{sgn}\left(\tau\pi\right)=\mathrm{sgn}\left(\tau\right)\frac{1}{k!}\sum_{\pi\in S_k} \mathrm{sgn}\left(\pi\right) =-\frac{1}{k!}\sum_{\pi\in S_k} \mathrm{sgn}\left(\pi\right),
\end{align}
where $\tau$ is any permutation.
\end{proof}

Note that, for $k=2$ (and $d> 1$), we have that, $\mathrm{dim}\left(\mathrm{Sym}_2(\mathbb{C}^d)\right) +  \mathrm{dim}\left(\mathrm{ASym}_2(\mathbb{C}^d)\right)=\binom{d+1}{2}+\binom{d}{2}=d^2$, this implies that $
    (\mathbb{C}^d)^{\otimes 2} = \mathrm{Sym}_2(\mathbb{C}^d) \oplus \mathrm{ASym}_2(\mathbb{C}^d)$.

Since $P^{(d,2)}_{\mathrm{sym}}$ and $P^{(d,2)}_{\mathrm{asym}}$ are both linear combinations of permutations, they commute with $U^{\otimes 2}$ for all unitaries $U$. Consequently, they are elements of a basis for the commutant $\mathrm{Comm}(\Ug(d),k=2)$ and they are orthogonal to each other. Furthermore, the commutant can have at most dimension two, since it is spanned by the identity and Flip operators (Example \ref{ex:comm2}). Thus, $P^{(d,k)}_{\mathrm{sym}}$ and $P^{(d,k)}_{\mathrm{asym}}$ form an orthogonal basis for the commutant.

Since the moment operator is the orthogonal projector onto the commutant (Theorem~\ref{prop:projComm}), we can derive the formula for the second-order moment operator (already derived in Corollary~\ref{ex:SecondMoment}), using $P^{(d,k)}_{\mathrm{sym}}$ and $P^{(d,k)}_{\mathrm{asym}}$ as the basis. We express the moment operator as:
    \begin{align}
        \ExU \left[U^{\otimes 2} O U^{\dagger \otimes 2}\right]&= \Big\langle\frac{P^{(d,2)}_{\mathrm{sym}}}{\norm{P^{(d,2)}_{\mathrm{sym}}}_{2}},O\Big\rangle_{\mathrm{HS}}\frac{P^{(d,2)}_{\mathrm{sym}}}{\norm{P^{(d,2)}_{\mathrm{sym}}}_{2}} + \Big\langle\frac{P^{(d,2)}_{\mathrm{asym}}}{\norm{P^{(d,2)}_{\mathrm{asym}}}_{2}},O\Big\rangle_{HS}\frac{P^{(d,2)}_{\mathrm{asym}}}{\norm{P^{(d,2)}_{\mathrm{asym}}}_{2}}\\
        &= \frac{\Tr\!\left(P^{(d,2)}_{\mathrm{sym}}O\right)}{\Tr\!\left(P^{(d,2)}_{\mathrm{sym}}\right)}P^{(d,2)}_{\mathrm{sym}} + \frac{\Tr\!\left(P^{(d,2)}_{\mathrm{asym}}O\right)}{\Tr\!\left(P^{(d,2)}_{\mathrm{asym}}\right)}P^{(d,2)}_{\mathrm{asym}}\\
        &=\frac{\Tr\!\left(O\right) + \Tr\!\left(\mathbb{F} O\right)}{d(d+1)}\left(\frac{\mathbb{I}+\mathbb{F}}{2}\right)+\frac{\Tr\!\left(O\right) - \Tr\!\left(\mathbb{F} O\right)}{d(d-1)}\left(\frac{\mathbb{I}-\mathbb{F}}{2}\right)\\
        &=\frac{\Tr\!\left(O\right) - d^{-1}\Tr\!\left(\mathbb{F} O\right)}{d^2-1}\mathbb{I}+\frac{\Tr\!\left(O\mathbb{F}\right) - d^{-1}\Tr\!\left(O\right)}{d^2-1}\mathbb{F},
    \end{align}
where in the first equality we used the fact that $P^{(d,2)}_{\mathrm{sym}}$ and $P^{(d,2)}_{\mathrm{asym}}$ are an orthonormal basis for the commutant, we utilized the fact that $P^{(d,2)}_{\mathrm{sym}}$ and $P^{(d,2)}_{\mathrm{asym}}$ are orthogonal projectors (meaning that they square to themselves and are Hermitian), in the third equality we substituted $P^{(d,2)}_{\mathrm{sym}}=\frac{1}{2}\left(\mathbb{I}+\mathbb{F}\right)$ and $P^{(d,2)}_{\mathrm{asym}}=\frac{1}{2}\left(\mathbb{I}-\mathbb{F}\right)$ along with their respective traces to arrive at the desired formula. 

Now we show a theorem that states that for all quantum states $\ket{\phi}$ in a $d$-dimensional Hilbert space, the moment operator of $\ketbra{\phi}{\phi}^{\otimes k}$ is a uniform linear combination of permutations, which can be written in terms of the projector on the symmetric subspace $P^{(d,k)}_{\mathrm{sym}}$.
 
\begin{theorem}
\label{th:MomPermSym}
Let $d,k\in \mathbb{N}$. For all $\ket{\phi} \in \mathbb{C}^d$, the moment operator is a uniform linear combination of permutations:
\begin{align}
    \ExU\left[U^{\otimes k}\ketbra{\phi}{\phi}^{\otimes k}U^{\dagger \otimes k}\right]=\frac{P^{(d,k)}_{\mathrm{sym}}}{\Tr\!\left(P^{(d,k)}_{\mathrm{sym}}\right)},
    \label{eq:MomPermSym}
\end{align}
with $P^{(d,k)}_{\mathrm{sym}}=\frac{1}{k!}\sum_{\pi\in S_k} V_d(\pi)$ and $\Tr\!\left(P^{(d,k)}_{\mathrm{sym}}\right)=\mathrm{dim}\left(\mathrm{Sym}_k(\mathbb{C}^d)\right)=\binom{k+d-1}{k}$.
\end{theorem}
\begin{proof}
For all $\sigma \in S_k$, we left-multiply both sides of Eq.~\eqref{eq:MomPerm} by $V_d(\sigma^{-1})$ and obtain:
    \begin{align}
    V_d(\sigma^{-1})\MomOp{\left(\ketbra{\phi}{\phi}\right)^{\otimes k}}=\sum_{\pi \in S_k} c_\pi   V_d(\sigma^{-1})V_d(\pi),
\end{align}
where $\MomOp{\left(\ketbra{\phi}{\phi}\right)^{\otimes k}}\coloneqq \ExU \left[\left(U\ketbra{\phi}{\phi}U^\dagger\right)^{\otimes k} \right]$.
Using the fact that $V_d(\sigma^{-1})$ commutes with $U^{\otimes k}$ for all unitaries $U$ and that $V_d(\sigma^{-1})\ket{\phi}^{\otimes k}=\ket{\phi}^{\otimes k}$, the LHS is: \begin{align}
V_d(\sigma^{-1})\MomOp{\left(\ketbra{\phi}{\phi}\right)^{\otimes k}}=\MomOp{\left(\ketbra{\phi}{\phi}\right)^{\otimes k}}  .  
\end{align}
The RHS is: 
\begin{align}
\sum_{\pi \in S_k} c_\pi   V_d(\sigma^{-1})V_d(\pi)=
\sum_{\pi \in S_k} c_\pi V_d(\sigma^{-1}\pi) = \sum_{\pi \in S_k} c_{\sigma\pi} V_d(\pi).   
\end{align}
Therefore, we have:
\begin{align}
\MomOp{\left(\ketbra{\phi}{\phi}\right)^{\otimes k}}  = \sum_{\pi \in S_k} c_{\sigma\pi} V_d(\pi).
\end{align}
Assuming $d>1$, the previous equation implies that all the coefficients must be the same, i.e., $c_{\sigma}=c_{I}$ for all $\sigma \in S_k$ (note that this can be shown without using linear independence of the permutation matrices, but just using that the permutation matrices are all different linear operators for $d>1$\footnote{To see it explicitly, one can choose the permutation $\sigma$, which fixes all indices except for two, which are interchanged.}).
Hence, we have:
\begin{align}
    \MomOp{\left(\ketbra{\phi}{\phi}\right)^{\otimes k}}=c_I\left(\sum_{\pi \in S_k} V_d(\pi)\right)=c_I k!\, P^{(d,k)}_{\mathrm{sym}}.
\end{align}
Thus, taking the trace, we get:
\begin{align}
    c_I=\frac{1}{k! \Tr\!\left(P^{(d,k)}_{\mathrm{sym}}\right)}=\frac{1}{k!} \frac{1}{\binom{k+d-1}{k}}.
\end{align}
\end{proof}
\begin{definition}[Haar measure on states]
Given a state $\ket{\phi}$ in $\mathbb{C}^d$, we denote
     \begin{align}
         \Expsi\left[\ketbra{\psi}{\psi}^{\otimes k}\right]\coloneqq \ExU\left[U^{\otimes k}\ketbra{\phi}{\phi}^{\otimes k}U^{\dagger \otimes k}\right].
     \end{align}
\end{definition}
 Note that the right invariance of the Haar measure implies that the definition of $\Expsi\left[\ketbra{\psi}{\psi}^{\otimes k}\right]$ does not depend on the choice of $\ket{\phi}$. Moreover, due to Theorem~\ref{th:MomPermSym}, we have:
 \begin{align}
     \Expsi\left[\ketbra{\psi}{\psi}^{\otimes k}\right]=\frac{P^{(d,k)}_{\mathrm{sym}}}{\Tr\!\left(P^{(d,k)}_{\mathrm{sym}}\right)}.
 \end{align}

\section{Vectorization formalism}
\label{sec:vecform}
We define the linear operator $\mathrm{vec}:\MatC{d} \rightarrow (\mathbb{C}^{d})^{\otimes 2}$ as follows: for all $i,j\in[d]$, we set $\mathrm{vec}\!\left(\ketbra{i}{j}\right)\coloneqq \ket{i}\otimes\ket{j}$. This means that if $A\in \MatC{d}$ is written as $A=\sum^d_{i,j=1} A_{i,j} \ketbra{i}{j}$, then $\mathrm{vec}(A)=\sum^d_{i,j=1} A_{i,j} \ket{i}\otimes\ket{j}$. This justifies the name \textquote{vectorization}, as the linear map $\mathrm{vec(\cdot)}$ sends the $d\times d$ matrix $A$ to the $d^2$-dimensional vector $\mathrm{vec}(A)$. Importantly, $\mathrm{vec(\cdot)}$ is a bijection, meaning that for all $\ket{v}\in (\mathbb{C}^d)^{\otimes 2}$ there exists a unique $A\in \MatC{d}$ such that $\ket{v}=\mathrm{vec}(A)$ (in fact, one can take $A=\sum^d_{i,j=1}\braket{i,j}{v}\ketbra{i}{j}$) and vice versa.

It is often convenient to express $\mathrm{vec(A)}$ in the form $\mathrm{vec}(A)=A\otimes I\ket{\Omega}$, where $\ket{\Omega}\coloneqq \mathrm{vec}(I)=\sum^d_{i=1}\ket{i,i}$ is the vectorization of the identity matrix, i.e., the non-normalized maximally entangled state.
To simplify the notation, we can use the shorthand $\kket{A}\coloneqq \mathrm{vec}(A)$.
Another useful property is that the canonical inner product between two vectorized operators $A$ and $B$ is equal to their Hilbert-Schmidt scalar product, i.e., $\bbrakket{A}{B}=\hs{A}{B}$, where $\bbrakket{A}{B}\coloneqq \mathrm{vec}\!\left(A\right)^\dagger \mathrm{vec}(B)$ is the canonical inner product. 
Using the so-called \emph{transpose}-trick $A\otimes I\ket{\Omega}=I\otimes A^T \ket{\Omega}$, we have the \emph{ABC-rule}: 
\begin{align}
    \kket{ABC}=A\otimes C^T \kket{B},
    \label{eq:ABCtrick}
\end{align}
for all $A,B,C \in \MatC{d}$.

Consider a linear superoperator $\Phi\colon \mathcal{L}(\mathbb{C}^d) \rightarrow \mathcal{L}(\mathbb{C}^d)$. It is worth noting that the map $\kket{\Phi(X)}$ is linear with respect to $\kket{X}$ for all $\kket{X} \in (\mathbb{C}^d)^{\otimes 2}$. This linearity implies the existence of a matrix $\mathrm{vec}(\Phi) \in \mathcal{L}((\mathbb{C}^d)^{\otimes 2})$ associated with this linear transformation. In other words, we can express the action of $\Phi$ on an operator $X$ as follows:
\begin{align}
\label{eq:phikraus}
\kket{\Phi(X)}= \mathrm{vec}(\Phi)\kket{X} \quad \text{for all $X\in\mathcal{L}(\mathbb{C}^d)$.}
\end{align}

It is important to mention that for all linear superoperators $\Phi:\mathcal{L}(\mathbb{C}^d)\rightarrow\mathcal{L}(\mathbb{C}^d)$ there exists a set of matrices $\{A_i\}^{d^2}_{i=1}$, $\{B_i\}^{d^2}_{i=1} \in \mathcal{L}(\mathbb{C}^d)$ such that $\Phi(X)=\sum^{d^2}_{i=1}A_i X B^\dagger_i$. 
To see why, observe that:
\begin{align}
\label{eq:choiIso}
    \Phi(X)=\Tr_{2}\left(I\otimes X^{T} \rho_{\Phi}\right),
\end{align} 
where $\rho_{\Phi} := \Phi \otimes \mathcal{I}(\ketbra{\Omega}{\Omega}) = \sum_{i,j=1}^d \Phi(\ketbra{i}{j}) \otimes \ketbra{i}{j}$ is the so-called Choi matrix of $\Phi$, $\mathcal{I}:\mathcal{L}(\mathbb{C}^d)\rightarrow\mathcal{L}(\mathbb{C}^d)$ is the identity channel and $\Tr_{2}$ indicates the partial trace with respect to the second tensor space. Now, performing a singular value decomposition (SVD) on $\rho_{\Phi}$, we have:
\begin{align}
\rho_\Phi=\sum^{d^2}_{i=1}\lambda_i\ketbra{u_i}{v_i}=\sum^{d^2}_{i=1}\ketbra{\mathrm{vec}(A_i)}{\mathrm{vec}(B_i)}=\sum^{d^2}_{i=1} A_i\otimes I \ketbra{\Omega}{\Omega}B^\dagger_i\otimes I,
\end{align}
where $\lambda_i\ge 0$, $\ket{u_i}$ and $\ket{v_i}$ are respectively the singular values and the column vectors associated to the SVD unitaries, and we defined $A_i\coloneqq \sqrt{\lambda_i}\mathrm{vec}^{-1}\!(\ket{u_i})$ and $B_i\coloneqq \sqrt{\lambda_i}\mathrm{vec}^{-1}\!(\ket{v_i})$. 
Substituting this expression back into \eqref{eq:choiIso}, we obtain the desired claim.
(Note that if $\Phi$ is completely positive, then its Choi matrix is positive, so we can perform an eigendecomposition instead of an SVD, and arrive at the same conclusion as before, with $B_i=A_i$.)

Using the \emph{ABC-rule} \eqref{eq:ABCtrick}, we can rewrite $\Phi(X)=\sum^{d^2}_{i=1}A_i X B^\dagger_i$ as $\kket{\Phi(X)}=\sum^{d^2}_{i=1} A_i \otimes B^*_i \kket{X}$, which implies:
\begin{align}
    \mathrm{vec}(\Phi)=\sum^{d^2}_{i=1} A_i \otimes B^*_i.
\end{align}
Additionally, it is often advantageous to express the moment operator $\mathcal{M}_{\mu_H}^{(k)}\!\left(O\right)= \ExU \left[U^{\otimes k} O U^{\dagger \otimes k}\right]$ in its vectorized form:
\begin{align}
    \kket{\mathcal{M}_{\mu_H}^{(k)}\!\left(O\right)}=\ExU \left[U^{\otimes k} \otimes U^{* \otimes k}\right] \kket{O}.
\end{align}

\begin{definition}[Vectorized moment operator]
    We define the \emph{vectorized moment operator} $\mathrm{M}^{(k)}_{\mu_H} \in \mathcal{L}\left((\mathbb{C}^d)^{\otimes 2k}\right)$ as:
    \begin{align}
        \mathrm{M}^{(k)}_{\mu_H}\coloneqq \mathrm{vec}(\mathcal{M}_{\mu_H}^{(k)})=\ExU \left[U^{\otimes k} \otimes U^{* \otimes k}\right].
    \end{align}
\end{definition}
Proposition~\ref{prop:projComm} has a vectorized version as follows:
\begin{proposition}
\label{prop:MomVec}
Let the commutant space be $\mathrm{Comm}\coloneqq \mathrm{Comm}(\Ug(d),k)$. Let $\{P_i\}^{\mathrm{dim}\left(\mathrm{Comm}\right)}_{i=1}$ be elements of an orthonormal basis of $\mathrm{Comm}$ with respect to the Hilbert-Schmidt inner product. Then, we have: 
    \begin{align}
    \mathrm{M}^{(k)}_{\mu_H}\coloneqq \ExU \left[U^{\otimes k} \otimes U^{* \otimes k}\right]=\sum^{\mathrm{dim}\!\left(\mathrm{Comm}\right)}_{i=1} \kket{P_i}\!\bbra{P_i}.
    \end{align}
Moreover, we have $\Tr\!\left(\mathrm{M}^{(k)}_{\mu_H}\right)=\mathrm{dim}\left(\mathrm{Comm}\right)$.
\end{proposition}
\begin{proof}
    Vectorizing the equation derived in Proposition~\ref{prop:projComm}, we have:
    \begin{align}
    \mathrm{M}^{(k)}_{\mu_H}\kket{O}=\kket{\MomOp{O}}=\sum^{\mathrm{dim}\left(\mathrm{Comm}\right)}_{i=1} \langle P_i ,O\rangle_{HS} \kket{P_i}=\sum^{\mathrm{dim}\left(\mathrm{Comm}\right)}_{i=1}\kket{P_i} \bbrakket{P_i}{O}. 
    \end{align}
    Since the previous equation holds for all $O\in \mathcal{L}\left((\mathbb{C}^d)^{\otimes k}\right)$ and therefore for all $\kket{O}\in (\mathbb{C}^{d})^{\otimes 2k}$, we have the thesis.
\end{proof}
Building upon the previous proposition and the fact that $\Tr\!\left(\mathrm{M}^{(k)}_{\mu_H}\right)=\ExU \left|\Tr(U)\right|^{2k}$, we can derive the following equation:
\begin{align}
     \ExU \left|\Tr(U)\right|^{2k}= \mathrm{dim}(\mathrm{Comm}(\Ug(d),k)).
\end{align}
In the case where $d\ge k$ (which is the case we are most interested in), we find that $\ExU \left|\Tr(U)\right|^{2k}=k!$. This is due to the fact that the $k$-th order commutant of the unitary group corresponds to the span of the permutation operators (as a consequence of Schur-Weyl duality, see Theorem~\ref{th:SchurW}), which are linearly independent for $d\ge k$ (as stated in Proposition~\ref{prop:permInd}).

It is also worth noting that the vectorized moment operator $\mathrm{M}^{(k)}_{\mu_H}$ is an orthogonal projector, since $\mathrm{M}^{(k)2}_{\mu_H}=\mathrm{M}^{(k)}_{\mu_H}$ (due to the right invariance of the Haar measure) and $\mathrm{M}^{(k)\dagger}_{\mu_H}=\mathrm{M}^{(k)}_{\mu_H}$ (due to Proposition~\ref{prop:Haar}).

Furthermore, as an exercise, we provide an alternative proof of Theorem~\ref{prop:projComm} based on vectorization.

\begin{theorem}[Projector onto the commutant]\
\label{th:ProjComm}
The moment operator $ {\mathcal{M}_{\mu_H}^{(k)}\!\left(\cdot\right)=\ExU \left[U^{\otimes k} \left(\cdot\right) U^{\dagger \otimes k}\right]}$  is the orthogonal projector onto the commutant $\mathrm{Comm}\coloneqq \mathrm{Comm}(\Ug(d),k)$ with respect to the Hilbert-Schmidt inner product.
In particular, let $P_1,\ldots,P_{\mathrm{dim}(\mathrm{Comm})}$ be an orthonormal basis of $\mathrm{Comm}$ and let $O \in \mathcal{L}((\mathbb{C}^d)^{\otimes k})$. Then, we have:
\begin{align}
\MomOp{O}=\sum^{\mathrm{dim}\left(\mathrm{Comm}\right)}_{i=1} \langle P_i ,O\rangle_{HS} P_i \,.
\end{align}

\end{theorem}
\begin{proof}
    Since the vectorized moment operator $\mathrm{M}^{(k)}_{\mu_H}$ is an orthogonal projector, it admits an eigendecomposition with eigenvalues $0$ and $1$ and eigenvectors $\{v_i\}^{r}_{i=1}$, where $r$ is the rank of the projector $\mathrm{M}^{(k)}_{\mu_H}$:
    \begin{align}
        \mathrm{M}^{(k)}_{\mu_H}\coloneqq \mathrm{vec}(\mathcal{M}_{\mu_H}^{(k)})=\ExU \left[U^{\otimes k} \otimes U^{* \otimes k}\right]=\sum^{r}_{i=1}\ketbra{v_i}{v_i}=\sum^{r}_{i=1}\kket{P_i}\!\bbra{P_i},
    \end{align}
  where we used the fact that for each $\ket{v_i}$, there exists $P_i \in \mathcal{L}((\mathbb{C}^d)^{\otimes k})$ such that $\ket{v_i} = \kket{P_i}$ for each $i \in [r]$.
  We note that $P_i$ are operators in $\mathrm{Comm}(\Ug(d), k)$. This can be shown as follows:
  \begin{align}
    P_i=\mathrm{vec}^{-1}\!\left(\kket{P_i}\right)=\mathrm{vec}^{-1}\!\left(\mathrm{vec}(\mathcal{M}_{\mu_H}^{(k)})\kket{P_i}\right)=\mathcal{M}_{\mu_H}^{(k)}\!(P_i)\,\, \in \mathrm{Comm}(\Ug(d),k),
  \end{align}
  where in the last step we used Lemma~\ref{le:PropMom}.\ref{propr:CoMom}. Furthermore, $\{P_i\}^r_{i=1}$ form an orthonormal set with respect to the Hilbert-Schmidt inner product:
  \begin{align}
      \hs{P_i}{P_j}=\bbrakket{P_i}{P_j}=\braket{v_i}{v_j}=\delta_{i,j} \quad \forall\,   i,j\in [r].
  \end{align}
  Moreover, they constitute a basis for $\mathrm{Comm}$ since, for all $A\in \mathrm{Comm}$, we have $\kket{A}=\mathrm{M}^{(k)}_{\mu_H}\kket{A}=\sum^{r}_{i=1}\bbrakket{P_i}{A}\kket{P_i}$, which implies that $A\in \operatorname{span}\left(\{P_i\}^r_{i=1}\right)$.
  Hence, the rank of $\mathrm{M}^{(k)}_{\mu_H}$ is equal to the dimension of the commutant, i.e., $r=\mathrm{dim}\left(\mathrm{Comm}\right)$.
  Consequently, we can conclude the proof by noting that for all $O\in \mathcal{L}\left((\mathbb{C}^d)^{\otimes k}\right)$:
  \begin{align}
      \mathcal{M}_{\mu_H}^{(k)}\!\left(O\right)= \mathrm{vec}^{-1}\!\left(\mathrm{vec}(\mathcal{M}_{\mu_H}^{(k)})\kket{O}\right)=\mathrm{vec}^{-1}\!\left(\sum^{r}_{i=1}\kket{P_i}\bbrakket{P_i}{O}\right)=\sum^{r}_{i=1}\hs{P_i}{O}P_i.
  \end{align}
\end{proof}
To address the issue of the permutation matrices not being orthonormal, we can introduce a vectorized version of Eq.\eqref{eq:momOpweing} that provides a concise representation of the moment operator.
\begin{align}
     \ExU \left[U^{\otimes k} \otimes U^{* \otimes k}\right] = \sum_{\pi,\sigma \in S_k} \mathrm{Wg}(\pi^{-1}\sigma,d) \kket{V_d(\pi)}\!\bbra{V_d(\sigma)}.
\end{align}
\begin{proof}
Using Eq.\eqref{eq:momOpweing}, we have:
\begin{align}
     \ExU \left[U^{\otimes k} \otimes U^{* \otimes k}\right] \kket{O}&=\ExU \kket{U^{\otimes k} O U^{\dagger \otimes k}}\\
     &= \sum_{\pi,\sigma \in S_k} \mathrm{Wg}\left(\pi^{-1}\sigma,d\right) \kket{V_d(\pi)}\Tr\!\left(V^\dagger_d(\sigma) O\right)\\
     &= \sum_{\pi,\sigma \in S_k} \mathrm{Wg}\left(\pi^{-1}\sigma,d\right) \kket{V_d(\pi)}\bbrakket{V_d(\sigma)}{O}.
\end{align}
Since this equation holds for all vectors $\kket{O}$ in $(\mathbb{C}^{d})^{\otimes 2k}$, we conclude the proof.
\end{proof}

We can now write the moment operator equations for $k=1$ and $k=2$ presented in Example~\ref{ex:SecondMoment} in their vectorized form: 
\begin{align}
\label{eq:mom1}
     \ExU \left[U \otimes U^{*}\right] &= \frac{1}{d}\kket{I}\!\bbra{I},\\
     \ExU \left[U^{\otimes 2} \otimes U^{* \otimes 2}\right] &= \frac{1}{d^2-1}\left[\kket{\mathbb{I}}\!\bbra{\mathbb{I}}-\frac{1}{d}\kket{\mathbb{I}}\!\bbra{\mathbb{F}}-\frac{1}{d}\kket{\mathbb{F}}\!\bbra{\mathbb{I}}+\kket{\mathbb{F}}\!\bbra{\mathbb{F}}\right].
\label{eq:mom2}
\end{align}
\begin{proof}
    For Equation \eqref{eq:mom1}, we have:
\begin{align}
            \ExU \left[U \otimes U^{*}\right] \kket{O}&=\ExU \kket{U O U^{\dagger}}=\frac{\Tr\!\left(O\right)}{d} \kket{I}=\frac{1}{d} \kket{I}\bbrakket{I}{O},
\end{align}
where we used $\bbrakket{I}{O}=\Tr\!\left(IO\right)=\Tr\!\left(O\right)$. For Equation \eqref{eq:mom2}, we have:
    \begin{align}
         \ExU \left[U^{\otimes 2} \otimes U^{* \otimes 2}\right] \kket{O}&=\ExU \kket{U^{\otimes 2} O U^{\dagger \otimes 2}}\\
         &=\frac{\Tr\!\left(O\right)-d^{-1}\Tr\!\left(\mathbb{F}O\right)}{d^2-1}\kket{\mathbb{I}} + \frac{\Tr\!\left(\mathbb{F}O\right)-d^{-1}\Tr\!\left(O\right)}{d^2-1}\kket{\mathbb{F}} \\
         &  = \frac{1}{d^2-1}\left[\kket{\mathbb{I}}\bbrakket{\mathbb{I}}{O}-\frac{1}{d}\kket{\mathbb{I}}\bbrakket{\mathbb{F}}{O}-\frac{1}{d}\kket{\mathbb{F}}\bbrakket{\mathbb{I}}{O}+\kket{\mathbb{F}}\bbrakket{\mathbb{F}}{O}\right],
    \end{align}
    where we used that $\mathbb{F}^\dagger=\mathbb{F}$, $\bbrakket{\mathbb{I}}{O}=\Tr\!\left(O\right)$ and $\bbrakket{\mathbb{F}}{O}=\Tr\!\left(\mathbb{F}O\right)$.
\end{proof}

Equations like the ones presented above can be effectively visualized and manipulated using Tensor Network diagrams. 

\section{Tensor network diagrams}
\label{sec:diagram}
Tensor Networks provide a graphical representation that simplifies the understanding and analysis of tensor operations \cite{Bridgeman_2017}. In these diagrams, tensors are represented as boxes (or nodes), and tensor contractions are represented as connections between boxes.
To illustrate the concept, we use specific diagrams for ket-states, bra-states, and matrices. A ket-state $\ket{\psi} \in \mathbb{C}^d$ is represented with a box and a \emph{leg} as follows:
\begin{align}
\ipic{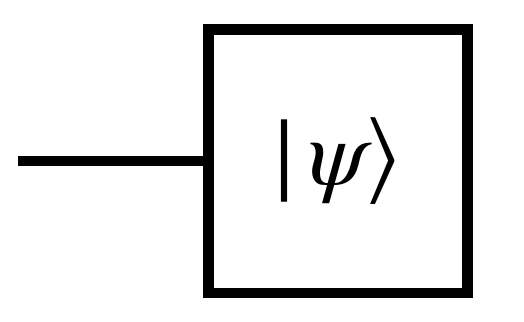}{0.3}.
\end{align}
Similarly, a bra-state $\bra{\psi}$ is represented as:
\begin{align}
    \ipic{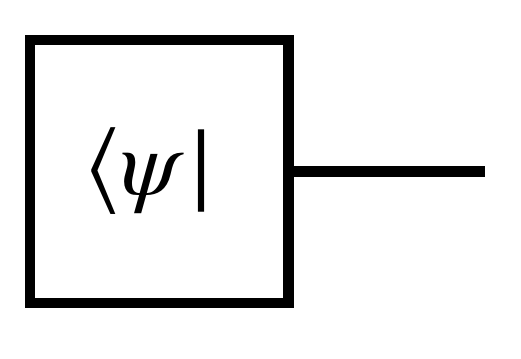}{0.3}.
\end{align}
Matrices, such as $A \in \MatC{d}$, are depicted like boxes with one leg in and one leg out as:
\begin{align}
    \ipic{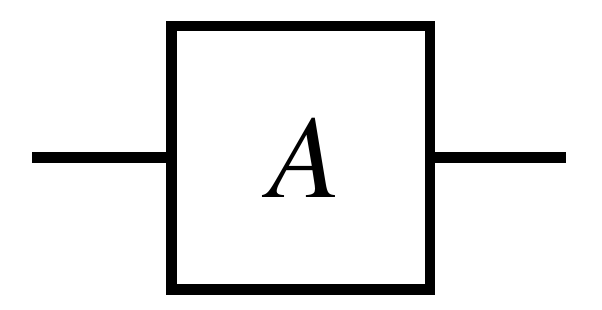}{0.3}.
\end{align}
In particular, the identity matrix $I \in \MatC{d}$ is simply represented by a line:
\begin{align}
    \ipic{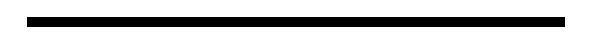}{0.3}.
\end{align}
The trace $\Tr(A)$ of an operator $A \in \MatC{d}$ is denoted as:
\begin{align}
    \ipic{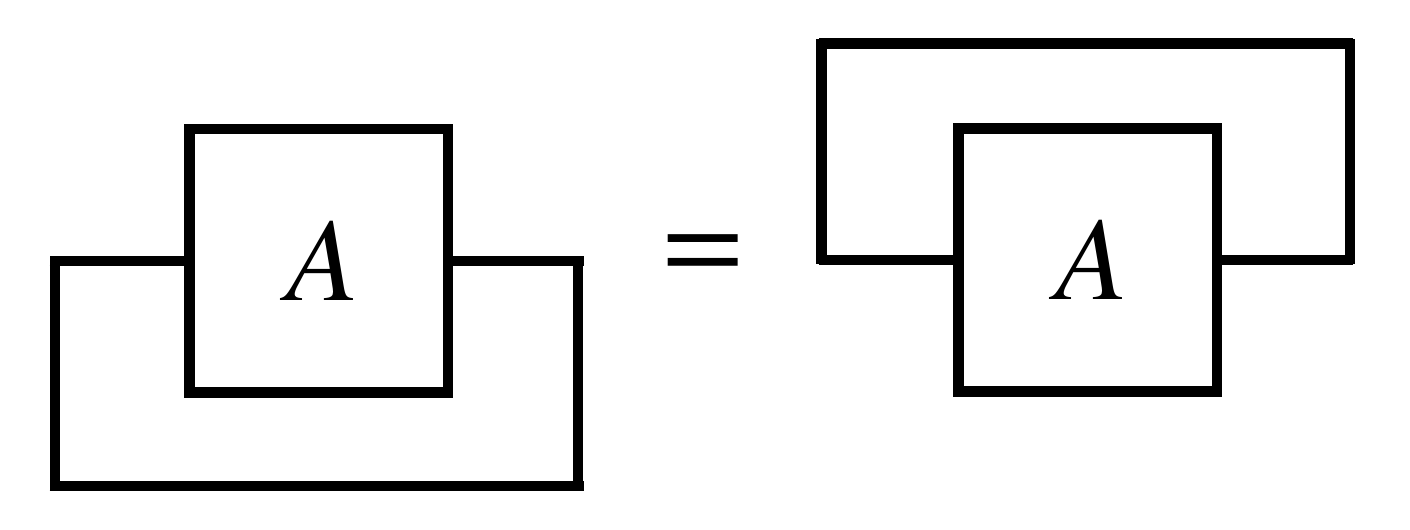}{0.3}.
\end{align}
Given $A,B\in \MatC{d}$, their product $AB$ is represented by placing the box representing matrix $A$ on the left side of the box representing matrix $B$. This can be illustrated as follows:
\begin{align}
    \ipic{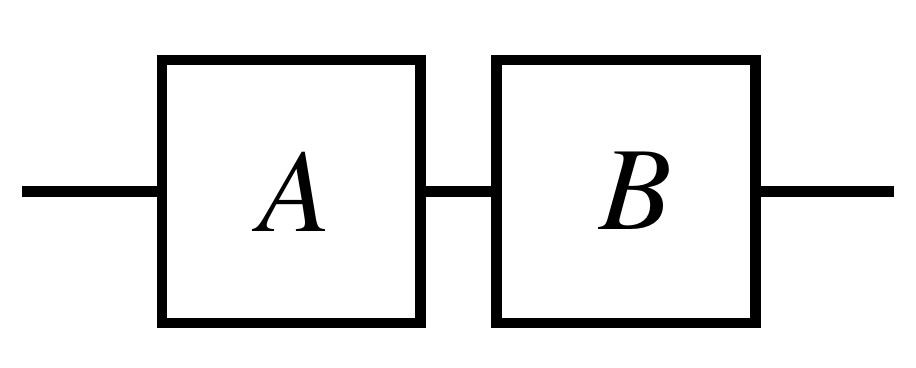}{0.3}.
\end{align}
The transpose $A^T$ of the matrix $A$ is depicted by: 
\begin{align}
    \ipic{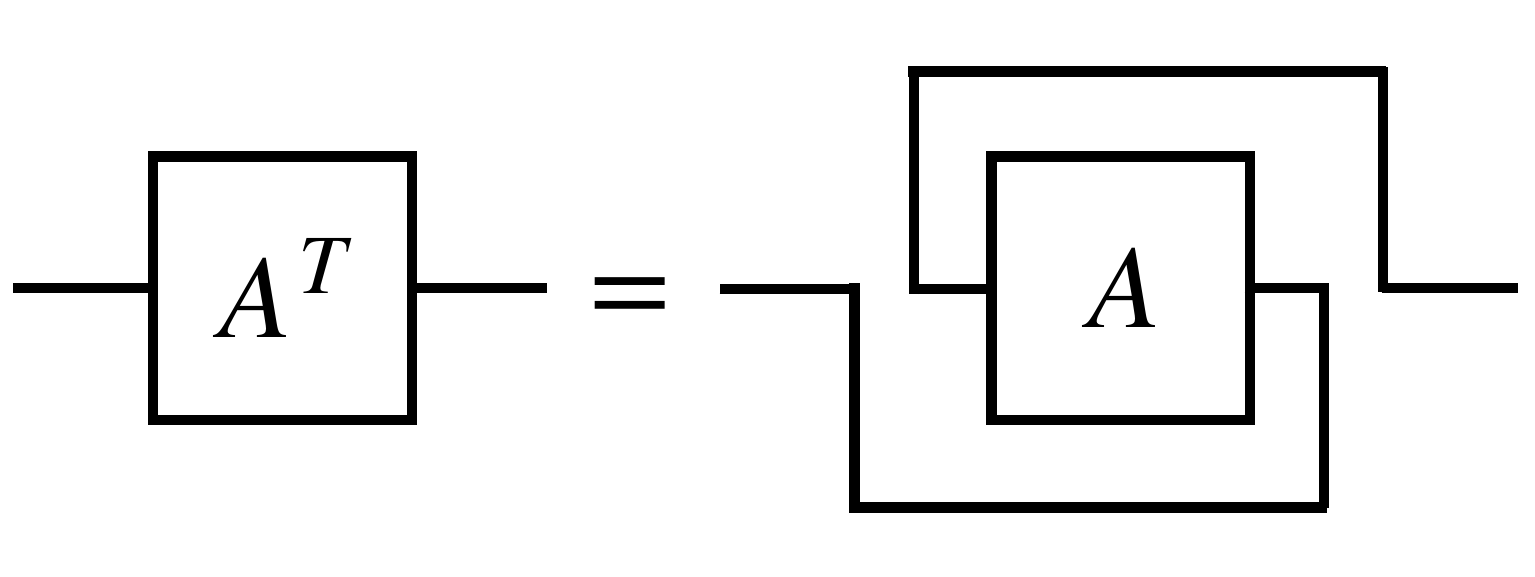}{0.3}.
\end{align}
Given $A\in \MatC{d}$ and $B\in \MatC{d^\prime}$, their tensor product $A \otimes B$ is denoted by arranging the box representing matrix $B$ below the box representing matrix $A$:
\begin{align}
    \ipic{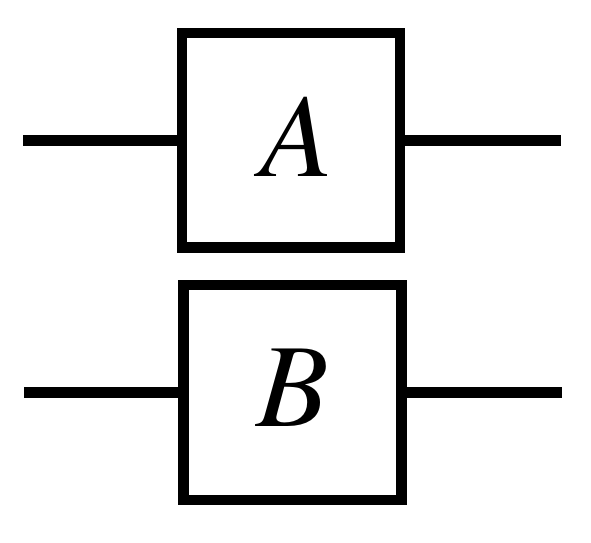}{0.3}.
\end{align}
$A\in \MatC{d} \otimes \MatC{d^\prime}$ is represented with a box that has two (different) legs on the left and two on the right:
\begin{align}
    \ipic{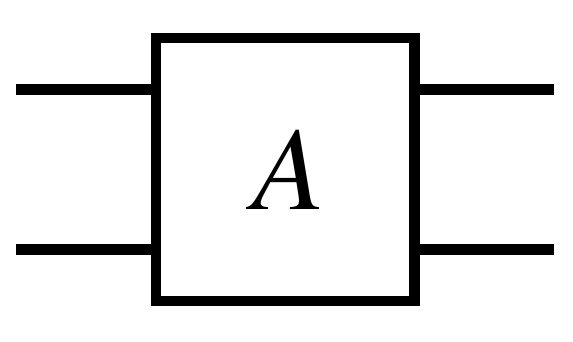}{0.3}.
\end{align}
The partial trace with respect to the second tensor space $\Tr_2\left(A\right)$ is denoted as:
\begin{align}
    \ipic{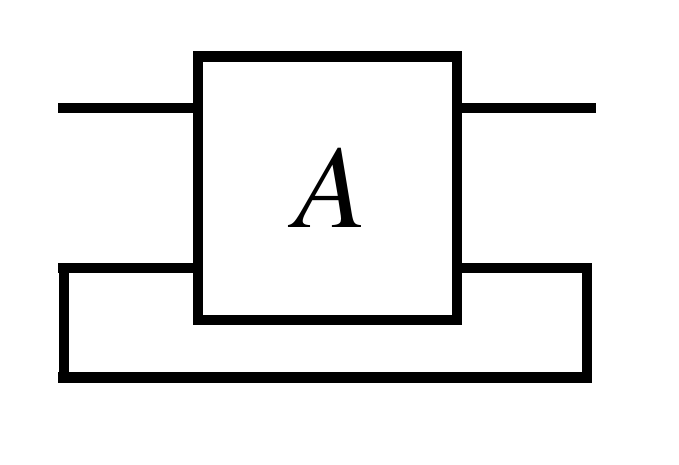}{0.28},
\end{align}
and similarly for the partial trace with respect to the first tensor space.

Particularly important is the non-normalized maximally entangled state $\ket{\Omega}=\kket{I}$, which is the vectorization of the identity and is denoted as:
\begin{align}
    \ipic{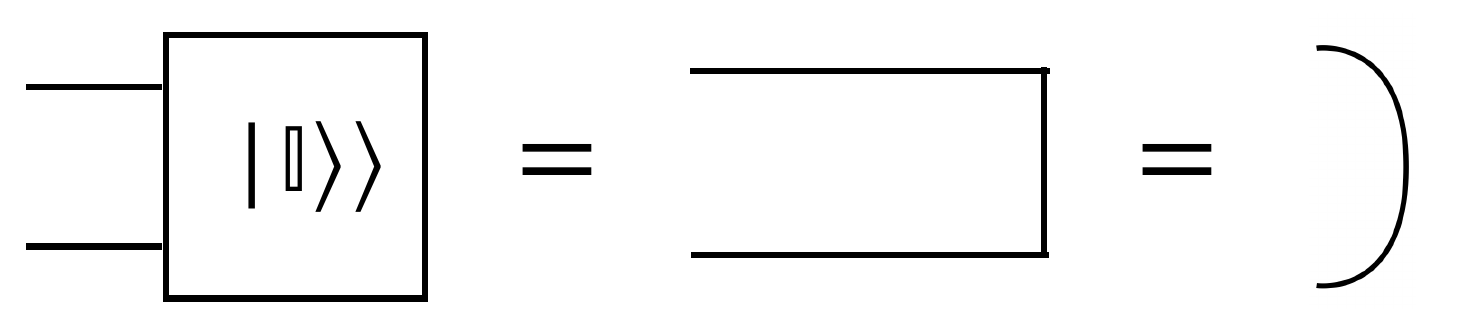}{0.35}.
\end{align}
Given an operator $A\in \MatC{d}$, its vectorization $\kket{A}=A\otimes I \ket{\Omega}$ is represented as:
\begin{align}
    \ipic{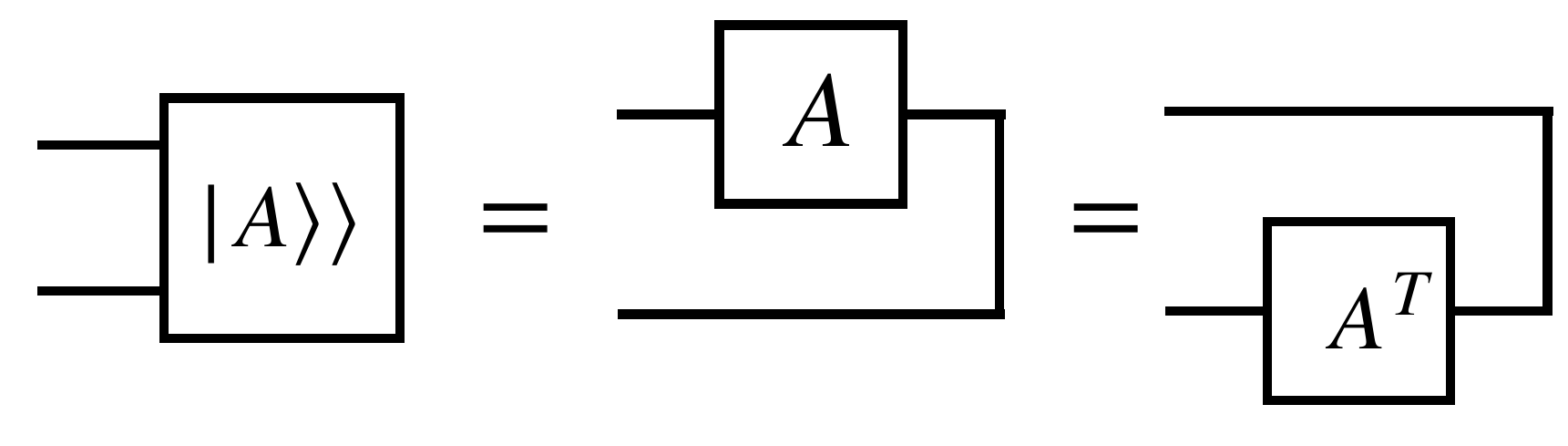}{0.3},
\end{align}
where the equality is due to the \emph{transpose-trick} $\kket{A}=A\otimes I \ket{\Omega}=I\otimes A^T \ket{\Omega}$. In Tensor Network notation, the diagram below makes it clear that $\bbrakket{A}{B}=\Tr\!\left(A^\dagger B\right)$:
\begin{align}
    \ipic{images/innerprod.pdf}{0.38},
\end{align}
as well as the \emph{ABC-rule}~\eqref{eq:ABCtrick}:
\begin{align}
    \ipic{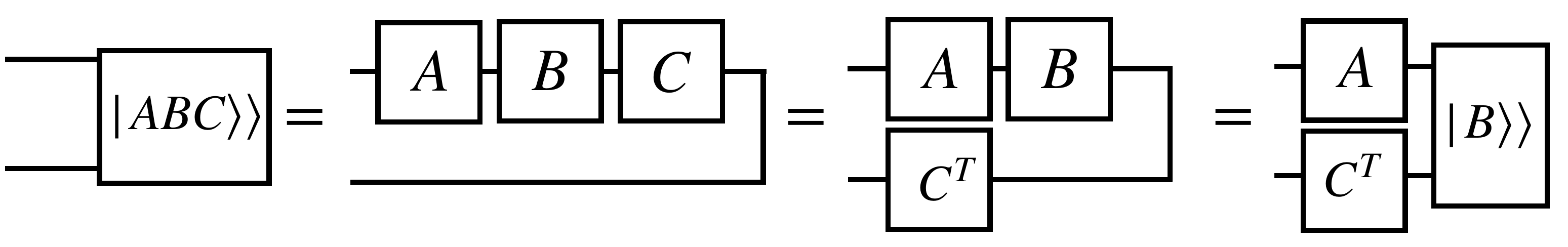}{0.26},
\end{align}
where $A,B,C\in \MatC{d}$.

In accordance with the Definition~\ref{def:permutation}, a permutation matrix $V_d(\pi)$, with $\pi\in S_k$, maps the tensor product state $\ket{\psi_1}\otimes \cdots \otimes \ket{\psi_k}$ to the tensor product state $\ket{\psi_{\pi^{-1}(1)}}\otimes \cdots \otimes \ket{\psi_{\pi^{-1}(k)}}$:
\begin{align}
    \ipic{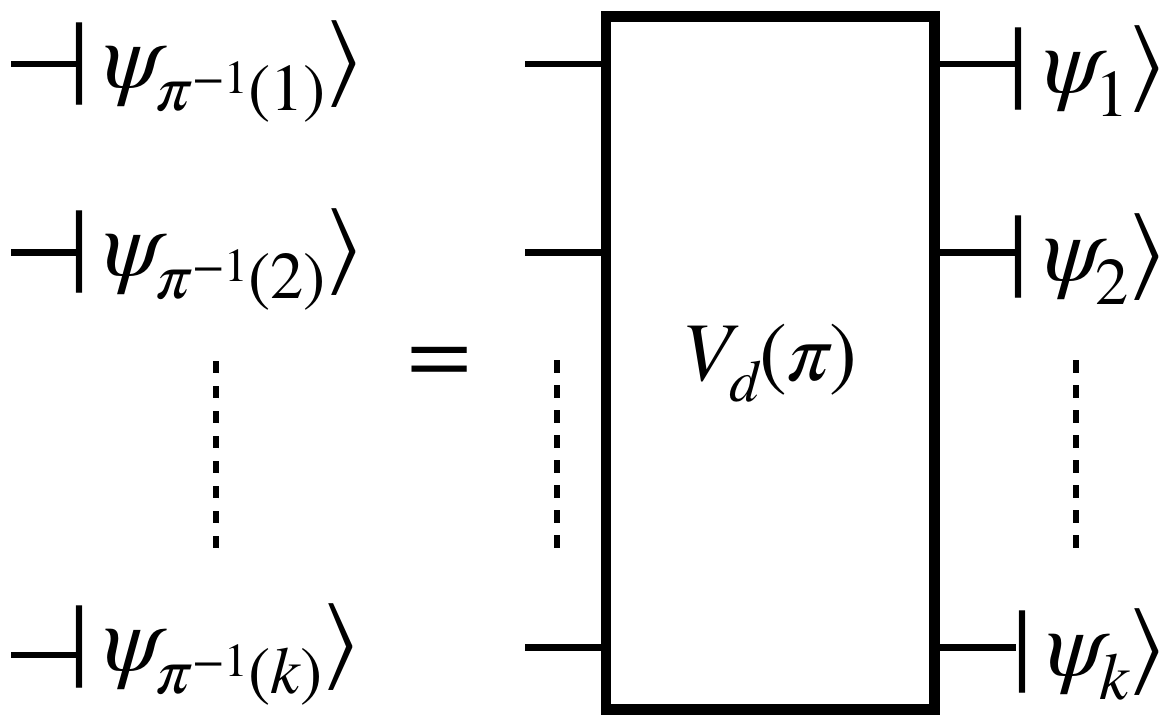}{0.30},
\end{align}
and it can be represented by a diagram with lines connecting the tensor space of the $i$-th element on the right with the the $\pi(i)$-th element on the left for each $i\in [k]$. 
For example, the identity operator $\mathbb{I}$ and the Flip operator $\mathbb{F}$ are represented as follows: 
\begin{align}
    \ipic{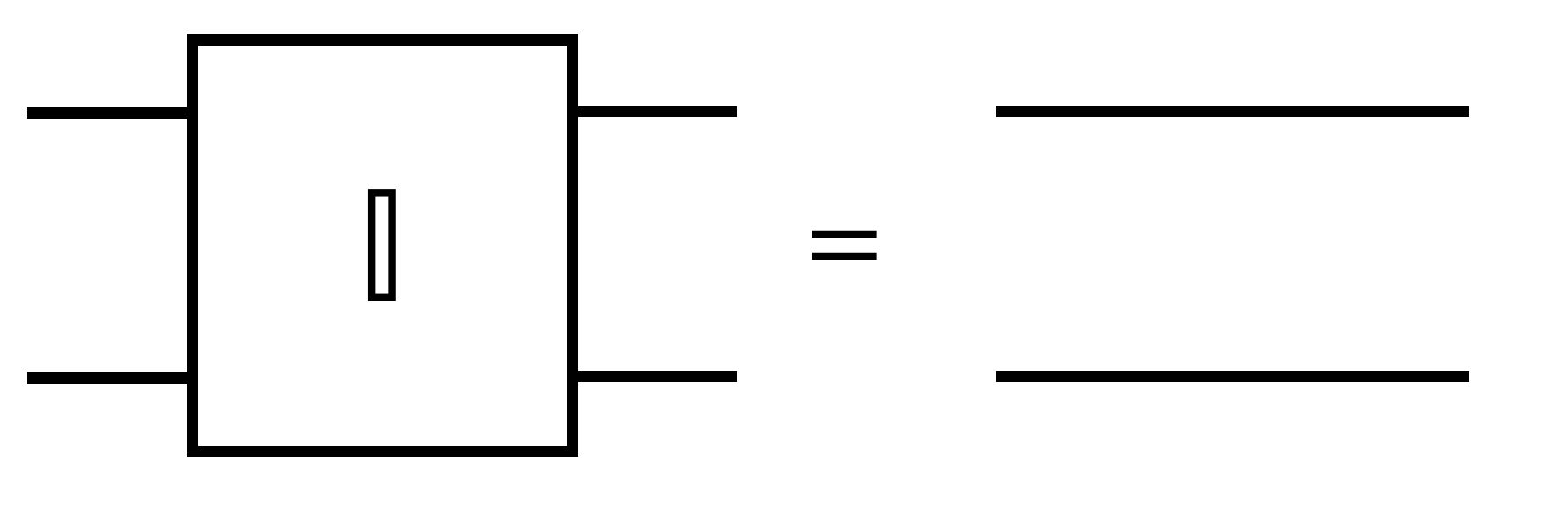}{0.20} \quad, \quad \ipic{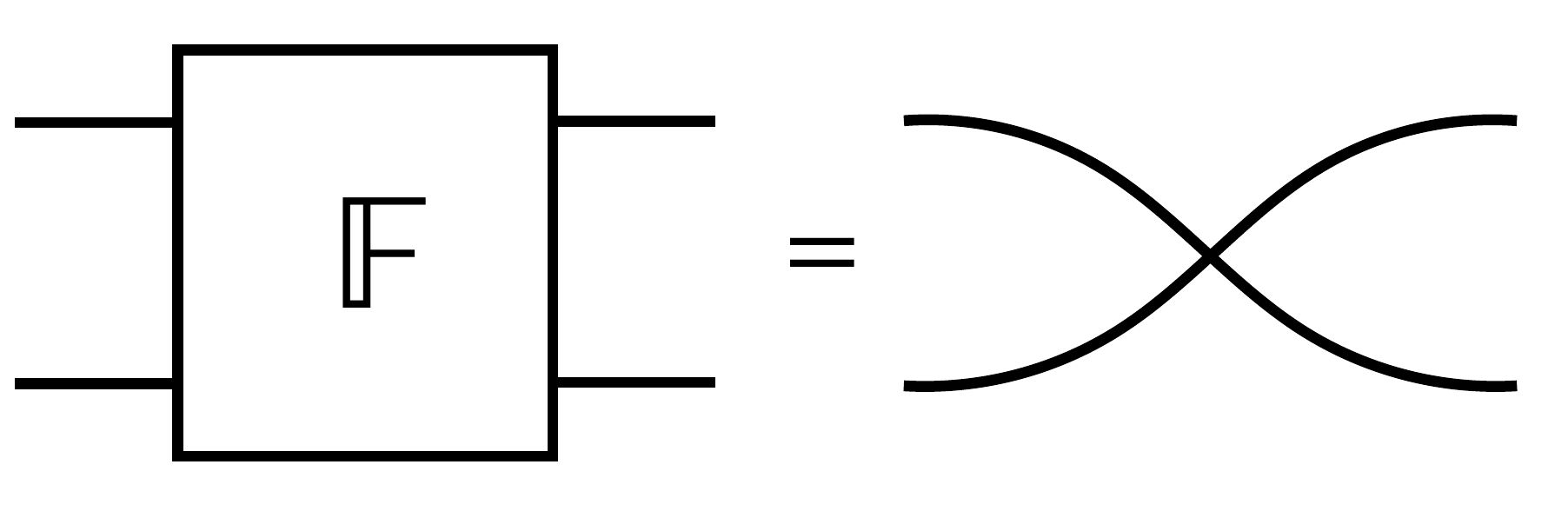}{0.20}.
\end{align}
The composition of permutations can be visualized by placing the corresponding diagrams next to each other.
When considering a matrix $A\in \MatC{d} \otimes \MatC{d^\prime}$, its vectorization is represented as:
\begin{align}
    \ipic{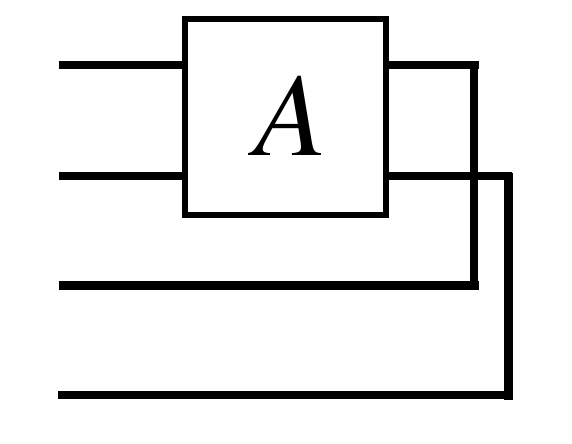}{0.35}.
\end{align}
For instance, the vectorization of the identity operator $\mathbb{I}$ and the Flip operator $\mathbb{F}$ are represented as:
\begin{align}
    \ipic{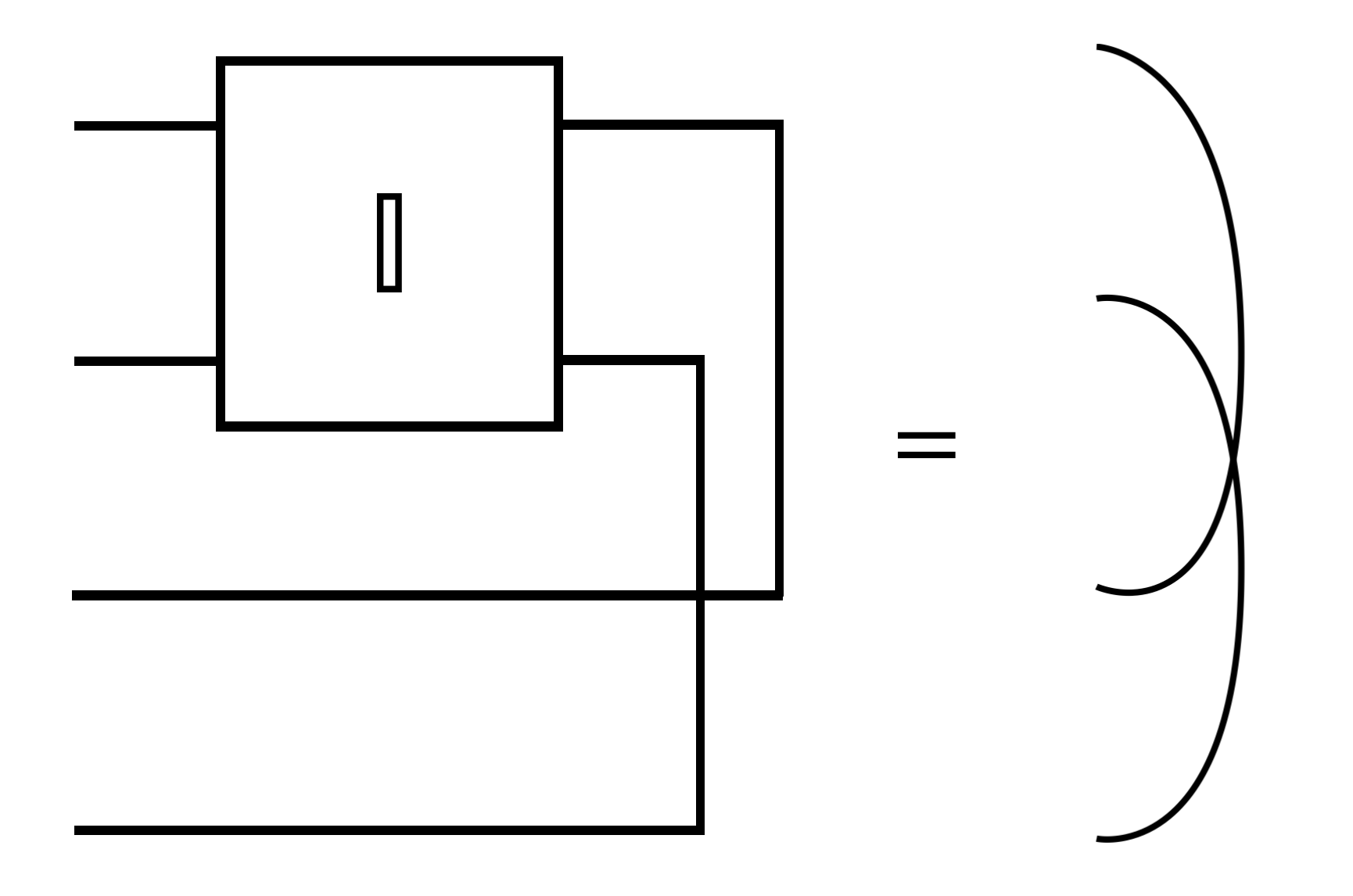}{0.15} \quad, \quad \ipic{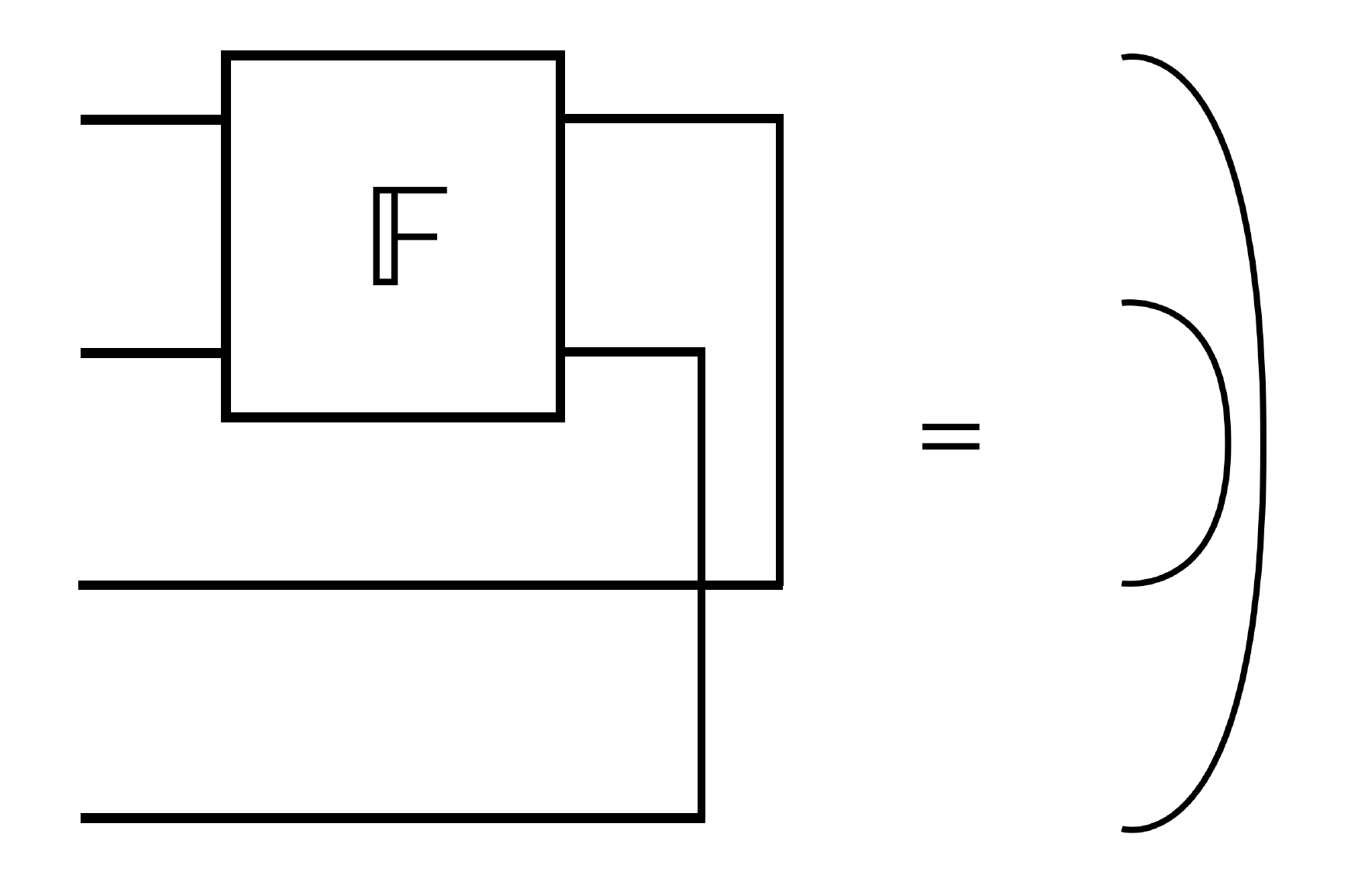}{0.15}.
\end{align}
Using this diagrammatic notation, we can express Eq.\eqref{eq:mom1} $\ExU \left[U \otimes U^{*}\right] = \frac{1}{d}\kket{I}\!\bbra{I}$ in Tensor Network notation as:
\begin{align}
\mathbb{E}_{U\sim\mu_H}\left[\ipic{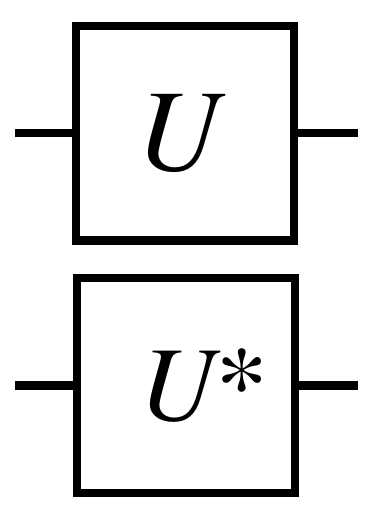}{0.26}\right]=\frac{1}{d}\ipic{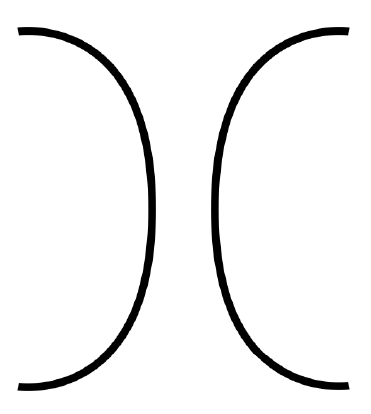}{0.26}.
\end{align}
Similarly, Eq.\eqref{eq:mom2} $\ExU \left[U^{\otimes 2} \otimes U^{* \otimes 2}\right] = \frac{1}{d^2-1}\left[\kket{\mathbb{I}}\!\bbra{\mathbb{I}}-\frac{1}{d}\kket{\mathbb{I}}\!\bbra{\mathbb{F}}-\frac{1}{d}\kket{\mathbb{F}}\!\bbra{\mathbb{I}}+\kket{\mathbb{F}}\!\bbra{\mathbb{F}}\right]$ can be represented as:

\begin{align}
\mathbb{E}_{U\sim\mu_H}\left[\ipic{images/2mom}{0.26}\right]=\frac{1}{d^2-1}\left[\;\;\ipic{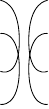}{1.6}-\frac{1}{d}\ipic{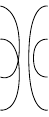}{1.6}-\frac{1}{d}\ipic{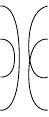}{1.6}+\ipic{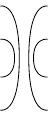}{1.6}\;\;\right]\;.\label{eq:graphweingarten}
\end{align}
Using the diagrammatic notation, we can easily see that:
\begin{align}
    \ExU \left[U_{i_1,j_1}U^*_{i_2,j_2}\right]&=\frac{1}{d}\delta_{i_1,i_2}\delta_{j_1,j_2},
\end{align}
where we used $\ExU \left[U_{i_1,j_1}U^*_{i_2,j_2}\right]=\bra{i_1,i_2}\ExU \left[U \otimes U^{*}\right]\ket{j_1,j_2}$ and $\delta_{a,b}$ represents the Kronecker delta. Similarly:
\begin{align}
    \ExU \left[U_{i_1,j_1}U_{i_2,j_2}U^{*}_{i_3,j_3}U^{*}_{i_4,j_4}\right]&=\frac{1}{d^2-1}\left[\delta_{i_1,i_3}\delta_{i_2,i_4}\delta_{j_1,j_3}\delta_{j_2,j_4}-\frac{1}{d}\delta_{i_1,i_3}\delta_{i_2,i_4}\delta_{j_1,j_4}\delta_{j_2,j_3}\right]\\
    & +\frac{1}{d^2-1}\left[-\frac{1}{d}\delta_{i_1,i_4}\delta_{i_2,i_3}\delta_{j_1,j_3}\delta_{j_2,j_4}+\delta_{i_1,i_4}\delta_{i_2,i_3}\delta_{j_1,j_4}\delta_{j_2,j_3}\right],
\end{align}
where we used $\ExU \left[U_{i_1,j_1}U_{i_2,j_2}U^{*}_{i_3,j_3}U^{*}_{i_4,j_4}\right]=\bra{i_1,i_2,i_3,i_4}\ExU \left[U^{\otimes 2} \otimes U^{* \otimes 2}\right]\ket{j_1,j_2,j_3,j_4}$.
Another useful relation is the \emph{partial-swap-trick} $\Tr_2(A\otimes B\, \mathbb{F})= AB$, which can be visualized as:
\begin{align}
    \ipic{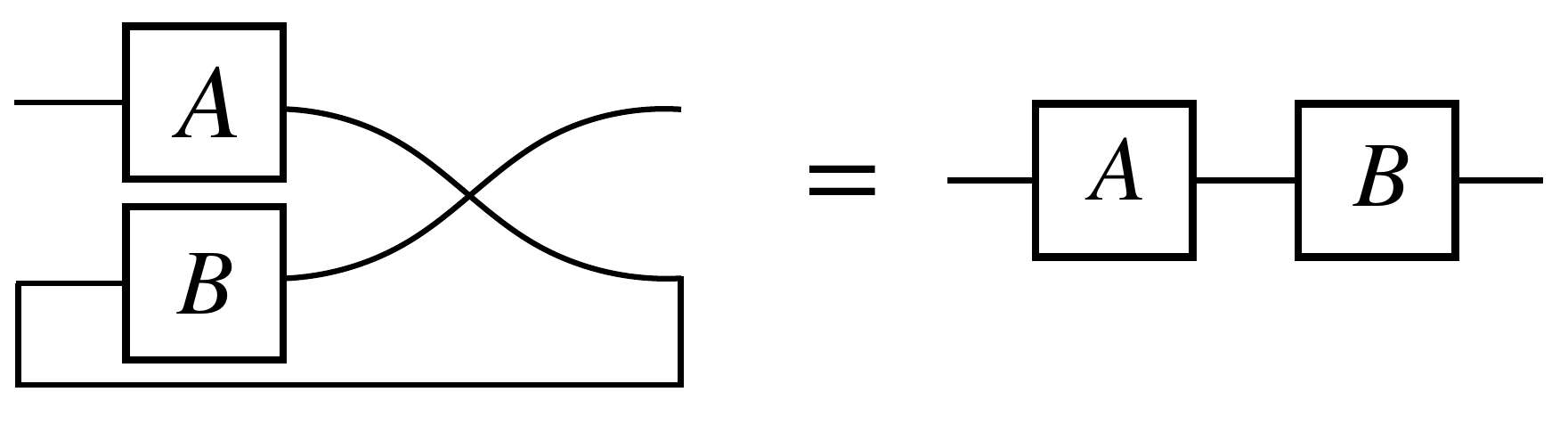}{0.3},
    \label{eq:partialswap}
\end{align}
and thus we have the \emph{swap-trick} formula $\Tr(A\otimes B\, \mathbb{F})= \Tr(AB)$.
Similarly we also have $\Tr_1(A\otimes B\, \mathbb{F})= BA$.

More generally, consider the cyclic permutation $\pi_{cyc} \in S_k$, which corresponds to the unitary operator: 
\begin{align}
    V_d(\pi_{cyc})=\sum_{i_1,\dots,i_k \in [d]} \ketbra{i_2,\dots,i_{k},i_{1}}{i_1,i_2,\dots,i_k}.
\end{align}
This unitary operator can be depicted diagrammatically as:
\begin{align}
     \ipic{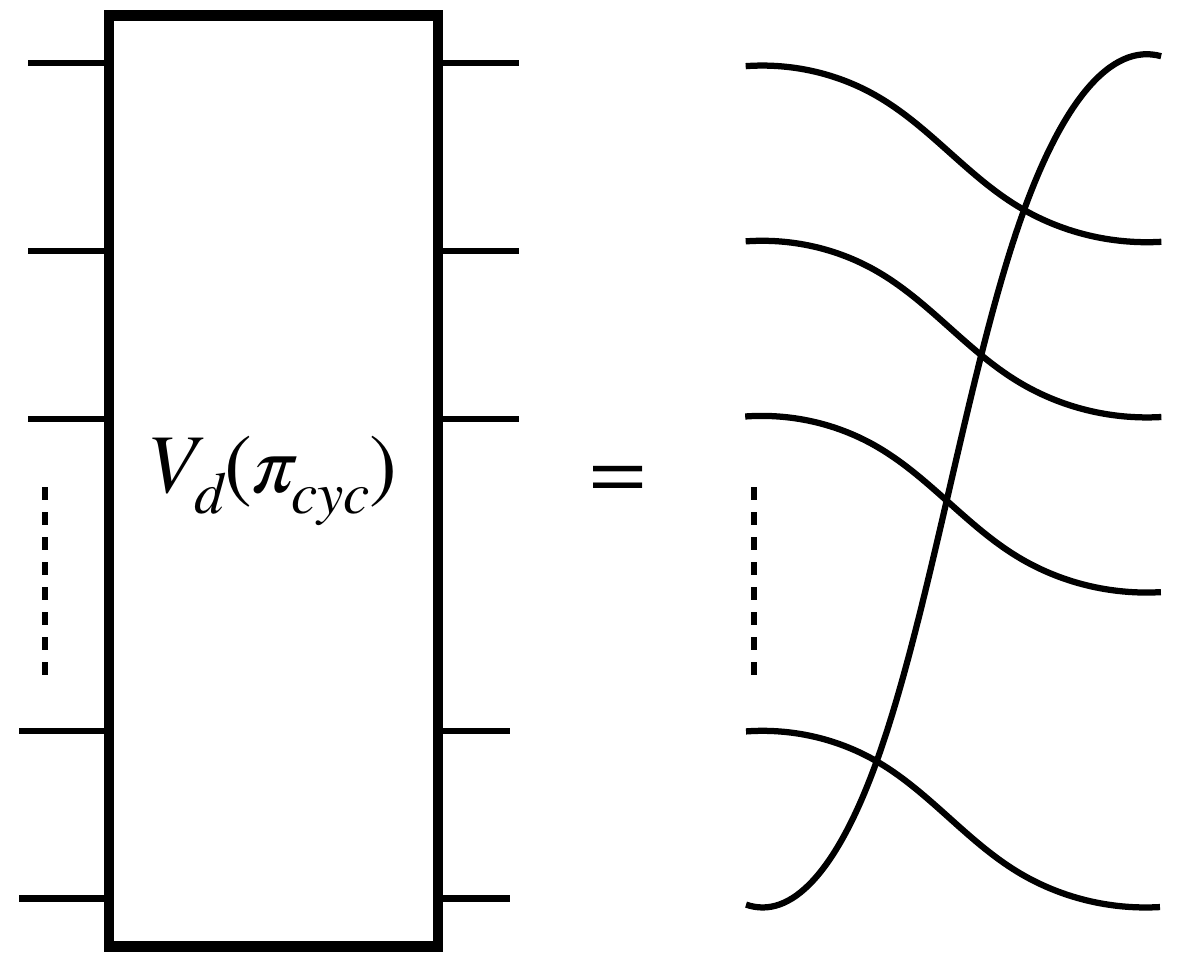}{0.24}. 
\end{align}

Now, we want to calculate $\Tr_{2,\ldots,n}\left(A_1\otimes \cdots \otimes A_n  V_d(\pi_{\text{cyc}})\right)$, and we will show that it simplifies to $A_1\cdots A_n$. We can verify this through direct calculation using resolutions of identities:
\begin{align}
    \Tr_{2,\ldots,n}\left(A_1\otimes \cdots \otimes A_k \, V_d(\pi_{\text{cyc}})\right)&=\sum_{i_1,\dots,i_k \in [d]} \Tr_{2,\cdots,k}\left(A_1\otimes \cdots \otimes A_k \,  \ketbra{i_2,\dots,i_{k},i_{1}}{i_1,i_2,\dots,i_k}\right) \nonumber \\
&=\sum_{i_1,\dots,i_k \in [d]} A_1 \ketbra{i_2}{i_1}\bra{i_2}A_2\ket{i_3}\bra{i_3}A_3\ket{i_4}\cdots \bra{i_k}A_k\ket{i_1} \nonumber \\
&=\sum_{i_1,i_2 \in [d]} A_1 \ketbra{i_2}{i_1}\bra{i_2}A_2 A_3\cdots A_k\ket{i_1}=A_1   \cdots A_k,
\end{align}
Alternatively, we can represent this calculation diagrammatically as:
 \begin{align}
    \ipic{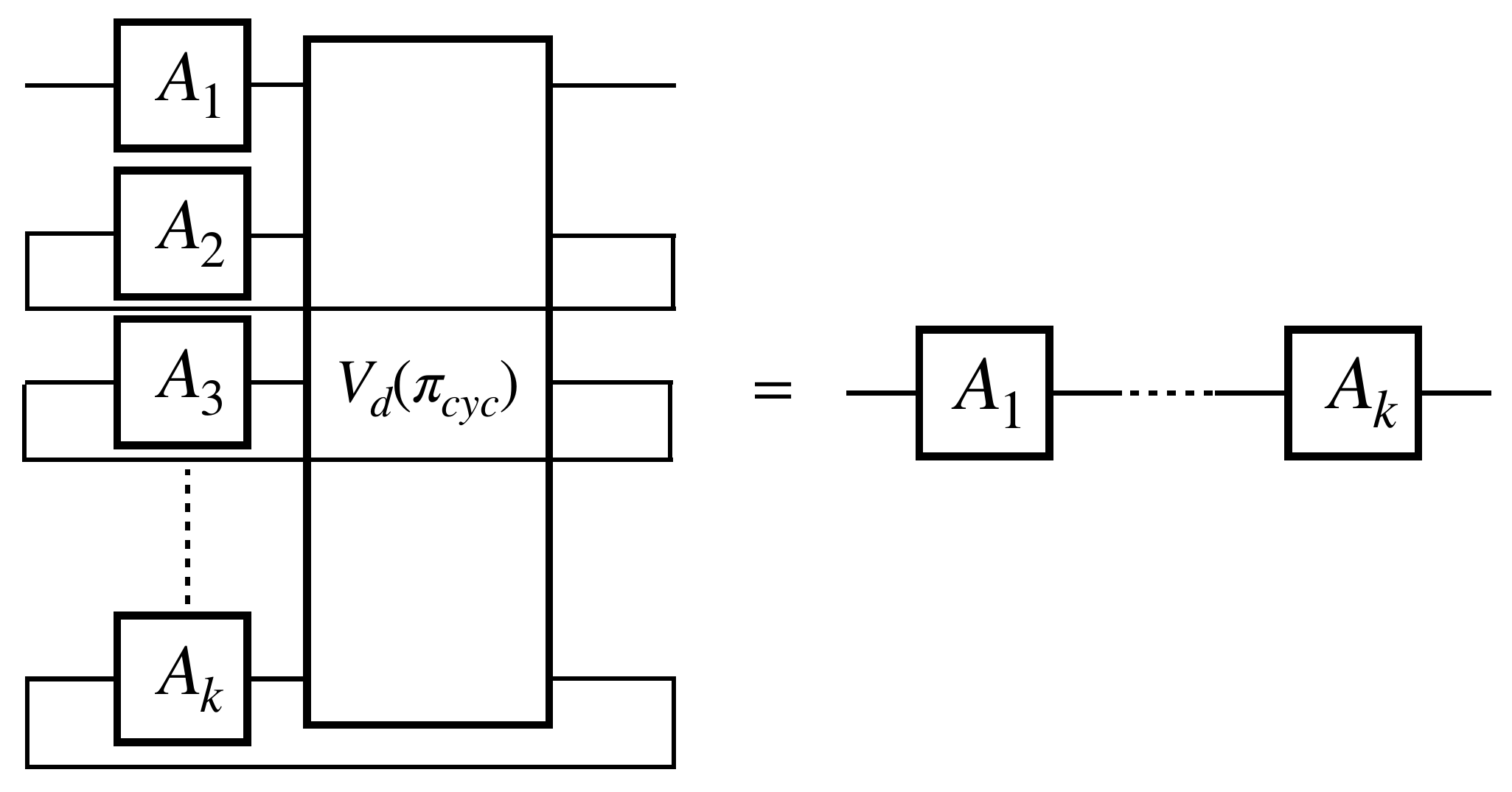}{0.24}.
\end{align}
Hence, we have the generalized \emph{swap-trick}, known as the \emph{cyclic-permutation-trick}:
\begin{align}
\Tr(A_1\otimes \cdots \otimes A_k \, V_d(\pi_{\text{cyc}}))= \Tr(A_1\cdots A_k).
\end{align}

\section{Unitary designs}
\label{sec:design}
Generating Haar random unitaries on a quantum computer can be a computationally expensive task, since most unitaries require an exponential number of elementary gates with respect to the number of qubits \cite{nielsen_chuang_2010,Brand_o_2021,Elben_2022} to be implemented. However, for many applications in quantum information, only low-order moments of the Haar measure are needed \cite{Huang_2020,DiVincenzo_2002,Bouland_2018,Roberts_2017,Hangleiter_2018,Bouland18,Pashayan_2020}. This motivates the definition of unitary $k$-designs \cite{Dankert_2009}, which are distributions of unitaries that match the moments of the Haar measure up to the $k$-th order, where $k\in \mathbb{N}$. As we will mention later, generating unitary $k$-designs can be done efficiently in the number of elementary gates. 
\begin{definition}[Unitary $k$-design]
\label{def:unitarydes}
   Let $\nu$ be a probability distribution defined over a set of unitaries $S\subseteq \Ug(d)$. 
The distribution $\nu$ is unitary $k$-design if and only if:
    \begin{align}
        \underset{V \sim \nu}{\mathbb{E}}\left[V^{\otimes k} O V^{\dagger \otimes k}\right]=\ExU \left[U^{\otimes k} O U^{\dagger \otimes k}\right],
        \label{eq:unitarydesign}
    \end{align}
    for all $O\in \mathcal{L}\left((\mathbb{C}^d)^{\otimes k}\right)$.
\end{definition}
For instance, consider a distribution $\nu$ where the set of unitaries $S$ is discrete and each unitary has an equal probability of being chosen. In this case, we have:
\begin{align}
\underset{V \sim \nu}{\mathbb{E}}\left[V^{\otimes k} O V^{\dagger \otimes k}\right]=\frac{1}{\left|S\right|}\sum_{V\in S}V^{\otimes k} O V^{\dagger \otimes k}.
\end{align}
An equivalent way to define a unitary $k$-design is to relate the vectorized moment operator of the distribution $\nu$ to that of the Haar measure $\mu_H$. Specifically, we have the following:
\begin{observation}
\label{obs:vecdesign}
A probability distribution $\nu$ is a unitary $k$-design if and only if:
    \begin{align}
    \underset{V \sim \nu}{\mathbb{E}}\left[V^{\otimes k}\otimes V^{*\otimes k}\right]=\ExU\left[U^{\otimes k}\otimes U^{*\otimes k}\right].
    \label{eq:vecdesign}
    \end{align}
\end{observation}
    \begin{proof}
        By taking the vectorization to both sides of Eq.\eqref{eq:unitarydesign}, we have:
    \begin{align}
        \underset{V \sim \nu}{\mathbb{E}}\left[V^{\otimes k}\otimes V^{*\otimes k}\right]\kket{O}=\ExU\left[U^{\otimes k}\otimes U^{*\otimes k}\right]\kket{O}.
    \end{align}
    Now Eq.\eqref{eq:vecdesign} follows by observing that for all $\ket{\psi}\in (\mathbb{C}^{d})^{\otimes 2k}$ there exists $O\in \mathcal{L}\left((\mathbb{C}^{d})^{\otimes k}\right) $ such that $\kket{O}=\ket{\psi}$. Similarly, the converse holds.
    \end{proof}

\begin{observation}
    If $\nu$ is a $(k+1)$-design, then $\nu$ is also a $k$-design.
\end{observation}
\begin{proof}
     Let $O=O^\prime\otimes \frac{I}{d} $ where $O^\prime \in \mathcal{L}\left((\mathbb{C}^d)^{\otimes k}\right)$. The claim follows by applying the Definition~\ref{def:unitarydes} with such defined $O$ and taking the partial trace with respect to the $(k+1)$-th tensor product space.
\end{proof}
It is also useful to define the following quantity, the frame potential \cite{Gross_2007}, which provides another way to verify if a probability distribution is a $k$-design.
\begin{definition}[Frame potential] Let $\nu$ be a probability distribution defined over the set of unitaries $S\subseteq \Ug(d)$. 
For a given $k\in \mathbb{N}$, we define the $k$-frame potential, denoted as $\mathcal{F}^{(k)}_{\nu}$, as follows:
\begin{align}
    \mathcal{F}^{(k)}_{\nu}\coloneqq \underset{U,V \sim \nu}{\mathbb{E}}\left[ \left|\Tr\!\left(UV^\dagger\right)\right|^{2k}\right].
\end{align}
\end{definition}
As we will show in Proposition~\ref{prop:framedesign}, a probability distribution $\nu$ is a $k$-design if and only if its $k$-frame potential coincides with that of the Haar measure $\mu_H$. Thus, given a probability distribution $\nu$ and a set of unitaries $S$, we can compute the frame potential $\mathcal{F}^{(k)}_{\nu}$ and compare it with that of the Haar measure to determine whether $\nu$ is a $k$-design or not.
It is also useful to define the notion of $k$-invariant distribution.
\begin{definition}[$k$-invariant measure]
Let $\nu$ be a probability distribution defined over a set of unitaries $S\subseteq \Ug(d)$. $\nu$ is $k$-invariant if and only if, for any polynomial $p(U)$ of degree $\le k$ in the matrix elements of $U$ and $U^*$, it holds
\begin{align}
\underset{U \sim \nu}{\mathbb{E}} \left[\, p\left(U\right) \right]=\underset{U \sim \nu}{\mathbb{E}}\left[\, p\left(UV\right) \right]=\underset{U \sim \nu}{\mathbb{E}}\left[\, p\left(VU\right)\right],
\end{align}
for all $V \in S$.
\end{definition}
We begin by making the following observation, which will be useful for computing the frame potential of the Haar measure $\mu_H$.
\begin{lemma}
\label{le:framComm} 
Let $\nu$ be a probability distribution defined over a set of unitaries $S\subseteq \Ug(d)$. If $\nu$ is $k$-invariant, then we have:
\begin{align}
\mathcal{F}^{(k)}_{\nu}=\dim\left(\mathrm{Comm}(S,k)\right),
\end{align}
where $\mathrm{Comm}(S,k)$ is defined in Definition~\ref{def:comm}.
\end{lemma}
 
\begin{proof}
By (right) $k$-invariance we have:
\begin{align}
    \mathcal{F}^{(k)}_{\nu}=\underset{U,V \sim \nu}{\mathbb{E}}\left[ \left|\Tr\!\left(UV^\dagger\right)\right|^{2k}\right]=\underset{U \sim \nu}{\mathbb{E}}\left[ \left|\Tr\!\left(U\right)\right|^{2k}\right].
\end{align}
Therefore:
    \begin{align}
    \mathcal{F}^{(k)}_{\nu}=\underset{U \sim \nu}{\mathbb{E}}\left[ \left|\Tr\!\left(U\right)\right|^{2k}\right]=\underset{U \sim \nu}{\mathbb{E}}\left[ \Tr\!\left(U\right)^{k} \Tr\!\left(U^*\right)^{k}\right] =\Tr\!\left( \underset{U \sim \nu}{\mathbb{E}} \left[U^{\otimes k} \otimes U^{*\otimes k}\right]\right).
\end{align}
We can conclude the proof by observing that $\Tr\!\left( \underset{U \sim \nu}{\mathbb{E}} \left[U^{\otimes k} \otimes U^{*\otimes k}\right]\right)=\dim\left(\mathrm{Comm}(S,k)\right)$, according to Proposition~\ref{prop:MomVec}. Note that although the proof of Proposition~\ref{prop:MomVec} was given for the Haar measure on the unitary group, it applies equally to any $k$-invariant measure.

\end{proof}
It is worth noting that the left and right invariance properties, and therefore Lemma~\ref{le:framComm}, hold for any uniform probability distribution $\nu$ defined over a subgroup of unitaries $S$.

In the following lemma, we show that the difference between the frame potential of a probability distribution $\nu$ (not necessarily $k$-invariant) and that of the Haar measure $\mu_H$ can be represented by the squared $2$-norm of the difference between their vectorized moment operator.
\begin{lemma}[Frame potential difference]
Let $\mathcal{F}^{(k)}_{\nu}$ and $\mathcal{F}^{(k)}_{\mu_H}$ be the frame potentials of the probability distribution $\nu$ and the Haar measure $\mu_H$, respectively. Then, we have:
\label{le:frameNorm}
\begin{align}
    \mathcal{F}^{(k)}_{\nu}-\mathcal{F}^{(k)}_{\mu_H}=\norm{\underset{V \sim \nu}{\mathbb{E}}\left[V^{\otimes k}\otimes V^{*\otimes k}\right]-\ExU\left[U^{\otimes k}\otimes U^{*\otimes k}\right]}^2_2.
\end{align}
\end{lemma}
\begin{proof}
Let $\mathrm{M}^{(k)}_{\mu_H}\coloneqq \ExU\left[U^{\otimes k}\otimes U^{*\otimes k}\right]$ and $\mathrm{M}^{(k)}_\nu\coloneqq \underset{V \sim \nu}{\mathbb{E}}\left[U^{\otimes k}\otimes U^{*\otimes k}\right]$.

First, we note that $\mathrm{M}^{(k)\dagger}_{\mu_H}=\mathrm{M}^{(k)}_{\mu_H}$  due to Proposition~\ref{prop:Haar}. Furthermore, because of the invariance of the Haar measure, we have $\mathrm{M}^{(k)2}_{\mu_H}=\mathrm{M}^{(k)}_{\mu_H}$ and $\mathrm{M}^{(k)}_\nu\mathrm{M}^{(k)}_{\mu_H}=\mathrm{M}^{(k)}_\nu\mathrm{M}^{(k)\dagger}_{\mu_H}=\mathrm{M}^{(k)}_{\mu_H}$. Thus, by using these properties, we obtain:
\begin{align}
    \norm{\mathrm{M}^{(k)}_\nu-\mathrm{M}^{(k)}_{\mu_H}}^2_2&=\Tr\left[\left(\mathrm{M}^{(k)}_\nu-\mathrm{M}^{(k)}_{\mu_H}\right)^\dagger\left(\mathrm{M}^{(k)}_\nu-\mathrm{M}^{(k)}_{\mu_H}\right)\right]\\
    &=\Tr\left[\mathrm{M}^{(k)\dagger}_\nu \mathrm{M}^{(k)}_\nu\right] + \Tr\left[\mathrm{M}^{(k)\dagger}_{\mu_H} \mathrm{M}^{(k)}_{\mu_H}\right] - \Tr\left[\mathrm{M}^{(k)\dagger}_\nu \mathrm{M}^{(k)}_{\mu_H}\right] - \Tr\left[\mathrm{M}^{(k)\dagger}_{\mu_H} \mathrm{M}^{(k)}_\nu \right]\\
    &=\Tr\left[\mathrm{M}^{(k)\dagger}_\nu \mathrm{M}^{(k)}_\nu\right] - \Tr\left[\mathrm{M}^{(k)}_{\mu_H}\right]\\
    &=\Tr\left[\mathrm{M}^{(k)}_\nu \mathrm{M}^{(k)\dagger}_\nu\right] - \Tr\left[\mathrm{M}^{(k)}_{\mu_H}\mathrm{M}^{(k)\dagger}_{\mu_H}\right].
\end{align}
We can conclude by observing that for $\mu=\mu_H$ and $\mu=\nu$ we have:
\begin{align}
\Tr\left[\mathrm{M}^{(k)}_\mu \mathrm{M}^{(k)\dagger}_\mu \right]=\underset{V,U \sim \mu}{\mathbb{E}}\Tr\left[(UV^\dagger)^{\otimes k}\otimes(UV^\dagger)^{*\otimes k}\right]=\underset{V,U \sim \mu}{\mathbb{E}}\left[\left|\Tr\!\left(UV^\dagger\right)\right|^{2k}\right]=\mathcal{F}^{(k)}_{\mu}.
\end{align}
\end{proof}
It follows from the previous Lemma that showing that a distribution $\nu$ is a $k$-design can be achieved by computing its frame potential and comparing it with that of the Haar measure $\mu_H$. This is due to the fact that if the two frame potentials coincide, the difference between their vectorized moment operators in $2$-norm must be zero, hence they must be equal (which is an equivalent unitary design definition because of Observation~\ref{obs:vecdesign}).
We can restate the result explicitly as follows:
\begin{proposition}[Frame potential $k$-design condition]
\label{prop:framedesign}
    We have:
    \begin{align}
        \mathcal{F}^{(k)}_{\nu}\ge\mathcal{F}^{(k)}_{\mu_H}=\dim(\operatorname{span} (V_d(\pi):\, \pi \in S_k)).
    \end{align}
Moreover, the equality holds if and only if $\nu$ is a $k$-unitary design. 

In particular, if $k\le d$, then $\dim(\operatorname{span} (V_d(\pi):\, \pi \in S_k))=k!$.
\end{proposition}
\begin{proof}
By Lemma~\ref{le:frameNorm}, Lemma~\ref{le:framComm} and Schur-Weyl duality~\ref{th:SchurW}, we have:
\begin{align}
\mathcal{F}^{(k)}_{\nu}\ge\mathcal{F}^{(k)}_{\mu_H}=\dim\left(\mathrm{Comm}(\Ug(d),k)\right)=\dim\left(\operatorname{span} (V_d(\pi):\, \pi \in S_k)\right).
\end{align}
Moreover, by Proposition~\ref{prop:permInd}, we have that if $k\le d$, then the number of linearly independent permutation matrices is $k!$.
\end{proof}
By utilizing this result, we can derive a straightforward lower bound on the cardinality of a discrete set $S$ of unitaries necessary to form a $k$-design. To establish this bound, let us consider the frame potential $\mathcal{F}^{(k)}_{\nu}$ associated with the uniform distribution $\nu$ over the discrete set $S$. We can then deduce that:
\begin{align}
    \mathcal{F}^{(k)}_{\nu}=\frac{1}{\left|S\right|^2}\sum^{|S|}_{i,j=1} \left|\Tr\!\left(U_iU_j^\dagger\right)\right|^{2k}\ge\frac{1}{\left|S\right|^2}\sum^{|S|}_{i=1} \left|\Tr\!\left(U_iU_i^\dagger\right)\right|^{2k}=\frac{1}{\left|S\right|} d^{2k}.
\end{align}
Furthermore, considering the fact that $\nu$ constitutes a $k$-design (with $k \le d$), by the previous proposition, we have that $\mathcal{F}^{(k)}_{\nu}=k!$. This implies that the cardinality of the set $S$ must satisfy $|S|\ge \frac{d^{2k}}{k!}$. Consequently, if we are considering an $n$-qubit system where the Hilbert space dimension is $d=2^n$, the cardinality of $S$ must grow at least exponentially with the number of qubits. For a more comprehensive analysis of lower bounds of $k$-design, we recommend referring to the following references \cite{Brand_o_2016,Roy_2009,Gross_2007,Brand_o_2021}.  

The following proposition provides equivalent definitions of unitary $k$-design:
\begin{proposition}[Equivalent definitions of unitary $k$-design.]
    Let $\nu$ be a probability distribution over a set of unitaries $S\subseteq \Ug$. Then, $\nu$ is a unitary $k$-design if and only if:
    \begin{enumerate}
        \item \quad $\underset{U \sim \nu}{\mathbb{E}}\left[U^{\otimes k} O U^{\dagger \otimes k}\right]=\ExU \left[U^{\otimes k} O U^{\dagger \otimes k}\right]$ for all $O\in \mathcal{L}\left((\mathbb{C}^d)^{\otimes k}\right)$. 
        \item \quad $\underset{V \sim \nu}{\mathbb{E}}\left[V^{\otimes k}\otimes V^{*\otimes k}\right]=\ExU\left[U^{\otimes k}\otimes U^{*\otimes k}\right]$.
        \item \quad $\mathcal{F}^{(k)}_{\nu}=\dim\left(\mathrm{Comm}(\Ug(d),k)\right)$.
        \item \quad $\underset{V \sim \nu}{\mathbb{E}}[\,p(V)\,]=\ExU [\,p(U)\,]$ for all polynomials $p(U)$ homogeneous of degree $k$ in the matrix elements of $U$ and homogeneous of degree $k$ in the matrix elements of $U^*$.
    \end{enumerate}
\end{proposition}
\begin{proof}
    The equivalence between $1.$ and $2.$ is shown in Observation~\ref{obs:vecdesign} and the one between $2.$ and $3.$ is shown in Proposition~\ref{prop:framedesign}. 
    To show that $2.$ implies $4.$, we observe that any homogeneous polynomial $p(U)$ of degree $k$ in the elements of $U$ and in the elements of $U^*$ can be written as $p(U)=\Tr[A U^{\otimes k}\otimes U^{* \otimes k}]$ for a matrix $A\in \mathcal{L}\left((C^d)^{\otimes 2 k}\right)$ with entries the coefficients of the polynomial. Similarly, to show that $4.$ implies $2.$, we note that, for all matrices $A\in \mathcal{L}\left((C^d)^{\otimes 2 k}\right)$, $p_A(U)\coloneqq \Tr[A U^{\otimes k}\otimes U^{* \otimes k}]$ is a homogeneous polynomial of degree $k$ in both the elements of $U$ and $U^*$. Therefore we have \begin{align}
    \Tr\!\left(A\underset{V \sim \nu}{\mathbb{E}}\left[V^{\otimes k}\otimes V^{*\otimes k}\right]\right)=\Tr\!\left(A\ExU \left[U^{\otimes k}\otimes U^{*\otimes k}\right]\right) \quad \text{for all $A\in \mathcal{L}\left((C^d)^{\otimes 2 k}\right)$,}        
\end{align}
    which implies the point $2$. 
\end{proof}
Moreover, if $\nu$ is a uniform distribution over a set of unitaries $S$ which forms a group, then if $\dim\left(\mathrm{Comm}(S,k)\right)=\dim\left(\mathrm{Comm}(\Ug(d),k)\right)$, then, by Proposition~\ref{le:framComm} and Proposition~\ref{prop:framedesign}, $\nu$ is a $k$-unitary design.

Now we turn to the definition of $k$-designs for distributions over sets of states, which are known as \emph{state} $k$-designs or \emph{spherical} $k$-designs \cite{ambainis2007quantum}.
\begin{definition}[State k-design]
   Let $\eta$ be a probability distribution over a set of states $S\subseteq \mathbb{C}^d$. 
   
   The distribution $\eta$ is said to be a state $k$-design (or also spherical $k$-designs) if and only if:
    \begin{align}
        \underset{\ket{\psi} \sim \eta}{\mathbb{E}}\left[\ketbra{\psi}{\psi}^{\otimes k}\right]=\Expsi \left[\ketbra{\psi}{\psi}^{\otimes k}\right].
    \end{align}
\end{definition}
A unitary $k$-design $\nu$ induces a $k$-state design for the probability distribution over states $U\ket{\psi_0}$, where $U$ is drawn from $\nu$ and $\ket{\psi_0}$ is fixed. This can be seen by setting $O = (\ketbra{\psi_0}{\psi_0})^{\otimes k}$ in Definition~\ref{def:unitarydes}.

Similarly to the unitary $k$-design case, we can define the so-called \emph{state frame potential}.
\begin{definition}[State frame potential]
    The state frame potential is defined as: 
    \begin{align}
        \mathcal{F}^{(k, \eta)}_{state}\coloneqq \underset{\ket{\psi},\ket{\phi}  \sim \eta}{\mathbb{E}}\left[\left|\braket{\psi}{\phi}\right|^{2k}\right].
    \end{align}
\end{definition}
We now state a proposition that relates the state frame potential to $k$-designs:
\begin{proposition}[State frame potential design] Let $\eta$ be a probability distribution over a set of states $S\subseteq \mathbb{C}^d$. Then:
    \begin{align}
        \mathcal{F}^{(k, \eta)}_{state}\ge \left(\mathrm{dim}\left(\mathrm{Sym}_k(\mathbb{C}^d)\right)\right)^{-1},
    \end{align}
    where $\mathrm{dim}\left(\mathrm{Sym}_k(\mathbb{C}^d)\right)=\binom{d+k-1}{k}$. Equality holds if and only if $\eta$ is a $k$-state design.
\end{proposition}
\begin{proof}
We denote $\mathrm{A}^{(k)}_{\mu_H}\coloneqq  \underset{\ket{\psi} \sim \mu_H}{\mathbb{E}}\left[\ketbra{\psi}{\psi}^{\otimes k}\right]$ and $\mathrm{A}^{(k)}_\eta\coloneqq  \underset{\ket{\psi} \sim \eta}{\mathbb{E}}\left[\ketbra{\psi}{\psi}^{\otimes k}\right]$. We have:
\begin{align}
    \norm{\mathrm{A}^{(k)}_\eta-\mathrm{A}^{(k)}_{\mu_H}}^2_2&=\Tr\left[\left(\mathrm{A}^{(k)}_\eta-\mathrm{A}^{(k)}_{\mu_H}\right)^2\right]\\
    &=\Tr\left[\left(\mathrm{A}^{(k)}_\eta\right)^2\right] + \Tr\left[\left(\mathrm{A}^{(k)}_{\mu_H}\right)^2\right] - 2\Tr\left[\mathrm{A}^{(k)}_\eta \mathrm{A}^{(k)}_{\mu_H}\right] \\
    &=\underset{\ket{\psi},\ket{\phi}  \sim \eta}{\mathbb{E}}\left[\left|\braket{\psi}{\phi}\right|^{2k}\right] + \Tr\left[\left(\mathrm{A}^{(k)}_{\mu_H}\right)^2\right]- 2\Tr\left[\mathrm{A}^{(k)}_\eta \mathrm{A}^{(k)}_{\mu_H}\right],
\end{align}
where in the second equality we used cyclicity of the trace.
Additionally, we have:
\begin{align}
\Tr\left[\left(\mathrm{A}^{(k)}_{\mu_H}\right)^2\right]
&=\Tr\left[\left(\frac{P^{(d,k)}_{\mathrm{sym}}}{\mathrm{dim}\left(\mathrm{Sym}_k(\mathbb{C}^d)\right)}\right)^2\right]\\
&=\Tr\left[\frac{P^{(d,k)}_{\mathrm{sym}}}{\left(\mathrm{dim}\left(\mathrm{Sym}_k(\mathbb{C}^d)\right)\right)^2}\right]\\
&=\left(\mathrm{dim}\left(\mathrm{Sym}_k(\mathbb{C}^d)\right)\right)^{-1},
\end{align}
where in the first equality we used Eq.\eqref{eq:MomPermSym}, in the second equality the fact that $P^{(d,k)}_{\mathrm{sym}}$ squares to itself (Theorem~\ref{th:Psymprojector}) and in the last equality the fact that $P^{(d,k)}_{\mathrm{sym}}$ is the orthogonal projector on the symmetric subspace (Theorem~\ref{th:Psymprojector}).
Similarly, we can show that $\Tr\left[\mathrm{A}^{(k)}_\eta \mathrm{A}^{(k)}_{\mu_H}\right]$ is equal to $\left(\mathrm{dim}\left(\mathrm{Sym}_k(\mathbb{C}^d)\right)\right)^{-1}$. Here, $\mathrm{dim}\left(\mathrm{Sym}_k(\mathbb{C}^d)\right)$ can be calculated using Proposition~\ref{th:dimsym}, which shows that it is equal to $\binom{d+k-1}{k}$. Therefore, we can conclude.

\end{proof}

It is worth noting that any uniform distribution $\nu$ defined over a set of unitaries $S=\{U_i\}^{d^2}_{i=1}$ that forms a basis for $\mathcal{L}(\mathbb{C}^{d})$ and satisfies $\Tr(U_i^\dagger U_j)=d\delta_{i,j}$ constitutes a $1$-design. This can be easily proven by computing the frame potential as follows:
\begin{align}
    \mathcal{F}_{\nu}^{(k=1)}=\frac{1}{|S|^2}\sum^{d^2}_{i,j=1}\left|\Tr(U_iU_j^{\dagger})\right|^2=\frac{1}{d^4}\sum^{d^2}_{i,j=1}d^2\delta_{i,j}=1=\dim(\operatorname{span} (V_d(\pi):\, \pi \in S_1)),
\end{align}
and applying Proposition \ref{prop:framedesign}.
Therefore, the uniform distribution defined over the Pauli basis $\tilde{\mathcal{P}}:=\{I,X,Y,Z\}^{\otimes n}$ is a $1$-design, where $n$ is the number of qubits and $d=2^n$. 

The Flip operator can be elegantly represented in terms of the Pauli basis using the following expression:
\begin{align}
    \Flip=\sum_{P,Q\in \tilde{\mathcal{P}}} \frac{1}{d^2} \Tr(\left(P\otimes Q\right) \Flip) P\otimes Q=\sum_{P,Q\in \tilde{\mathcal{P}}} \frac{1}{d^2} \Tr(P Q) P\otimes Q=\sum_{P\in \tilde{\mathcal{P}}} \frac{1}{d} P\otimes P,
\end{align}
where we wrote the Flip operator in the Pauli basis and used the \emph{swap-trick}.
Using this, it is also evident that the Pauli group forms a $1$-design (according to Definition \ref{def:unitarydes}). In fact, for any $O\in \MatC{d}$, we have:
\begin{align}
    \frac{1}{|\tilde{\mathcal{P}}|}\sum_{P\in \tilde{\mathcal{P}}} P O P^{\dagger}=\frac{1}{d^2}\sum_{P\in \tilde{\mathcal{P}}}\Tr_2\left((I\otimes O) ( P\otimes P) \Flip\right)=\frac{1}{d}\Tr_2\left((I\otimes O) \Flip\Flip\right)=\frac{\Tr(O)}{d}=\ExU \left[U O U^{\dagger}\right],
    \nonumber
\end{align}
where in the first step we used the \emph{partial-swap-trick} (Eq.\eqref{eq:partialswap}) and in the last step (Eq.\eqref{eq:1momHaar}).

An important set of unitaries is the Clifford group $\mathrm{Cl}(n)$ \cite{gottesman1998heisenberg} i.e. the set of unitaries which sends the Pauli group $\mathcal{P}_n$ in itself under the adjoint operation:
\begin{align}
    \mathrm{Cl}(n)\coloneqq \{U\in \mathrm{U}(2^n)\colon UPU^{\dagger} \in \mathcal{P}_n \text{ for all $P\in \mathcal{P}_n$} \},
    \label{eq:Cliffordgroup}
\end{align}
where $\mathcal{P}_n\coloneqq \{i^k\}^{3}_{k=0}\times \{I,X,Y,Z\}^{\otimes n}$.
It can be proven that the uniform distribution over the Clifford group, forms a $3$-design for all $d=2^n$ \cite{webb2016clifford,Zhu_2017}, but it fails to be a $4$-design \cite{zhu2016clifford}. 

Moreover, it can be shown that any Clifford circuit can be implemented with $O(n^2/\log(n))$ gates from the set $\{\Hadamard,\CNOT,\PhaseS\}$ where $\Hadamard$, $\CNOT$ and $\PhaseS$ are the Hadamard, Controlled-NOT and Phase gate, respectively \cite{Aaronson_2004,nielsen_chuang_2010}. Furthermore, there are efficient algorithms to sample uniformly from the Clifford group \cite{berg2021simple,Koenig_2014}.  
This is of fundamental importance for applications in quantum computing, since if we are interested in reproducing moments of the Haar measure up to the third moment, it is sufficient to sample efficient quantum circuits that correspond to Clifford unitaries (instead of sampling from the Haar measure, which would require implementing quantum circuits with an exponential number of gates in the number of qubits \cite{nielsen_chuang_2010}).

\section{Approximate unitary designs}
\label{sec:apprdesign}
In many cases, having an exact unitary design may not be necessary, and having an \emph{approximate} one may suffice.
Various definitions of approximate unitary designs have been proposed in the literature \cite{low2010pseudorandomness,Brand_o_2016,Brand_o_2021}. Here we will explore some of them and their relationships.

We begin by defining the concept of Tensor Product Expander (TPE)-$\varepsilon$-approximate $k$-design \cite{Brand_o_2016}. To simplify the notation, we use $U^{\otimes k,k}\coloneqq U^{\otimes k}\otimes U^{*\otimes k}$.
\begin{definition}[Tensor Product Expander (TPE)-$\varepsilon$-approximate $k$-design]
    Let $\varepsilon>0$.
    We say that $\nu$ is a TPE $\varepsilon$-approximate $k$-design if and only if:
    \begin{align}
        \norm{\underset{V \sim \nu}{\mathbb{E}}\left[V^{\otimes k,k}\right]-\ExU \left[U^{\otimes k,k}\right]}_\infty\le \varepsilon.
    \end{align}
\end{definition}

The TPE notion of approximate unitary design is particularly advantageous because it can be \textquote{amplified}. This means that if a distribution of unitaries $\nu$ is a TPE $\lambda$-approximate $k$-design, then the distribution of unitaries associated to the product $U_\Pmath \cdots U_1$ of unitaries $\{U_i\}^\Pmath_{i=1}$, each independently distributed according to $\nu$, is a TPE $\lambda^\Pmath$-approximate design (as we prove below). We denote by $\nu_\Pmath$ this distribution of the product of $\Pmath$ unitaries sampled independently by $\nu$. Note that the vectorized moment operator associated with $\nu_\Pmath$ is:
\begin{align}
    \underset{U \sim \nu_\Pmath}{\mathbb{E}}U^{\otimes k,k}=\underset{U_1,\dots,U_\Pmath \sim \nu}{\mathbb{E}}\left(U_\Pmath \cdots U_1\right)^{\otimes k,k}.
\end{align}
We state the proposition below:
\begin{proposition}[Amplification of TPE]\
Let $1 >\lambda \ge \varepsilon > 0 $.
If $\nu$ is a TPE $\lambda$-approximate $k$-design, then $\nu_\Pmath$ is a TPE $\varepsilon$-approximate $k$-design, where $\Pmath$ is a positive integer such that $\Pmath\ge \frac{1}{\log(\lambda^{-1})}\log(\varepsilon^{-1}) $.
\end{proposition}
\begin{proof}
We have to bound the following quantity $
        \norm{\underset{V \sim \nu_\Pmath}{\mathbb{E}}\left[V^{\otimes k,k}\right]-\ExU \left[U^{\otimes k,k}\right]}_\infty$.
We first observe that the vectorized moment operator associated to $\nu_\Pmath$ is:
\begin{align}
    \label{eq:productmeas}
     \underset{V \sim \nu_\Pmath}{\mathbb{E}}\left[V^{\otimes k,k}\right]=\underset{V_1,\dots,V_\Pmath \sim \nu}{\mathbb{E}}\left(V_\Pmath \cdots V_1\right)^{\otimes k,k}=\left(\underset{V \sim \nu}{\mathbb{E}}\left[V^{\otimes k,k}\right]\right)^{\Pmath},
\end{align}
where we used the independence of $U_1,\dots,U_\Pmath$.
Moreover, we have:
\begin{align}
    \label{eq:amplif}
    \left(\underset{V \sim \nu}{\mathbb{E}}\left[V^{\otimes k,k}\right]\right)^{\Pmath}-\ExU \left[U^{\otimes k,k}\right]=\left(\underset{V \sim \nu}{\mathbb{E}}\left[V^{\otimes k,k}\right]-\ExU \left[U^{\otimes k,k}\right]\right)^{\Pmath},
\end{align}
which can be shown by induction. The basis step is trivial for $\Pmath=1$. Suppose the claim is true for $\Pmath$, then:
\begin{align}
    &\left(\underset{V \sim \nu}{\mathbb{E}}\left[V^{\otimes k,k}\right]-\ExU \left[U^{\otimes k,k}\right]\right)^{\Pmath+1}\\
    &=\left(\underset{V \sim \nu}{\mathbb{E}}\left[V^{\otimes k,k}\right]-\ExU \left[U^{\otimes k,k}\right]\right)^{\Pmath}\left(\underset{V \sim \nu}{\mathbb{E}}\left[V^{\otimes k,k}\right]-\ExU \left[U^{\otimes k,k}\right]\right)\\
    &=\left(\left(\underset{V \sim \nu}{\mathbb{E}}\left[V^{\otimes k,k}\right]\right)^{\Pmath}-\ExU \left[U^{\otimes k,k}\right]\right)\left(\underset{V \sim \nu}{\mathbb{E}}\left[V^{\otimes k,k}\right]-\ExU \left[U^{\otimes k,k}\right]\right)\\
    &=\left(\underset{V \sim \nu}{\mathbb{E}}\left[V^{\otimes k,k}\right]\right)^{\Pmath+1}-\ExU \left[U^{\otimes k,k}\right],
\end{align}
where in the last step we used the invariance of the Haar measure to conclude that 
\begin{align}
V^{k,k}\ExU \left[U^{\otimes k,k}\right]=\ExU \left[U^{\otimes k,k}\right] V^{k,k}=\ExU \left[U^{\otimes k,k}\right].
\end{align}
Therefore we have:
    \begin{align}
        \norm{\underset{V \sim \nu_\Pmath}{\mathbb{E}}\left[V^{\otimes k,k}\right]-\ExU \left[U^{\otimes k,k}\right]}_\infty&=\norm{\left(\underset{V \sim \nu}{\mathbb{E}}\left[V^{\otimes k,k}\right]-\ExU \left[U^{\otimes k,k}\right]\right)^{\Pmath}}_\infty\\
        &\le\norm{\underset{V \sim \nu}{\mathbb{E}}\left[V^{\otimes k,k}\right]-\ExU \left[U^{\otimes k,k}\right]}_\infty^{\Pmath}\\
        &\le \lambda^\Pmath,
    \end{align}
    where in the first step we used Eq.\eqref{eq:productmeas} and Eq.\eqref{eq:amplif}, in the second step we used the submultiplicativity\footnote{Note that the submultiplicativity inequality is saturated if $\underset{V \sim \nu}{\mathbb{E}}\left[V^{\otimes k,k}\right]$ is Hermitian (or more generally normal).}  of the infinity norm and in the last step we used that $\nu$ is a TPE $\lambda$-approximate $k$-design. By choosing $\Pmath\ge \frac{\log(\varepsilon)}{\log(\lambda)}$, we can conclude that $\nu_\Pmath$ is a TPE $\varepsilon$-approximate $k$-design.
\end{proof}

The previous Proposition implies that if we have a TPE $\lambda$-approximate $k$-design and we want to amplify it in order to achieve a TPE with approximation precision $\varepsilon$ inverse exponential in the number of qubits i.e. $\varepsilon=\tilde{\varepsilon} d^{-c}$ where $\tilde{\varepsilon}, c > 0$ and $d=2^n$, then we can achieve this by choosing a number of repetitions $\Pmath\ge \frac{1}{\log(\lambda^{-1})}\left(\log(\tilde{\varepsilon}^{-1})+ n c \log(2)\right)$. Note that this expression is linear in the number of qubits $n$ if $\lambda=O(1)$.
 
We will now define the notion of diamond $\varepsilon$-approximate $k$-design.
\begin{definition}[Diamond $\varepsilon$-approximate $k$-design]
We say that $\nu$ is a diamond $\varepsilon$-approximate $k$-design if and only if
\begin{align}
    \norm{\mathcal{M}^{(k)}_{\nu}-\mathcal{M}^{(k)}_{\mu_H}}_{\lozenge}\le \varepsilon,
\end{align}
where the diamond norm of a superoperator $\Phi:\mathcal{L}(\mathbb{C}^{D})\rightarrow \mathcal{L}(\mathbb{C}^{D})$ can be defined as 
\begin{align}    
\norm{\Phi}_{\lozenge}\coloneqq \sup_{X\neq 0}\frac{\norm{\Phi\otimes\mathcal{I}\left(X\right)}_1}{\norm{X}_1},
\end{align} 
where $\mathcal{I}:\mathcal{L}(\mathbb{C}^{D})\rightarrow \mathcal{L}(\mathbb{C}^{D})$ is the identity channel.
\end{definition}
It is worth mentioning that the diamond norm distance between two quantum channels has an important operational interpretation, as it is closely related to the one-shot distinguishability probability between the two channels \cite{wilde_2013}.

In the following Proposition, we will see the relation between the Diamond approximate design definition and the TPE one previously introduced.

\begin{proposition}[Diamond vs TPE]\
\begin{itemize}
    \item We have:
    \begin{align}
    \label{eq:diamondTPE}
        \norm{\mathcal{M}^{(k)}_{\nu}-\mathcal{M}^{(k)}_{\mu_H}}_{\lozenge}  \le d^k \norm{\underset{V \sim \nu}{\mathbb{E}}\left[V^{\otimes k,k}\right]-\ExU \left[U^{\otimes k,k}\right]}_\infty.
    \end{align}
    Thus, if $\nu$ is a TPE $\varepsilon$-approximate $k$-design, then $\nu$ is a diamond $\varepsilon d^k $-approximate $k$-design.
    \item Conversely, we have: 
    \begin{align}
    \label{eq:TPEdiamond}
    \norm{\underset{V \sim \nu}{\mathbb{E}}\left[V^{\otimes k,k}\right]-\ExU \left[U^{\otimes k,k}\right]}_\infty \le d^{k/2} \norm{\mathcal{M}^{(k)}_{\nu}-\mathcal{M}^{(k)}_{\mu_H}}_{\lozenge}.
\end{align}
    Hence, if $\nu$ is a diamond $\varepsilon$-approximate $k$-design, then $\nu$ is a TPE $\varepsilon d^{k/2}$-approximate $k$-design. 
\end{itemize}
\end{proposition}
\begin{proof}
   Let $\Phi\coloneqq \mathcal{M}^{(k)}_{\nu}-\mathcal{M}^{(k)}_{\mu_H}$. We have:
\begin{align}
    \norm{\Phi}_{\lozenge}&=\sup_{X\neq 0}\frac{\norm{\Phi\otimes\mathcal{I}\left(X\right)}_1}{\norm{X}_1}\le \sqrt{d^{2k}}\sup_{X\neq 0}\frac{\norm{\Phi\otimes\mathcal{I}\left(X\right)}_2}{\norm{X}_1}\le d^k \sup_{X\neq 0}\frac{\norm{\Phi\otimes\mathcal{I}\left(X\right)}_2}{\norm{X}_2},
\end{align}
where we used that $\norm{A}_1\le\sqrt{D}\norm{A}_2$ and $\norm{A}_2 \le \norm{A}_1$ for $A\in \mathcal{L}(\mathbb{C}^D)$.
Now we observe that:
\begin{align}
    \sup_{X\neq 0}\frac{\norm{\Phi\otimes\mathcal{I}\left(X\right)}_2}{\norm{X}_2}=\sup_{X\,:\,\norm{X}_2=1}\norm{\mathrm{vec}\!\left(\Phi\otimes\mathcal{I}\right)\kket{X}}_2=\norm{\mathrm{vec}\!\left(\Phi\otimes\mathcal{I}\right)}_\infty,
\end{align}
where in the first equality we used that the Hilbert-Schmidt norm of a matrix is the $2$-norm of the vectorized matrix and that $\mathrm{vec}\left(\Phi\otimes\mathcal{I}\left(X\right)\right)=\mathrm{vec}\!\left(\Phi\otimes\mathcal{I}\right)\kket{X}$ (see Eq.\eqref{eq:phikraus}), while in the third equality we used that  $\norm{A\ket{v}}_2=\sqrt{\bra{v}A^{\dagger} A \ket{v} }$ and that the supremum of this quantity over all the vectors $\ket{v}$ that are normalized with respect to the $2$-norm is achieved by the largest eigenvalue of $A^{\dagger} A$, which coincides with the largest singular value of $A$.
Now we have: \begin{align}
\norm{\mathrm{vec}\!\left(U^{\otimes k}\left(\cdot\right)U^{\dagger\otimes k}\otimes\mathcal{I}\right)}_\infty=\norm{U^{\otimes k}\otimes I^{\otimes k}\otimes U^{*\otimes k}\otimes I^{\otimes k}}_\infty=\norm{U^{\otimes k}\otimes U^{*\otimes k}\otimes I^{\otimes k}\otimes I^{\otimes k}}_\infty,    
\end{align}
where in the last step we permuted the tensor products using the $p$-norm property $\norm{UAU^\dagger}_p=\norm{A}_p$ with $U$ unitary matrix.
Using this we can write:
\begin{align}
    \norm{\mathcal{M}^{(k)}_{\nu}-\mathcal{M}^{(k)}_{\mu_H}}_{\lozenge}&\le d^k\norm{\mathrm{vec}\!\left(\mathcal{M}^{(k)}_{\nu}\otimes\mathcal{I}\right)-\mathrm{vec}\!\left(\mathcal{M}^{(k)}_{\mu_H}\otimes\mathcal{I}\right)}_\infty\\
    &= d^k\norm{\left(\underset{V \sim \nu}{\mathbb{E}}\left[V^{\otimes k,k}\right]-\ExU \left[U^{\otimes k,k}\right]\right)\otimes I^{\otimes k}\otimes I^{\otimes k}}_\infty\\
    &=d^k\norm{\underset{V \sim \nu}{\mathbb{E}}\left[V^{\otimes k,k}\right]-\ExU \left[U^{\otimes k,k}\right]}_\infty,
\end{align}
where in the last step we used the $p$-norm property $\norm{A\otimes B}_p=\norm{A}_p \norm{B}_p$.
Hence the first claim follows.
Now for the second claim, reviewing in reverse what we have shown for the first claim, we have:
\begin{align}
    \norm{\underset{V \sim \nu}{\mathbb{E}}\left[V^{\otimes k,k}\right]-\ExU \left[U^{\otimes k,k}\right]}_\infty=\sup_{X\neq 0}\frac{\norm{\Phi\otimes\mathcal{I}\left(X\right)}_2}{\norm{X}_2}.
\end{align}
By using that $\norm{A}_2 \le \norm{A}_1$, that $\norm{A}_1\le\sqrt{D}\norm{A}_2$ for $A\in \mathcal{L}(\mathbb{C}^D)$ and the definition of diamond norm, we have:
\begin{align}
    \norm{\underset{V \sim \nu}{\mathbb{E}}\left[V^{\otimes k,k}\right]-\ExU \left[U^{\otimes k,k}\right]}_\infty \le \sqrt{d^k} \norm{\mathcal{M}^{(k)}_{\nu}-\mathcal{M}^{(k)}_{\mu_H}}_{\lozenge}.
\end{align}
\end{proof}

Motivated by Proposition~\ref{prop:framedesign}, we introduce an alternative notion of approximation based on the frame potential. 
Notably, the frame potential may be easier to estimate numerically than the other approximation notions previously introduced (if approached naively by their definition), as it involves computing the trace of matrices in $\mathcal{L}(\mathbb{C}^{d})$, rather than working with matrices defined in tensor product spaces. Specifically, we define a frame potential $\varepsilon$-approximate $k$-design as follows:
\begin{definition}[Frame potential $\varepsilon$-approximate $k$-design]
We say that $\nu$ is a frame potential $\varepsilon$-approximate $k$-design if and only if:
\begin{align}
    \sqrt{\mathcal{F}^{(k)}_{\nu}-\mathcal{F}^{(k)}_{\mu_H}}\le \varepsilon.
\end{align}
\end{definition}
Note that the argument of this square root is always non-negative (due to Lemma~\ref{le:frameNorm}).

We now establish a relationship between the Frame potential notion of approximate $k$-design and the TPE one.
\begin{proposition}[Frame potential vs TPE relation]\
    \begin{itemize}
    \item We have:
    \begin{align}
        \sqrt{\mathcal{F}^{(k)}_{\nu}-\mathcal{F}^{(k)}_{\mu_H}} \le d^k \norm{\underset{V \sim \nu}{\mathbb{E}}\left[V^{\otimes k,k}\right]-\ExU \left[U^{\otimes k,k}\right]}_\infty.
    \end{align}
    Therefore, if $\nu$ is a TPE $\varepsilon$-approximate $k$-design, then $\nu$ is a frame potential $\varepsilon d^k $-approximate $k$-design.
    \item Conversely, we have:
    \begin{align}
        \norm{\underset{V \sim \nu}{\mathbb{E}}\left[V^{\otimes k,k}\right]-\ExU \left[U^{\otimes k,k}\right]}_\infty \le \sqrt{\mathcal{F}^{(k)}_{\nu}-\mathcal{F}^{(k)}_{\mu_H}}.
    \end{align}
    Thus, if $\nu$ is frame potential $\varepsilon$-approximate $k$-design, then $\nu$ is a TPE $\varepsilon$-approximate $k$-design.
\end{itemize}
\end{proposition}
\begin{proof}
By using Lemma~\eqref{le:frameNorm}, we have:
\begin{align}
    \sqrt{\mathcal{F}^{(k)}_{\nu}-\mathcal{F}^{(k)}_{\mu_H}}=\norm{\underset{V \sim \nu}{\mathbb{E}}\left[V^{\otimes k,k}\right]-\ExU \left[U^{\otimes k,k}\right]}_2.
\end{align}
Now the first claim follows just by using that $\norm{A}_2 \le \sqrt{D}\norm{A}_\infty$ for all $A\in \mathcal{L}(\mathbb{C}^D)$, while the second claim follows by $\norm{A}_\infty\le\norm{A}_2$. 
\end{proof}
In the following, we relate the diamond notion of approximation with the frame potential one.
\begin{proposition}[Diamond vs Frame potential]\
\begin{itemize}
    \item We have:
    \begin{align}
         \norm{\mathcal{M}^{(k)}_{\nu}-\mathcal{M}^{(k)}_{\mu_H}}_{\lozenge}\le d^k \sqrt{\mathcal{F}^{(k)}_{\nu}-\mathcal{F}^{(k)}_{\mu_H}} .
    \end{align}
    Hence, if $\nu$ is a frame potential $\varepsilon$-approximate $k$-design, then $\nu$ is a diamond $\varepsilon d^k $-approximate $k$-design.
    \item Conversely, we have:
\begin{align}
    \sqrt{\mathcal{F}^{(k)}_{\nu}-\mathcal{F}^{(k)}_{\mu_H}}\le d^{3k/2} \norm{\mathcal{M}^{(k)}_{\nu}-\mathcal{M}^{(k)}_{\mu_H}}_{\lozenge} .
\end{align}
    Hence, if $\nu$ is a diamond $\varepsilon$-approximate $k$-design, then $\nu$ is a frame potential $\varepsilon d^{3k/2}$-approximate $k$-design.
\end{itemize}
\end{proposition}
\begin{proof} 
Using Eq.\eqref{eq:diamondTPE}, we have:
    \begin{align}
        \norm{\mathcal{M}^{(k)}_{\nu}-\mathcal{M}^{(k)}_{\mu_H}}_{\lozenge}  \le d^k \norm{\underset{V \sim \nu}{\mathbb{E}}\left[V^{\otimes k,k}\right]-\ExU \left[U^{\otimes k,k}\right]}_\infty\le d^k \sqrt{\mathcal{F}^{(k)}_{\nu}-\mathcal{F}^{(k)}_{\mu_H}},
    \end{align}
where in the last step we used that $\norm{A}_\infty \le \norm{A}_2$ and Eq.\eqref{le:frameNorm}. Therefore we have shown the first claim.
For the second claim, we have: 
\begin{align}
    \sqrt{\mathcal{F}^{(k)}_{\nu}-\mathcal{F}^{(k)}_{\mu_H}}\le \sqrt{d^{2k}}\norm{\underset{V \sim \nu}{\mathbb{E}}\left[V^{\otimes k,k}\right]-\ExU \left[U^{\otimes k,k}\right]}_\infty \le d^k d^{k/2} \norm{\mathcal{M}^{(k)}_{\nu}-\mathcal{M}^{(k)}_{\mu_H}}_{\lozenge} ,
\end{align}
where in the first step we used Lemma~\ref{le:frameNorm} and that $\norm{A}_2 \le \sqrt{D}\norm{A}_\infty$ for $A\in \mathcal{L}(\mathbb{C}^D)$, while in the second step we used Eq.\eqref{eq:TPEdiamond}.
\end{proof}
Another notion of unitary design approximation was introduced in \cite{Brand_o_2016}, which we refer to as relative error $\varepsilon$-approximate $k$-design. This is defined as follows:
\begin{definition}[Relative error $\varepsilon$-approximate $k$-design]\
We say that $\nu$ is a relative error $\varepsilon$-approximate $k$-design if and only if 
    \begin{align}
        (1-\varepsilon)\mathcal{M}^{(k)}_{\nu}\preccurlyeq  \mathcal{M}^{(k)}_{\mu_H} \preccurlyeq (1+\varepsilon)\mathcal{M}^{(k)}_{\nu},
    \end{align}
    where $A \preccurlyeq B$ means that $B-A$ is completely positive, and $A$ and $B$ are linear superoperators.
\end{definition}
For further details on this approximation notion, see \cite{Brand_o_2016}.
In \cite{Brand_o_2016,Haferkamp2022randomquantum,Harrow_2009}, it was shown that one-dimensional local random quantum circuits, which are quantum circuits of $n$ qubits formed by $2$-qubit Haar random gates, of size $O\left(n^2 \mathrm{poly}(k)\right)$ are (relative error) $\varepsilon$-approximate $k$-designs for all $k \in O\left(2^{0.4n}\right)$. This result has been extended to higher dimensional quantum circuits in \cite{Harrow_2023}.
\section{Examples and applications}
\label{sec:applications}
This section explores diverse examples and applications where the Haar measure plays a fundamental role in quantum information. We derive well-known formulas that reduce to computing moments over the Haar measure, including the twirling of quantum channels and the average gate fidelity. These formulas lay the foundation for various applications, such as Randomized Benchmarking \cite{Knill_2008}.

Moreover, we show how these moment calculations can be translated into probability statements using concentration inequalities. Concentration inequalities serve as valuable tools for establishing rigorous bounds and enhancing our understanding of the statistical behavior associated with Haar random variables.

Furthermore, we provide detailed insights into two notable examples showcasing the applications of the theory of unitary design. We examine Barren Plateaus \cite{McClean_2018} in Variational Quantum Algorithms, shedding light on the optimization landscapes encountered in such algorithms. Additionally, we delve into Classical Shadow tomography \cite{Huang_2020}, where the theory of unitary design aids in designing efficient measurement strategies for reconstructing properties of unknown quantum states.
\subsection{Examples of moment calculations}
We start by deriving a formula that plays an important role in quantum information, in particular in Randomized Benchmarking \cite{Knill_2008,Emerson_2005}. This formula illustrates that when a unitary operator $U$, randomly chosen from a $2$-design, is applied before a quantum channel $\Phi$, followed by the application of its adjoint $U^\dagger$, the resulting output channel resembles a depolarizing channel.
\begin{example} [Twirling of a quantum channel is a depolarizing channel]
Let $\nu$ a unitary $2$-design distribution. Consider a quantum channel $\Phi: \mathcal{L}(\mathbb{C}^{d})\rightarrow \mathcal{L}(\mathbb{C}^{d})$ and a quantum state $\rho \in \mathcal{S}\left(\mathbb{C}^d\right)$. Then:
\begin{align}
\underset{U \sim \nu}{\mathbb{E}} \left[U^\dagger \Phi \left(U\rho U^\dagger\right)U\right]&=p_{\Phi}\rho + \left(1-p_{\Phi}\right)\Tr(\rho)\frac{I}{d},
\label{eq:twirling}
\end{align}
where the left-hand side represents the so-called \textquote{twirling} of $\Phi$, and we define:
\begin{align}
    p_{\Phi}\coloneqq \frac{d^2 F_e\left( \Phi \right) -1}{d^2-1}.
    \label{eq:ptwirl}
\end{align}
Here, $F_e\left( \Phi \right)$ denotes the entanglement fidelity given by $F_e\left( \Phi \right)\coloneqq \frac{1}{d^2}\bra{\Omega}\Phi \otimes \mathcal{I}(\ketbra{\Omega}{\Omega})\ket{\Omega}$, where $\mathcal{I}:\mathcal{L}(\mathbb{C}^{d})\rightarrow \mathcal{L}(\mathbb{C}^{d})$ represents the identity channel.
\end{example}
\begin{proof}
Due to the $2$-design property of $\nu$ and the fact that the quantity of interest is a second moment quantity, we can average over the Haar measure $\mu_H$ instead of $\nu$.
Considering a Kraus decomposition for $\Phi$ with operators $\{K_i\}^{d^2}_{i=1}$, we have: 
\begin{align}
\ExU \left[U^\dagger \Phi \left(U\rho U^\dagger\right)U\right]&=\sum^{d^2}_{i=1}\ExU \left[U^\dagger K_i U\rho U^\dagger K^\dagger_i U\right]\\&=\sum^{d^2}_{i=1}\ExU \Tr_2\left[(I\otimes \rho)U^{\dagger \otimes 2} (K_i\otimes K^\dagger_i) U^{\otimes 2} \mathbb{F}\right]\\
&=\sum^{d^2}_{i=1} \Tr_2\left[(I\otimes \rho)\ExU\left(U^{ \otimes 2} (K_i\otimes K^\dagger_i) U^{\dagger\otimes 2}\right) \mathbb{F}\right],
\end{align}
where in the second equality we used that $AB=\Tr_2\left(A\otimes B \,\mathbb{F}\right)$, and in the third equality that $\ExU \left[f(U)\right]=\ExU \left[f(U^\dagger)\right]$ (Proposition \ref{prop:Haar}) for all integrable functions $f$.
Using Eq.\eqref{eq:2momHaar} for the second moment, we have:
\begin{align}
\ExU\left[U^{ \otimes 2} \left(\sum^{d^2}_{i=1}K_i\otimes K^\dagger_i\right) U^{\dagger\otimes 2}\right]=c_{\Idd}\Idd + c_{\Flip}\Flip,    
\end{align} where the coefficients $c_{\mathbb{I}}$ and $c_{\mathbb{F}}$ are given by:
\begin{align}
c_{\mathbb{I}}&=\frac{\sum^{d^2}_{i=1}\Tr\!\left(K_i\otimes K^\dagger_i\right)-d^{-1}\Tr\!\left(\sum^{d^2}_{i=1}K_i\otimes K^\dagger_i\mathbb{F}\right)}{d^2-1}=\frac{\sum^{d^2}_{i=1}|\Tr\!\left(K_i\right)|^2-1}{d^2-1}\\
c_{\mathbb{F}}&=\frac{\Tr\!\left(\mathbb{F}\sum^{d^2}_{i=1}K_i\otimes K^\dagger_i\right)-d^{-1}\Tr\!\left(\sum^{d^2}_{i=1}K_i\otimes K^\dagger_i\right)}{d^2-1}=\frac{d-d^{-1}\sum^{d^2}_{i=1}\left|\Tr\!\left(K_i\right)\right|^2}{d^2-1},
\end{align}
where we used the \emph{swap-trick} and the fact that $\sum^{d^2}_{i=1} K^\dagger_i K_i = I$.
We observe that $c_\Flip=d^{-1}(1-c_{\Idd})$. Therefore:
\begin{align}
\ExU \left[U^\dagger \Phi \left(U\rho U^\dagger\right)U\right]&= \Tr_2\left[(I\otimes \rho)\left(c_{\Idd,}\Idd + c_{\Flip}\Flip\right) \mathbb{F}\right]\\
&= c_{\Idd}\Tr_2\left((I\otimes \rho)\Flip\right) + \frac{1}{d}(1-c_{\Idd})\Tr_2(\left(I\otimes \rho)\right)\\
&= c_{\Idd}\rho + (1-c_{\Idd})\Tr\!\left(\rho\right)\frac{I}{d}\,.
\end{align}

To conclude the proof, we observe that $c_{\Idd}=p_{\Phi}$, as defined in Eq.\eqref{eq:ptwirl}. This follows from the relationship $\sum^{d^2}_{i=1}\left|\Tr\!\left(K_i\right)\right|^2=d^2 F_e\left( \Phi \right)$, as it can be easily seen: 
\begin{align}
\label{eq:ent_fid}
F_e\left(\Phi\right)&\coloneqq \frac{1}{d^2}\bra{\Omega}\Phi \otimes \mathcal{I}(\ketbra{\Omega}{\Omega})\ket{\Omega}\\
     &=\sum^{d^2}_{i=1}\frac{1}{d^2}\bra{\Omega}  K_i\otimes I \ket{\Omega} \bra{\Omega}K^\dagger_i  \otimes I\ket{\Omega}\\&=\frac{1}{d^2}\sum^{d^2}_{i=1}\left|\Tr\!\left(K_i\right)\right|^2.
\end{align}
It is worth noting that this equality can also be understood visually with the help of diagrams.
\end{proof}
It should be noted that the above formula, along with many other derivations in this tutorial, is based on the calculation of the moment operator that depends on the commutant of the unitary group. However, if we were to compute the Haar expected value defined over another subgroup on the unitary group instead on the full unitary group, we could still use the same approach by characterizing the commutant of this subgroup.

The following example introduces the concept of average gate fidelity, a measure used to assess the quality of quantum gates \cite{PhysRevA.60.1888,Nielsen_2002}.
\begin{example}[Average gate fidelity]
    Let $\nu$ be a state $2$-design distribution. Consider a quantum channel $\Phi: \mathcal{L}(\mathbb{C}^{d})\rightarrow \mathcal{L}(\mathbb{C}^{d})$ and a unitary channel $\mathcal{U}\left(\cdot\right)=U\left(\cdot\right)U^\dagger$. Then, the average gate fidelity is given by:
    \begin{align}
        \underset{\ket{\psi} \sim \nu}{\mathbb{E}} 
 \left[ \bra{\psi}\mathcal{U}^{\dagger} \circ \Phi \left(\ketbra{\psi}{\psi}\right)\ket{\psi}\right]=\frac{dF_e\left(\mathcal{U}^\dagger \circ \Phi \right)+1}{d+1},
    \end{align}
    where $\mathcal{U}^\dagger\left(\cdot\right)=U^\dagger\left(\cdot\right)U$ represents the adjoint channel of $\mathcal{U}$, and $F_e\left( \Phi \right)\coloneqq \frac{1}{d^2}\bra{\Omega}\Phi \otimes \mathcal{I}(\ketbra{\Omega}{\Omega})\ket{\Omega}$ corresponds to the \emph{entanglement-fidelity}.
\end{example}
\begin{proof}
It is worth noting that this formula can be easily derived also from the previous Eq.~\ref{eq:twirling} without the need to explicitly compute any additional Haar moment. However, for the sake of completeness, we will still perform the derivation without using Eq.~\ref{eq:twirling}.
We have: 
\begin{align}
\underset{\ket{\psi} \sim \nu}{\mathbb{E}}  \left[ \bra{\psi}U^\dagger \Phi \left(\ketbra{\psi}{\psi}\right)U\ket{\psi}\right]&=\sum^{d^2}_{i=1} \Expsi\left[ \Tr\!\left(\ketbra{\psi}{\psi}U^\dagger K_i \ketbra{\psi}{\psi}K^\dagger_i U\right)\right]\\
&=\sum^{d^2}_{i=1} \Expsi\left[\Tr\!\left(\ketbra{\psi}{\psi}^{\otimes 2}\left(U^\dagger K_i \otimes K^\dagger_i U\right)\mathbb{F}\right)\right]\\
&= \Tr\!\left(\frac{\mathbb{I}+\mathbb{F}}{d(d+1)}\sum^{d^2}_{i=1}\left(U^\dagger K_i \otimes K^\dagger_i U\right)\mathbb{F}\right)\\
&=\frac{1}{d(d+1)}\left(\sum^{d^2}_{i=1}\Tr\!\left(U^\dagger K_i K^\dagger_i U\right)+ \sum^{d^2}_{i=1}\Tr\!\left(U^\dagger K_i \otimes K^\dagger_i U \right)\right)\\
&=\frac{1}{d(d+1)}\left(d+ \sum^{d^2}_{i=1}\left|\Tr\!\left(U^\dagger K_i\right)\right|^2\right),
\end{align}
where we used in the first equality the Kraus decomposition of $\Phi$ and the fact that $\nu$ is a $2$-design, in the second equality the \emph{swap-trick}, in the third equality we used linearity of the expected value and the moment operator formula \eqref{eq:MomPermSym}, in the fourth equality again the \emph{swap-trick}, in the fifth equality the fact that $\sum^{d^2}_{i=1}K^\dagger_i K_i= I$.
Finally, we observe that $F_e\left(\mathcal{U}^\dagger\circ\Phi\right)=\frac{1}{d^2}\sum^{d^2}_{i=1}\left|\Tr\!\left(U^\dagger K_i\right)\right|^2$ (see Eq.\eqref{eq:ent_fid}), which concludes the proof.

\end{proof}

The following example considers a bipartite state $\rho\in \mathcal{S}\left(\mathbb{C}^d\otimes \mathbb{C}^d\right)$ shared between Alice and Bob, who apply unitary operations $U$ and $U^*$ to their respective subsystems, where $U$ is sampled according to a 2-design \cite{horodecki1998reduction}. The resulting output state can be expressed as a combination of the maximally mixed state and the maximally entangled state $d^{-1}\ketbra{\Omega}{\Omega}$, and the coefficients of this combination depend on the overlap of the state $\rho$ with the maximally entangled state $d^{-1}\ketbra{\Omega}{\Omega}$.

\begin{example}\
Let $\nu$ be a unitary $2$-design distribution, and $\rho \in \mathcal{S}\left(\mathbb{C}^d\otimes \mathbb{C}^d\right)$ be a quantum state.
    \begin{align}
        \underset{U \sim \nu}{\mathbb{E}}\left[U\otimes U^* \rho U^\dagger \otimes U^{*\dagger}\right]=p_\rho \frac{\Idd}{d^2} + \left(1-p_\rho\right)\frac{\ketbra{\Omega}{\Omega}}{d},
    \end{align}
    where $p_\rho\coloneqq d^2\frac{\left(1-d^{-1}\bra{\Omega}\rho\ket{\Omega}\right)}{d^2-1}$ with $p_\rho\in \left[0,\frac{d^2}{d^2-1}\right]$ .
\end{example}
\begin{proof}
By considering the partial transpose over the second subsystem and the 2-design formula (Eq.\eqref{eq:2momHaar}), we obtain:
    \begin{align}
        \underset{U \sim \nu}{\mathbb{E}}\left[U\otimes U^* \rho U^\dagger \otimes U^{*\dagger}\right]=\ExU\left[U^{\otimes 2} \rho^{T_B} U^{\dagger \otimes 2}\right]^{T_B}=c_{\Idd,\rho^{T_B}} \Idd^{T_B} + c_{\Flip,\rho^{T_B}}\Flip^{T_B}=c_{\Idd,\rho^{T_B}} \Idd + c_{\Flip,\rho^{T_B}}\ketbra{\Omega}{\Omega},
    \end{align}
    with 
    \begin{align}
    c_{\Idd,\rho^{T_B}}=\frac{\Tr\!\left(\rho^{T_B}\right)-d^{-1}\Tr\!\left(\rho^{T_B}\mathbb{F}\right)}{d^2-1}=\frac{1-d^{-1}\bra{\Omega}\rho\ket{\Omega}}{d^2-1},\end{align} and 
    \begin{align}
    c_{\Flip,\rho^{T_B}}=\frac{\Tr\!\left(\mathbb{F}\rho^{T_B}\right)-d^{-1}\Tr\!\left(\rho^{T_B}\right)}{d^2-1}=\frac{\bra{\Omega}\rho\ket{\Omega}-d^{-1}}{d^2-1}.
    \end{align}
    Noting that $c_{\Idd,\rho^{T_B}}=p_\rho d^{-2} $ and  $c_{\Flip,\rho^{T_B}}=(1-p_\rho) d^{-1} $, we obtain the desired result.
\end{proof}
Now, let us consider the problem of evaluating the expected purity of the reduced state obtained from a pure state distributed according to a $2$-design distribution \cite{Lubkin}.
\begin{example}[Purity]
\label{ex:purityavg}
    Consider the complex Hilbert space of two-qudit systems $\mathcal{H}_A\otimes \mathcal{H}_B$ of dimensions respectively $d_A=\mathrm{dim}(\mathcal{H}_A)$ and $d_B=\mathrm{dim}(\mathcal{H}_B)$. 
    Given $\ket{\psi} \in \mathcal{H}_A\otimes \mathcal{H}_B$, let $\rho_A\coloneqq \Tr_B\left(\ketbra{\psi}{\psi}\right)$. We have:
    \begin{align}
        \underset{\ket{\psi} \sim \nu}{\mathbb{E}} \Tr\!\left(\rho_A^2\right)=\frac{d_B + d_A}{d_A d_B+1},
    \end{align}
    where $\nu$ is a $2$-design distribution.
    \begin{proof}
    We can express the expected value as follows:
        \begin{align}
        \underset{\ket{\psi} \sim \nu}{\mathbb{E}} \Tr\!\left(\rho_A^2\right)=\Expsi \Tr\!\left(\rho_A^{\otimes 2} \Flip_A \right)=\underset{\ket{\psi} \sim \nu}{\mathbb{E}} \Tr\!\left(\ketbra{\psi}{\psi}^{\otimes{2}}  \left(\Flip_A \otimes \Idd_{B}\right) \right),
    \end{align}
    where in the first equality we used the \emph{swap-trick} on the two copies of the subsystem $A$ and denoted with the subscript $A$ the Flip operator acting on such two copies, and in the second equality we used that $\rho_A=\Tr_B(\ketbra{\psi}{\psi} \Idd_A \otimes \Idd_B)$. In tensor diagrams this is also clear to visualize. 
    Since $\nu$ is a $2$-design, we have:
    \begin{align}
        \underset{\ket{\psi} \sim \nu}{\mathbb{E}} \Tr\!\left(\rho_A^2\right)&= \Tr\!\left(\Expsi\left[\ketbra{\psi}{\psi}^{\otimes{2}}\right]  \left(\Flip_A \otimes \Idd_{B}\right) \right)=\Tr\!\left(\left[\frac{\Idd_A \otimes \Idd_B + \Flip_A \otimes \Flip_B}{d_A d_B(d_A d_B+1)}\right]  \left(\Flip_A \otimes \Idd_{B}\right) \right)\\
        &=\frac{1}{d_A d_B(d_A d_B+1)}\left( d_A d^2_B + d^2_A d_B  \right)=\frac{d_B + d_A}{d_A d_B+1} .
    \end{align}
    \end{proof}
\end{example}
Now let us consider a system like in the previous example, where $d_A\le d_B$.
Since the entanglement entropy $S\left(\rho_A\right)\coloneqq -\Tr\!\left(\rho_A\log_2\rho_A\right)$ is lower bounded by the 2-Renyi entropy $S_2\left(\rho_A\right)\coloneqq -\log_2\left(\Tr\!\left(\rho^2_A\right)\right)$ (Fact 1 in \cite{ram2016renyi}), we have:
\begin{align}
    \underset{\ket{\psi} \sim \nu}{\mathbb{E}}  S\!\left(\rho_A\right)\ge \underset{\ket{\psi} \sim \nu}{\mathbb{E}}  S_2\!\left(\rho_A\right)=-\underset{\ket{\psi} \sim \nu}{\mathbb{E}} \log_2\!\left(\Tr\!\left(\rho^2_A\right)\right)\ge - \log_2\!\left(\underset{\ket{\psi} \sim \nu}{\mathbb{E}}\Tr\!\left(\rho^2_A\right)\right),
\end{align}
where we used Jensen's inequality with the fact that $-\log_2(x)$ is a convex function. By using the result from the previous Example~\ref{ex:purityavg}, we have:
\begin{align}
   \underset{\ket{\psi} \sim \nu}{\mathbb{E}} S\!\left(\rho_A\right)\ge  \log_2\!\left(\frac{d_A d_B+1}{d_B + d_A}\right)\ge \log_2\!\left(d_A\right)-\log_2\!\left(\frac{d_B + d_A}{d_B}\right) \ge \log_2\!\left(d_A\right) - 1,
\end{align}
Since $S\left(\rho_A\right)\le \log_2(d_A)$, we conclude that the expected value of the entanglement entropy of the reduced state of a pure state distributed according to a $2$-design is close to its maximum value.
However, it is important to note that for finite system sizes, there exists a gap between the exact average entanglement entropy and the maximum value, as demonstrated by Page \cite{Page_1993,Hayden_2006}:
\begin{equation}
  \underset{\ket{\psi} \sim \nu}{\mathbb{E}}  S\!\left(\rho_A\right)=\frac{1}{\ln(2)}\left(-\frac{d_A-1}{2 d_B}+\sum_{k=d_B+1}^{d_A d_B} \frac{1}{k} \right) >  \log_2\!\left(d_A\right)-\frac{1}{\ln(2)}\frac{d_A}{2d_B}.
\end{equation}
This value is commonly referred to as the \emph{Page entropy}.
For a more in-depth analysis, see \cite{Liu_2018}, where higher Renyi-entropies are also discussed.

\begin{example}[Expectation value]
\label{ex:expvalue}
Consider a unitary $k$-design distribution denoted by $\nu$. Let $O\in \MatC{d}$ be an operator. Then:
\begin{align}
    \underset{\ket{\psi} \sim \nu}{\mathbb{E}} \bra{\psi} O \ket{\psi}^{ k}=\frac{1}{\binom{k+d-1}{k}}\Tr\left[O^{\otimes k}P^{(d,k)}_{\mathrm{sym}} \right].
\end{align}
In the case of $k=2$, it simplifies to:
\begin{align}
    \underset{\ket{\psi} \sim \nu}{\mathbb{E}} \bra{\psi} O \ket{\psi}^{ 2}=\frac{1}{d\left(d+1\right)}\left( \Tr\!\left(O\right)^2+\Tr\!\left(O^2\right)\right).
\end{align}
\end{example}
\begin{proof}
This follows easily from Eq.\eqref{eq:MomPermSym}:
    \begin{align}
    \underset{\ket{\psi} \sim \nu}{\mathbb{E}} \bra{\psi} O \ket{\psi}^{ k}&=\Tr\left[O^{\otimes k}  \underset{\ket{\psi} \sim \nu}{\mathbb{E}}\left(\ketbra{\psi}{\psi}\right)^{\otimes k} \right]\\
    &=\frac{1}{\binom{k+d-1}{k}}\Tr\left[O^{\otimes k}P^{(d,k)}_{\mathrm{sym}} \right].
\end{align}
If $k=2$, then $P^{(d,k)}_{\mathrm{sym}}=\frac{1}{2}\left(\mathbb{I}+\mathbb{F}\right)$. Using that $\Tr\left[O^{\otimes 2} \right]=\Tr\left[O\right]^2$ and $\Tr\left[O^{\otimes 2} \mathbb{F}\right]=\Tr\left[O^2\right]$ (\emph{swap-trick}), leading to the desired expression.
\end{proof}

Applying Example \ref{ex:expvalue} with $O=\ketbra{\phi}{\phi}$ with $\ket{\phi}\in \mathbb{C}^d$ and noting that $\ketbra{\phi}{\phi}^{\otimes k}P^{(d,k)}_{\mathrm{sym}}=\ketbra{\phi}{\phi}^{\otimes k}$, we have the following result.

\begin{example}
\label{cor:overlap}
    Consider a unitary $k$-design distribution denoted by $\nu$.
    Let $\ket{\phi}\in \mathbb{C}^d$. The following holds:
\begin{align}
    \underset{\ket{\psi} \sim \nu}{\mathbb{E}} \left[\left|\braket{\psi}{\phi}\right|^{2k}\right]=\frac{1}{\binom{k+d-1}{k}}\le \frac{k!}{d^k}.
\end{align}
In particular, for $k=1$, we have $\underset{\ket{\psi} \sim \nu}{\mathbb{E}}\left[\left|\braket{\psi}{\phi}\right|^2\right]=\frac{1}{d}$. 
\end{example}

\subsection{Concentration inequalities}
By utilizing concentration inequalities, such as Markov's inequality, it is possible to establish a connection between the exponential decay of the expected value with the number of qubits and a probabilistic statement.
Markov's inequality states that for a non-negative random variable $X$ and any $\varepsilon > 0$, the probability that $X$ exceeds $\varepsilon$ is bounded by the ratio of the expected value of $X$ to $\varepsilon$:
\begin{align}
    \mathrm{Prob}\left(X\ge\varepsilon\right)\le \frac{\mathbb{E}\left[X\right]}{\varepsilon}.
\end{align}
In a more general form, if $g$ is a strictly increasing non-negative function, the inequality can be expressed as:
\begin{align}
    \mathrm{Prob}\left(X\ge\varepsilon\right)\le \frac{\mathbb{E}\left[g(X)\right]}{g(\varepsilon)}.
\end{align}
To do an example, consider $P \in \{I,X,Y,Z\}^{\otimes n} \backslash \{I^{\otimes n}\}$ where $X,Y,Z$ are Pauli matrices.
By applying Markov's inequality and utilizing Example \eqref{ex:expvalue}, we find that when sampling a (Haar) random state $\ket{\psi}$, the probability that the expectation value of $P$ exceeds a threshold $\varepsilon > 0$ decays exponentially with the number of qubits $n$:
\begin{align}
    \underset{\ket{\psi}\sim \mu_H}{\mathrm{Prob}}\left(\left|\bra{\psi}P\ket{\psi}\right|\ge\varepsilon\right)=\underset{\ket{\psi}\sim \mu_H}{\mathrm{Prob}}\left(\left|\bra{\psi}P\ket{\psi}\right|^2\ge\varepsilon^2\right)\le \frac{1}{(2^n+1)\varepsilon^2}.
    \label{eq:markovPauli}
\end{align}
The same holds if the expected value is taken with respect to a state $2$-design.

Another concentration inequality that can be particularly useful when analyzing the averages over Haar-random states is Levy's lemma.

\begin{lemma}[Levy's lemma \cite{Ledoux2001}]
Consider the set $\mathbb{S}^{2 d-1}\coloneqq \{v\in \mathbb{C}^{d} :\|v\|_2=1 \}$. Let $f: \mathbb{S}^{2 d-1} \rightarrow \mathbb{R}$ be a function satisfying the Lipschitz condition $|f(v)-f(w)| \leq L\|v-w\|_2$. For all $\varepsilon \geq 0$, we have the probability bound:
\begin{align}
\underset{\ket{\phi}\sim \mu_H}{\operatorname{Prob}}\left[\left|f\left(\phi\right)-\Expsi\left[f\left(\psi\right)\right]\right| \geq \varepsilon\right] \leq 2 \exp \left(-\frac{2 d \varepsilon^2}{9 \pi^3 L^2}\right).
\end{align}
\end{lemma}

Using Levy's lemma, we can improve the dependence on the number of qubits $n$ to a double exponential form compared to what we obtained by a simple application of Markov's inequality in Eq.~\eqref{eq:markovPauli}, as shown in the next example. However, this double exponential form comes at the cost of requiring the state to be drawn from the Haar measure rather than a state $2$-design.
\begin{example}
Let $O\in \mathrm{Herm}\left(\mathbb{C}^d\right)$ be a Hermitian operator. For all $\varepsilon \geq 0$, we have:
\begin{align}
    \underset{\ket{\psi}\sim \mu_H}{\mathrm{Prob}}\left(\left|\bra{\psi}O\ket{\psi}-\frac{\Tr(O)}{d}\right|\ge\varepsilon\right)\le 2 \exp \left(-\frac{ d \varepsilon^2}{18 \pi^3  \norm{O}^2_\infty}\right).
\end{align}
In particular, if $O$ is a Pauli string $P\in \{I,X,Y,Z\}^{\otimes n} \backslash \{I^{\otimes n}\}$, we have:
\begin{align}
    \underset{\ket{\psi}\sim \mu_H}{\mathrm{Prob}}\left(\left|\bra{\psi}P\ket{\psi}\right|\ge\varepsilon\right)\le 2 \exp \left(-\frac{ 2^n \varepsilon^2}{18 \pi^3}\right).
\end{align}
\end{example}
\begin{proof}
To apply Levy's lemma, we consider the function $f(\psi) = \bra{\psi}O\ket{\psi}$ and compute its expected value and Lipschitz constant.
First, we observe that
\begin{align}
\Expsi[f(\psi)]=\Expsi[\bra{\psi}O\ket{\psi}]=\Tr[O\Expsi\ketbra{\psi}{\psi}]=\frac{\Tr(O)}{d}.    
\end{align}

Next, we determine the Lipschitz constant. We have
\begin{align}
    |f(v)-f(w)| =\left|\Tr\left[\left(\ketbra{u}{u}-\ketbra{v}{v}\right)O\right]\right|\le \|O\|_\infty\|\ketbra{u}{u}-\ketbra{v}{v}\|_1
\end{align}
where we used the matrix Hölder inequality. We then observe that:
    \begin{align}
        \norm{\ketbra{u}{u}-\ketbra{v}{v}}_1=2 \sqrt{1-\left|\braket{u}{v}\right|^2}.
    \end{align}
This can be seen by considering $Q\coloneqq \ketbra{u}{u}-\ketbra{v}{v}$.
The Hermitian matrix $Q$ has a rank of at most $2$, which means it can have at most two non-zero eigenvalues denoted as $\lambda_1$ and $\lambda_2$. Since the trace of $Q$ is zero, we have $\lambda_2 = -\lambda_1$. Additionally, we know that $\Tr(Q^2) = \lambda_1^2 + \lambda_2^2 = 2\lambda_1^2$, and $\Tr(Q^2) = 2(1 - \left|\braket{u}{v}\right|^2)$. Therefore, we can conclude that $\lambda_1 = \sqrt{1 - \left|\braket{u}{v}\right|^2}$.
The $1$-norm of $Q$ is given by $\norm{Q}_1 = |\lambda_1|+|\lambda_2|$, which simplifies to $\norm{Q}_1=2\sqrt{1-\left|\braket{u}{v}\right|^2}$.

Moreover, we have the following inequality: 
\begin{align}
        2\sqrt{1-\left|\braket{u}{v}\right|^2}\le2\sqrt{2-2\left|\braket{u}{v}\right|}\le2\sqrt{2-2\Re(\braket{u}{v})}=2\norm{u-v}_2.
    \end{align}
This implies that
\begin{align}
        \norm{\ketbra{u}{u}-\ketbra{v}{v}}_1\le 2\norm{u-v}_2.
\end{align}
Hence, we have $|f(v)-f(w)|\le 2\|O\|_\infty \norm{u-v}_2 $.
By applying Levy's lemma with the Lipschitz constant of $f$ being $2\|O\|_\infty$, we can conclude.

\end{proof}
Combining Markov's inequality with higher moments provides a useful trick for computing bounds. By using $\mathrm{Prob}\left(X\ge\varepsilon\right)\le\mathrm{Prob}\left(X^k\ge\varepsilon^k\right)\le \frac{\mathbb{E}\left[X^k\right]}{\varepsilon^k}$ , we can derive bounds based on the $k$-th moment of a random variable.
This approach is particularly useful in situations involving $k$-unitary designs, where we can often compute moments only up to a given order.

Another useful inequality when averaging over the full Haar measure is the following. If $\alpha>0$, then $\mathrm{Prob}\left[X\ge\varepsilon\right]\le\mathrm{Prob}\left[\exp(\alpha X)\ge\exp(\alpha \varepsilon)\right]\le \exp(-\alpha\varepsilon)\mathbb{E}\left[\exp(\alpha X)\right]$.
To illustrate an application of this, let us consider an example provided also in the notes of Richard Kueng \cite{KuengNotes_2019}.
\begin{example}
Consider a $n$-qubit state $\ket{\phi}\in \mathbb{C}^d$ with $d=2^n$. If we randomly pick a state $\ket{\psi}$ from the Haar measure, the probability that the overlap between $\ket{\psi}$ and $\ket{\phi}$ is larger than $\varepsilon>0$ decays double exponentially with the number of qubits $n$:
\begin{align}
    \underset{\ket{\psi}\sim \mu_H}{\mathrm{Prob}}\left[\left|\braket{\psi}{\phi}\right|^2\ge\varepsilon\right]\le2\exp(-\frac{d}{2}\varepsilon).
\end{align}
\end{example}
\begin{proof}
We have:
\begin{align}
    \underset{\ket{\psi}\sim \mu_H}{\mathrm{Prob}}\left[\left|\braket{\psi}{\phi}\right|^2\ge\varepsilon\right]&
    \le\underset{\ket{\psi}\sim \mu_H}{\mathrm{Prob}}\left[\exp(\frac{d}{2}\left|\braket{\psi}{\phi}\right|^2)\ge\exp(\frac{d}{2}\varepsilon)\right]\\
    &\le \exp(-\frac{d}{2}\varepsilon)\Expsi\left[\exp(\frac{d}{2}\left|\braket{\psi}{\phi}\right|^2)\right]\\
    &=\exp(-\frac{d}{2}\varepsilon)\sum^\infty_{k=0} \frac{1}{k!}\frac{d^k}{2^k}\mathbb{E}\left|\braket{\psi}{\phi}\right|^{2k}\\
    &\le \exp(-\frac{d}{2}\varepsilon)\sum^\infty_{k=0}\frac{1}{2^k}= 2\exp(-\frac{d}{2}\varepsilon).
\end{align}
where we used the result of Example~\ref{cor:overlap}.
\end{proof}
For more in-depth exploration of the applications of concentration inequalities in the context of Haar random unitaries and $k$-designs, see \cite{Brand_o_2021,KuengNotes_2019,Low_2009,low2010pseudorandomness}.

\subsection{Applications in Quantum Machine Learning}
Let us now analyze an application of unitary designs in the context of Quantum Machine Learning, specifically in Variational Quantum Algorithms (VQAs) \cite{Cerezo_2021}. In VQAs, the problem at hand is formulated as the minimization of a cost function. This cost function is typically constructed using the expectation value of an observable, which is estimated on a quantum computer using a prepared quantum state that depends on the parameters of the quantum gates. These gate parameters are optimized with a classical optimizer with the goal of minimizing the cost function.

Specifically, consider a Hilbert space $\mathcal{H}$ associated with an $n$-qubit system of dimension $d=2^n$. In VQAs, it is customary to define the cost function as:
\begin{align}
C\left(\boldsymbol{\theta}\right)\coloneqq \Tr\left[U(\boldsymbol{\theta})\rho_0U^\dagger(\boldsymbol{\theta}) O\right],
\end{align}
where $O\in \mathrm{Herm}(\mathcal{H})$ represents an observable, $\rho_0 \in \mathcal{S}\left(\mathcal{H}\right)$ is a fixed initial state, and $U(\boldsymbol{\theta})$ denotes the \emph{parameterized} unitary transformation (or \emph{ansatze}) defined as:
\begin{align}
U(\boldsymbol{\theta})=\prod_{l=1}^L e^{-i \theta_l H_l}.
\end{align}
Here, $L$ is a positive integer representing the total number of parametrized quantum gates, $\theta_l \in \mathbb{R}$ are the parameters associated with the parametrized gates, and $H_l\in \mathrm{Herm}(\mathcal{H})$ are their Hamiltonian generators.
In the context of VQAs, it is often assumed that the observable $O$ can be expressed as a linear combination of Pauli operators as:
\begin{align}
    O=\sum^m_{i=1} a_i P_i,
\end{align}
where $P_i\in\{I,X,Y,Z\}^{\otimes n} \backslash \{I^{\otimes n}\}$, $m$ is a positive integer such that $m\in O\left(\mathrm{poly}(n)\right)$ and $a_i\in\mathbb{R}$ are coefficients satisfying $a_i\in O\left(\mathrm{poly}(n)\right)$ for all $i\in [m]$. These assumptions on $O$ also hold for all $H_l$.
It is worth noting that the defined observable $O$ satisfies the properties $\Tr\!\left(O\right)=0$ and $\Tr\!\left(O^2\right) \in O\left(\mathrm{poly}(n)2^n\right)$.
\begin{observation}
Let $\nu$ be a distribution defined over the set of unitaries $\{U(\boldsymbol{\theta})\}_{\boldsymbol{\theta}\in \mathbb{R}^L}$. If $\nu$ forms a 2-design distribution\footnote{For an analysis with approximate notions of designs, refer to \cite{Holmes_2022}.}, then the following properties hold:
\begin{align}
\underset{U \sim \nu}{\mathbb{E}}\left[C(\boldsymbol{\theta})\right]=0,\quad \underset{U\sim \nu}{\mathrm{Var}}\left[C(\boldsymbol{\theta})\right]\in O\!\left(\frac{\mathrm{poly}(n)}{2^n}\right).
\end{align}
\end{observation}
\begin{proof}
By Eq.\eqref{eq:1momHaar}, we have: 
    \begin{align}
    \underset{U \sim \nu}{\mathbb{E}}\left[C(\boldsymbol{\theta})\right]=\Tr\left[\underset{U \sim \nu}{\mathbb{E}}\left(U(\boldsymbol{\theta})\rho_0U^\dagger(\boldsymbol{\theta})\right)O\right]=\frac{\Tr(O)}{d}=0,
    \end{align}
where in the last step we used that $O$ is traceless.
The second equality follows from $\underset{U\sim \nu}{\mathrm{Var}}\left[C(\boldsymbol{\theta})\right]=\underset{U \sim \nu}{\mathbb{E}}\left[C(\boldsymbol{\theta})^2\right]-\underset{U \sim \nu}{\mathbb{E}}\left[C(\boldsymbol{\theta})\right]^2$ and Eq.\eqref{eq:2momHaar}:
\begin{align}
\underset{U \sim \nu}{\mathbb{E}}\left[C(\boldsymbol{\theta})^2\right]&=\Tr\left[\underset{U \sim \nu}{\mathbb{E}}\left(U^{\otimes 2}(\boldsymbol{\theta})\rho_0^{\otimes 2}U^{\dagger \otimes 2}(\boldsymbol{\theta})\right) O^{\otimes 2}\right]\\
&=c_{\Idd,\rho_0^{\otimes 2}}\Tr\left[\Idd O^{\otimes 2}\right]+c_{\Flip,\rho_0^{\otimes 2}} \Tr\left[\Flip O^{\otimes 2}\right]\\
&=c_{\Flip,\rho_0^{\otimes 2}} \Tr\left[O^2\right],
\end{align}
where $c_{\mathbb{F},\rho_0^{\otimes 2}}=\frac{\Tr\!\left(\mathbb{F}\rho_0^{\otimes 2}\right)-d^{-1}\Tr\!\left(\rho_0^{\otimes 2}\right)}{d^2-1}=\frac{\Tr\!\left(\rho_0^2\right)-2^{-n}}{2^{2n}-1}$, and in the last step we used that $O$ is traceless. The proof is concluded by noting that $c_{\mathbb{F},\rho_0^{\otimes 2}}\le \frac{1}{2^n(2^n+1)}$ and using that $\Tr\left[O^2\right]\in O\left(\mathrm{poly}(n)2^n\right)$.

\end{proof}
We can then apply Chebyshev's inequality, which states that for all $\varepsilon>0$ we have:
\begin{align}
 \underset{U\sim \nu}{\mathrm{Prob}}\left(\left|C(\boldsymbol{\theta})-\underset{U \sim \nu}{\mathbb{E}}\left[C(\boldsymbol{\theta})\right]\right|\ge\varepsilon \right)\le \frac{1}{\varepsilon^2}\underset{U\sim \nu}{\mathrm{Var}}\left[C(\boldsymbol{\theta})\right].
\end{align}
This inequality provides an upper bound on the probability of encountering a point in the parameter space where the cost function deviates from its expected value by more than $\varepsilon$. 
In particular, the probability of finding a point with a cost function larger than $\varepsilon$ decays exponentially with the number of qubits: $\underset{U\sim \nu}{\mathrm{Prob}}\left(\left|C(\boldsymbol{\theta})\right|\ge\varepsilon \right)\in O\left(\varepsilon^{-2}\frac{\mathrm{poly}(n)}{ 2^n}\right)$.

Similarly, with slightly more involved calculations, we can show that the exponential decay also applies to the variance of the partial derivatives of the cost function. This phenomenon, where the variance of the partial derivatives of the cost function decays exponentially with the number of qubits $n$, is commonly referred to as \emph{Barren Plateaus} \cite{McClean_2018}.

To analyze the partial derivatives of the cost function, we can express the parameterized unitary circuit $U(\boldsymbol{\theta})$ as the product of two unitary operators: $U(\boldsymbol{\theta}) = U_A U_B$, where $U_A=\prod_{l=\mu+1}^L e^{-i \theta_l H_l}$ and $U_B=\prod_{l=1}^\mu e^{-i \theta_l H_l}$. Consequently, we can write the partial derivative of the cost function as follows:
\begin{align}
\partial_\mu C\left(\boldsymbol{\theta}\right) & =\operatorname{Tr}\left[\left(\partial_\mu U(\boldsymbol{\theta})\right) \rho_0 U^{\dagger}(\boldsymbol{\theta}) O\right]+\operatorname{Tr}\left[U(\boldsymbol{\theta}) \rho_0\left(\partial_\mu U^{\dagger}(\boldsymbol{\theta})\right) O\right] \\
& =-i \operatorname{Tr}\left[ U_A H_\mu U_B \rho_0 U^\dagger_B U^\dagger_A O\right]+i\operatorname{Tr}\left[U_A U_B \rho_0 U^\dagger_B H_\mu U^\dagger_A O\right] \\
& =i \operatorname{Tr}\left[U_B \rho_0 U_B^{\dagger}\left[ H_\mu,U_A^{\dagger} O U_A\right]\right],
\end{align}
where we denoted by $\partial_\mu$ the partial derivative with respect to $\theta_\mu$, we used that $\partial_\mu U(\boldsymbol{\theta})=-i U_A H_\mu U_B $ and the cyclicity of the trace. Using this expression for the partial derivative of the cost function, we can prove the following:
\begin{observation}
Let $\nu_A, \nu_B$  be probability distributions defined over the sets of unitaries $\{U_A(\boldsymbol{\theta})\}_{\boldsymbol{\theta}\in \mathbb{R}^{L-\mu}}$ and $\{U_B(\boldsymbol{\theta})\}_{\boldsymbol{\theta}\in \mathbb{R}^{\mu}}$, respectively. Suppose that both $\nu_A$ and $\nu_B$ are $2$-designs distributions.\footnote{Check \cite{McClean_2018} for the case when only one between $\nu_A$ and $\nu_B$ is a $2$-design.} In this case, the following properties hold:
\begin{align}
\underset{\substack{U_A\sim  \nu_A\\U_B\sim \nu_B}}{\mathbb{E}}\left[\partial_\mu C(\boldsymbol{\theta})\right]=0,\quad \underset{\substack{U_A\sim  \nu_A\\U_B\sim \nu_B}}{\mathrm{Var}}\left[\partial_\mu C(\boldsymbol{\theta})\right]\in O\!\left(\frac{\mathrm{poly}(n)}{2^n-1}\right).
\end{align}
\end{observation}
\begin{proof}
    We have:
    \begin{align}
        \underset{\substack{U_A\sim  \nu_A\\U_B\sim \nu_B}}{\mathbb{E}}\left[\partial_\mu C(\boldsymbol{\theta})\right]&=i \Tr\!\left(\underset{U_B\sim \nu_B}{\mathbb{E}} b\left(U_B \rho_0 U_B^{\dagger}\right)\left[H_\mu,\underset{U_A\sim \nu_A}{\mathbb{E}}\left(U_A^{\dagger} O U_A\right)\right]\right)\\
        &=i \Tr\!\left(\underset{U_B\sim \nu_B}{\mathbb{E}}\left(U_B \rho_0 U_B^{\dagger}\right)\left[ H_\mu,\frac{\Tr(O)}{d}I\right]\right)\\
        &=0,
    \end{align}
    where we used that $O$ is traceless.
    Thus, for the variance, we have to compute $\underset{\substack{U_A\sim  \nu_A\\U_B\sim \nu_B}}{\mathrm{Var}}\left[\partial_\mu C(\boldsymbol{\theta})\right]=\underset{\substack{U_A\sim  \nu_A\\U_B\sim \nu_B}}{\mathbb{E}}\left[\partial_\mu C(\boldsymbol{\theta})^2\right]$: 
    \begin{align}
        \underset{\substack{U_A\sim  \nu_A\\U_B\sim \nu_B}}{\mathbb{E}}\left[\partial_\mu C(\boldsymbol{\theta})^2\right]&=-\underset{\substack{U_A\sim  \nu_A\\U_B\sim \nu_B}}{\mathbb{E}}\Tr\!\left(U_B \rho_0 U_B^{\dagger}\left[H_\mu, U_A^{\dagger} O U_A\right]\right)^2\\
        &=-\underset{\substack{U_A\sim  \nu_A\\U_B\sim \nu_B}}{\mathbb{E}}\Tr\!\left(U^{\otimes 2}_B \rho^{\otimes 2}_0 U_B^{\dagger \otimes 2}\left[H_\mu, U_A^{\dagger} O U_A\right]^{\otimes 2}\right)\\
        &=-c_{\Flip,\rho^{\otimes 2}}\underset{U_A\sim \nu_A}{\mathbb{E}}\Tr\!\left(\Flip\left[H_\mu, U_A^{\dagger} O U_A\right]^{\otimes 2}\right)\\
        &=-c_{\Flip,\rho^{\otimes 2}}\underset{U_A\sim \nu_A}{\mathbb{E}}\Tr\!\left(\left[H_\mu, U_A^{\dagger} O U_A\right]^2\right),
    \end{align}
where in the third equality we utilized Eq.\eqref{eq:2momHaar} as well as the fact that the trace of a commutator is always zero, we have defined $c_{\mathbb{F},\rho_0^{\otimes 2}}=\frac{\Tr\!\left(\mathbb{F}\rho_0^{\otimes 2}\right)-d^{-1}\Tr\!\left(\rho_0^{\otimes 2}\right)}{d^2-1}=\frac{\Tr\!\left(\rho_0^2\right)-d^{-1}}{d^2-1}$, and in the last step we used the \emph{swap-trick}.
Now, using the fact that:
\begin{align}
\Tr\!\left([A,B]^2\right)=2\Tr\!\left(ABAB\right)-2\Tr\!\left(AABB\right)=2\Tr\!\left(A^{\otimes 2}B^{\otimes2} \Flip\right)-2\Tr\!\left((A^2\otimes I) B^{\otimes 2} \Flip\right),
\end{align}
we have:
    \begin{align}
       & \underset{\substack{U_A\sim  \nu_A\\U_B\sim \nu_B}}{\mathbb{E}}\left[\partial_\mu C(\boldsymbol{\theta})^2\right]=-c_{\Flip,\rho^{\otimes 2}}\underset{U_A\sim \nu_A}{\mathbb{E}}\Tr\!\left(\left[H_\mu, U_A^{\dagger} O U_A\right]^2\right)\\
        &=-2c_{\Flip,\rho^{\otimes 2}}\left(\Tr\!\left(H_\mu^{\otimes 2}\underset{U_A\sim \nu_A}{\mathbb{E}}(U_A^{\dagger} O U_A)^{\otimes2} \Flip\right)-\Tr\!\left((H_\mu^2\otimes I) \underset{U_A\sim \nu_A}{\mathbb{E}}(U_A^{\dagger} O U_A)^{\otimes 2} \Flip\right)\right)\\
        &=-2c_{\Flip,\rho^{\otimes 2}}\left(c_{\Idd,O^{\otimes 2}}\Tr\!\left(H_\mu^{\otimes 2} \,\Idd\, \Flip\right)-c_{\Idd,O^{\otimes 2}}\Tr\!\left((H_\mu^2 \otimes I) \,\Idd \,\Flip\right)-c_{\Flip,O^{\otimes 2}}\Tr\!\left((H_\mu^2\otimes I)\, \Flip \,\Flip\right)\right)\\
        &=2dc_{\Flip,\rho^{\otimes 2}}c_{\Flip,O^{\otimes 2}}\Tr\!\left(H_\mu^2\right),
    \end{align}
    where in the third equality we used again Eq.\eqref{eq:2momHaar} along with the fact that $H_\mu$ is traceless, and we defined $c_{\mathbb{F},O^{\otimes 2}}$ as $c_{\mathbb{F},O^{\otimes 2}}=\frac{\Tr\left(\mathbb{F}O^{\otimes 2}\right)-d^{-1}\Tr\left(\O^{\otimes 2}\right)}{d^2-1}=\frac{\Tr\left(O^2\right)}{d^{2}-1}$.
    Thus, we have:
\begin{align}
    \underset{\substack{U_A\sim  \nu_A\\U_B\sim \nu_B}}{\mathbb{E}}\left[\partial_\mu C(\boldsymbol{\theta})^2\right]&\le\frac{2d}{(d^{2}-1)^2}(\Tr\!\left(\rho_0^2\right)-d^{-1})\Tr\!\left(O_\mu^2\right)\Tr\!\left(H_\mu^2\right)\le 2\frac{\Tr\!\left(O^2\right)}{(2^n-1)^2}\frac{\Tr\!\left(H_\mu^2\right)}{2^{n}+1}.
\end{align}
We can conclude by using that $\Tr\!\left[O^2\right]\in O\!\left(\mathrm{poly}(n)2^n\right)$ and $\Tr\left[H_\mu^2\right]\in O\!\left(\mathrm{poly}(n)2^n\right)$.
\end{proof}

Using again Chebyshev's inequality, we can conclude that the probability of finding a point in the parameter space with a partial derivative larger than a threshold $\varepsilon$ decreases exponentially with the number of qubits. Specifically, we have:  
\begin{align}
    \underset{\substack{U_A\sim  \nu_A\\U_B\sim \nu_B}}{\mathrm{Prob}}\!\left(\left|\partial_\mu C(\boldsymbol{\theta})\right|\ge\varepsilon \right)\in O\left(\varepsilon^{-2}\frac{\mathrm{poly}(n)}{2^n-1}\right).
\end{align}
There are several recent works that leverage Haar integration to compute the concentration of variational cost functions. For example, \cite{Larocca_2022,fontana2023adjoint,ragone2023unified} extend the Barren Plateaus diagnosing method (variance computation) to the case where the variational circuit lives in a unitary subgroup (over which it forms 2-designs), which is particularly relevant in the case of symmetric ansatzes~\cite{Larocca_2022Group,Meyer_2023}. 
Additionally, see~\cite{Wang_2021,singkanipa2024unital,mele2024noiseinduced} for cost-concentration analyses in the context of noisy circuits.

\subsection{Classical shadow tomography}
The \emph{classical shadow} protocol, introduced by Huang et al. \cite{Huang_2020}, is a method for predicting properties of quantum systems based on the theory of $k$-designs. In this tutorial, we will explore the classical shadow protocol as example of application.

Let $\rho \in \mathcal{S}\left(\mathbb{C}^d\right)$ be a quantum state of $n$-qubits, and let $O_1,\dots,O_M\in \mathrm{Herm}\left(\mathbb{C}^d\right)$ be Hermitian operators, where $d=2^n$.
The goal is to estimate $\Tr(O_1\rho),\dots,\Tr(O_M\rho)$ with a desired accuracy and probability of success. We assume that the full classical description of the state $\rho$ is unknown but it can be \textquote{queried} on a quantum device multiple times.
When the state $\rho$ is queried, a unitary $U$ is sampled randomly from a probability distribution $\mu$, and it is applied to $\rho$. The resulting state, $U\rho U^\dagger$, is measured in the computational basis $\{\ket{b}\}_{b\in [d]}$. Information about the sampled unitary $U$ and the measurement outcome $\ket{b}$ is stored in a classical memory, which can be efficiently implemented for appropriately chosen distributions of unitaries. The state $U^\dagger\ketbra{b}{b}U$ is referred to as a classical snapshot.

Now, the expected value of the classical snapshot $\mathbb{E}\left[U^\dagger\ketbra{b}{b}U\right]$ is considered, where $U$ is distributed according to the probability distribution $\mu$, and $b$ is distributed according to the Born's rule probability distribution $\bra{b}U\rho U^\dagger\ket{b}$.
More explicitly, we can associate a completely positive trace-preserving linear map $\mathcal{M}\left(\rho\right)$ to $\mathbb{E}\left[U^\dagger\ketbra{b}{b}U\right]$, denoted as \emph{measurement channel}: 
\begin{align}
\mathcal{M}\left(\rho\right)\coloneqq \sum^d_{b=1}\underset{U \sim \mu}{\mathbb{E}}\left[\bra{b}U \rho U^\dagger \ket{b}U^\dagger \ketbra{b}{b}U\right].
\end{align}

Assuming that $\mathcal{M}$ is invertible, 
$\hat{\rho}\coloneqq \mathcal{M}^{-1}\left(U^\dagger \ketbra{b}{b} U\right)$ serves as an unbiased estimator for $\rho$, meaning $\mathbb{E}\left[\hat{\rho}\right]=\rho$. The matrix $\hat{\rho}$ is commonly known as the \emph{classical shadow} of the state $\rho$.
 Consequently, $\hat{o}_i\coloneqq \Tr(O_i\hat{\rho})$ is an unbiased estimator for  $\Tr( O_i\rho)$:
\begin{align}
    \hat{o}_i\coloneqq \Tr(O_i \mathcal{M}^{-1}\left(U^\dagger \ketbra{b}{b} U\right)) \quad \text{implies}\quad\mathbb{E}\left[\hat{o}_i\right]=\Tr( O_i\rho) \text{  for all  } i\in [M].
    \label{eq:estimator}
\end{align}
For appropriately chosen probability distributions $\mu$, the estimator $\hat{o}_i$ can be efficiently computed classically.

To estimate the number $N$ of copies of $\rho$ (sample complexity) to achieve an additive accuracy $\varepsilon > 0$ in the estimation of $\Tr( O_i\rho)$ for all $i \in [M]$, with a failure probability of at most $\delta > 0$, it is important to bound the variance of the estimator:
\begin{align}
\mathrm{Var}\left(\hat{o}_i\right)\coloneqq \mathbb{E}\left[\hat{o}^2_i\right]-\mathbb{E}\left[\hat{o}_i\right]^2.
\label{eq:varshadow}
\end{align}
If the median of means \cite{Huang_2020,Koh_2022} is used as the estimator to post-process the data $\hat{o}_i$ for each $i \in [M]$, then a number of copies
\begin{align}
 N= O\left(\frac{\log\left(2M/\delta\right)}{\varepsilon^2}\underset{i\in [m]}{\mathrm{max}}\left[\mathrm{Var}\left(\hat{o}_i\right)\right]\right)
\end{align}
is enough to estimate, for each $i\in [m]$, $\Tr( O_i\rho)$ up to precision $\varepsilon$ with success probability at least $1-\delta$.

Therefore, the essential ingredients are the construction of the estimator through the inversion of the measurement channel $\mathcal{M}\left(\cdot\right)$ and bounding the variance. Below we observe that the measurement channel $\mathcal{M}\left(\cdot\right)$ is related to the second moment operator over the probability distribution $\mu$, while the variance is related to a third moment.
\begin{observation}
\label{prop:MeasVarObs}
    \begin{align}
    \label{eq:MeasChMom}
\mathcal{M}\left(\rho\right)&=\Tr_1\left(\left(\rho \otimes I\right) \left(\sum^d_{b=1}\underset{U \sim \mu}{\mathbb{E}}\left[U^{\dagger \otimes 2} \ketbra{b}{b}^{ \otimes 2} U^{ \otimes 2}\right]\right)\right)\\
\mathrm{Var}\left(\hat{o}_i\right)&=\Tr\!\left(\left(\rho \otimes \mathcal{M}^{-1}\left(O_i\right) \otimes \mathcal{M}^{-1}\left(O_i\right)\right) \left(\sum^d_{b=1}\underset{U \sim \mu}{\mathbb{E}}\left[U^{\dagger \otimes 3} \ketbra{b}{b}^{ \otimes 3} U^{ \otimes 3}\right]\right)\right)-\Tr( O_i\rho)^2 .
\end{align}
\end{observation}
\begin{proof}
For the measurement channel $\mathcal{M}(\rho)$, we have:
\begin{align}
\mathcal{M}\left(\rho\right)&\coloneqq \sum^d_{b=1}\underset{U \sim \mu}{\mathbb{E}}\left[\bra{b}U \rho U^\dagger \ket{b}U^\dagger \ketbra{b}{b}U\right]\\
&=\sum^d_{b=1}\underset{U \sim \mu}{\mathbb{E}}\left[\Tr\!\left(\rho U^\dagger \ketbra{b}{b} U  \right)U^\dagger \ketbra{b}{b}U\right],\\
&=\Tr_1\left(\left(\rho \otimes I\right) \left(\sum^d_{b=1}\underset{U \sim \mu}{\mathbb{E}}\left[U^{\dagger \otimes 2} \ketbra{b}{b}^{ \otimes 2} U^{ \otimes 2}\right]\right)\right).
\end{align}
For the variance Eq.\eqref{eq:varshadow}, we need to consider:
\begin{align}
\mathbb{E}\left(\hat{o}^2_i\right)&\coloneqq \sum^d_{b=1}\underset{U \sim \mu}{\mathbb{E}}\left[\bra{b}U \rho U^\dagger \ket{b}\Tr\!\left(O_i\mathcal{M}^{-1}\left(U^\dagger \ketbra{b}{b}U\right)\right)^2\right].
\end{align}
We first observe that $\Tr\!\left(A\mathcal{M}\left(B\right)\right)=\sum^d_{b=1}\underset{U \sim \mu}{\mathbb{E}}\left[\bra{b}U B U^\dagger \ket{b} \bra{b}U A U^\dagger \ket{b}\right]=\Tr\!\left(\mathcal{M}\left(A\right)B\right)$ for all $A,B \in \MatC{d}$. This implies that $\Tr\!\left(A\mathcal{M}^{-1}\left(B\right)\right)=\Tr\!\left(\mathcal{M}^{-1}\left(A\right)B\right)$ for all $A,B \in \MatC{d}$. Therefore:
\begin{align}
\mathbb{E}\left(\hat{o}^2_i\right)&=\sum^d_{b=1}\underset{U \sim \mu}{\mathbb{E}}\left[\bra{b}U \rho U^\dagger \ket{b}\Tr\!\left(\mathcal{M}^{-1}\left(O_i\right)U^\dagger \ketbra{b}{b}U\right)^2\right]\\
&=\sum^d_{b=1}\underset{U \sim \mu}{\mathbb{E}}\left[\Tr\!\left(\rho U^\dagger \ketbra{b}{b} U  \right)\Tr\!\left(\mathcal{M}^{-1}\left(O_i\right)U^\dagger \ketbra{b}{b}U\right)^2\right]\\
&=\Tr\!\left(\left(\rho \otimes \mathcal{M}^{-1}\left(O_i\right) \otimes \mathcal{M}^{-1}\left(O_i\right)\right) \left(\sum^d_{b=1}\underset{U \sim \mu}{\mathbb{E}}\left[U^{\dagger \otimes 3} \ketbra{b}{b}^{ \otimes 3} U^{ \otimes 3}\right]\right)\right).
\end{align}
\end{proof}

We will now consider the unitary probability distribution that corresponds to the uniform distribution over the Clifford group, that we defined in Eq.\eqref{eq:Cliffordgroup}. It is important to remember that the any unitary of the Clifford group can be implemented with $O(n^2/\log(n))$ elementary gates \cite{Aaronson_2004} and that there are efficient algorithms to sample uniformly from the Clifford group \cite{berg2021simple,Koenig_2014}. 

Since the uniform distribution over the Clifford group is an exact $3$-design, its first three moments coincide with those of the Haar measure.
Thus, we need to insert in Eq.\eqref{eq:MeasChMom} the formula for the second moment over the Haar measure to find the expression of the measurement channel $\mathcal{M}\left(\rho\right)$ and then invert it.

\begin{observation}
The measurement channel is:
\begin{align}
    \mathcal{M}\left(\rho\right)=\frac{1}{d+1}\left(\Tr\!\left(\rho\right)I+\rho\right).
\end{align}
Thus, its inverse is: \begin{align}
\label{eq:InverseChannelCl}
\mathcal{M}^{-1}\left(\rho\right)=(d+1)\rho - \Tr\!\left(\rho\right)I.
\end{align}
\end{observation}
\begin{proof}
We have:
\begin{align}
\mathcal{M}\left(\rho\right)&=\sum^d_{b=1}\Tr_1\left(\left(\rho \otimes I\right) \underset{U \sim \mu}{\mathbb{E}}\left[U^{ \otimes 2} \ketbra{b}{b}^{ \otimes 2} U^{\dagger \otimes 2}\right]\right)\\
&=\sum^d_{b=1}\Tr_1\left(\left(\rho \otimes I\right) \frac{\Idd + \Flip}{d(d+1)}\right)\\
&=\frac{1}{d+1}\left(\Tr\!\left(\rho\right)I+\rho\right),
\end{align}
where in the second equality we utilized Eq.\eqref{eq:MomPermSym} and in the second equality the \emph{partial-swap-trick} \eqref{eq:partialswap}.

Therefore, we have that $\mathcal{M}^{-1}\left(\rho\right)=(d+1)\rho - \Tr\!\left(\rho\right)I$, since it can be easily verified that this expression satisfies  $\mathcal{M}^{-1}\left(\mathcal{M}\left(\rho\right)\right)=\rho$.
\end{proof}

To bound the variance we need to compute a third moment over the Haar measure of the unitary group, due to the $3$-design property of the Clifford group.
\begin{observation}
The variance is bounded by $\mathrm{Var}\left(\hat{o}_i\right)\le 3 \Tr\!\left(O^2_i\right)$.
\end{observation}
\begin{proof}
By utilizing Observation~\ref{prop:MeasVarObs}, we have:
\begin{align}
\label{eq:varproof}
\mathrm{Var}\left(\hat{o}_i\right)&=\Tr\!\left(\left(\rho \otimes \mathcal{M}^{-1}\left(O_i\right) \otimes \mathcal{M}^{-1}\left(O_i\right)\right) \left(\sum^d_{b=1}\underset{U \sim \mu}{\mathbb{E}}\left[U^{\dagger \otimes 3} \ketbra{b}{b}^{ \otimes 3} U^{ \otimes 3}\right]\right)\right)-\Tr( O_i\rho)^2 .
\end{align}
Using Proposition~\ref{prop:Haar}, we know that $\underset{U \sim \mu}{\mathbb{E}}\left[U^{\dagger \otimes 3} \ketbra{b}{b}^{ \otimes 3} U^{ \otimes 2}\right]=\underset{U \sim \mu}{\mathbb{E}}\left[U^{ \otimes 3} \ketbra{b}{b}^{ \otimes 3} U^{ \dagger\otimes 3}\right]$. Therefore by applying Eq.\eqref{eq:MomPermSym}, we can express the third moment as:
\begin{align}
    &\underset{U \sim \mu}{\mathbb{E}}\left[U^{\dagger \otimes 3} \ketbra{b}{b}^{ \otimes 3} U^{ \otimes 3}\right]=\frac{1}{d(d+1)(d+2)}\left(\sum_{\pi \in S_3 }V_{d}(\pi)\right)\\
    &=\frac{1}{d(d+1)(d+2)}\left(\ipic{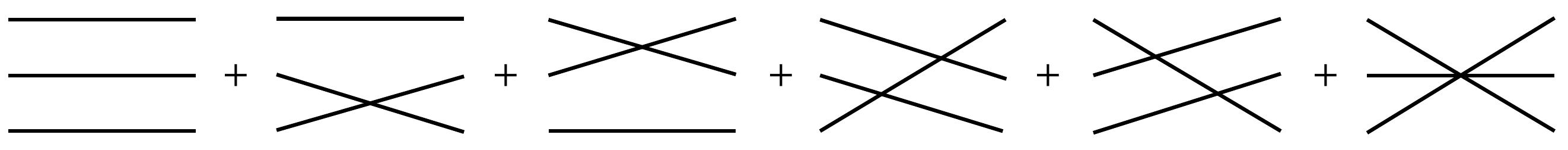}{0.21}\right),
\end{align}
where we used the Tensor Network representation of permutations defined in Section~\ref{sec:diagram}.
Substituting this expression into the first term of Eq.\eqref{eq:varproof}, up to a global factor $\left((d+1)(d+2)\right)^{-1}$, we have : 
\begin{align}
\label{eq:VarComplic}
& \Tr\!\left(\left(\rho \otimes \mathcal{M}^{-1}\left(O_i\right) \otimes \mathcal{M}^{-1}\left(O_i\right)\right) \left(\sum_{\pi \in S_3 }V_{d}(\pi)\right)\right) \nonumber\\
&=\Tr\!\left(\ipic{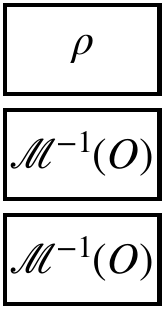}{0.33}\left(\ipic{images/s3.pdf}{0.26}\right)\right) \nonumber \\
&=\Tr\!\left(\rho\right)\Tr\!\left(\mathcal{M}^{-1}\left(O_i\right)\right)\Tr\!\left(\mathcal{M}^{-1}\left(O_i\right)\right)+\Tr\!\left(\rho\right)\Tr\!\left(\mathcal{M}^{-1}\left(O_i\right)^2\right)+\Tr\!\left(\rho\mathcal{M}^{-1}\left(O_i\right)\right)\Tr\!\left(\mathcal{M}^{-1}\left(O_i\right)\right)  \nonumber \\ 
&\quad+\Tr\!\left(\rho\mathcal{M}^{-1}\left(O_i\right)^2\right)+\Tr\!\left(\rho\mathcal{M}^{-1}\left(O_i\right)^2\right)+\Tr\!\left(\rho\mathcal{M}^{-1}\left(O_i\right)\right)\Tr\!\left(\mathcal{M}^{-1}\left(O_i\right)\right).
\end{align}
To simplify the calculations, we can utilize a trick exploited in \cite{Huang_2020}, namely that the variance $\mathrm{Var}\left(\hat{o}_i\right)$ only depends on the traceless part of the operator $O_i$, which is defined as $O^{(0)}_{i}\coloneqq O_i-\Tr(O_i)\frac{I}{d}$.

This follows by using that $\mathrm{Var}\left(\hat{o}_i\right)=\mathbb{E}\left[\left(\hat{o}_i-\Tr\!\left(O_i\rho\right) \right)^2\right]$ and observing that $\hat{o}_i-\Tr\!\left(O_i\rho\right) $ only depends on the traceless part of $O_i$. In fact, by denoting $\hat{\rho}=\mathcal{M}^{-1}\left(U^\dagger \ketbra{b}{b} U\right)$, we have:
\begin{align}
    \hat{o}_i-\Tr\!\left(O_i\rho\right)=\Tr(O_i \hat{\rho})-\Tr\!\left(O_i\rho\right)=\Tr(O^{(0)}_{i} \hat{\rho})-\Tr(O^{(0)}_{i}\rho),
\end{align}
where we used that $\Tr(\hat{\rho})=1$, since $\mathcal{M}^{-1}$ is trace-preserving (because $\mathcal{M}$ is trace-preserving). 

Therefore, by substituting $O_i$ with $O^{(0)}_{i}$ in Eq.\eqref{eq:VarComplic}, the right-hand side simplifies nicely:
\begin{align}
&\Tr\!\left(\left(\rho \otimes \mathcal{M}^{-1}\left(O^{(0)}_{i}\right) \otimes \mathcal{M}^{-1}\left(O^{(0)}_{i}\right)\right) \left(\sum_{\pi \in S_3 }V_{d}(\pi)\right)\right)\\
&=0+\Tr\!\left(\rho\right)\Tr\!\left(\mathcal{M}^{-1}\left(O^{(0)}_{i}\right)^2\right)+0+\Tr\!\left(\rho\mathcal{M}^{-1}\left(O^{(0)}_{i}\right)^2\right)+\Tr\!\left(\rho\mathcal{M}^{-1}\left(O^{(0)}_{i}\right)^2\right)+0
\\&=(d+1)^2\Tr\!\left(O^{(0) 2}_{i}\right)+2(d+1)^2\Tr\!\left(\rho O^{(0) 2}_{i}\right),
\end{align}
where we used the fact, according to Eq.\eqref{eq:InverseChannelCl}, that $\mathcal{M}^{-1}\left(O^{(0)}_{i}\right)=(d+1)O^{(0)}_{i}$ is traceless.
Substituting this expression back into Eq.\eqref{eq:varproof}, we have:
\begin{align}
    \mathrm{Var}\left(\hat{o}_i\right)&=\frac{(d+1)^2}{(d+1)(d+2)}\left(\Tr\!\left(O^{(0) 2}_{i}\right)+2\Tr\!\left(\rho O^{(0) 2}_{i}\right)\right)-\Tr(O^{(0)}_{i}\rho)^2\\
    &\le \Tr\!\left(O^{(0) 2}_{i}\right)+2\Tr\!\left(\rho O^{(0) 2}_{i}\right)\\
    &\le 3\Tr\!\left(O^{(0) 2}_{i}\right),
\end{align}
where in the second step we used that $I-\rho$ is positive semi-definite, and therefore $\Tr\!\left(\rho O^{(0) 2}_{i}\right)\le \Tr\!\left( O^{(0) 2}_{i}\right)$. 

Finally, we conclude the proof by noting that $\Tr\!\left(O^{(0) 2}_{i}\right)\le \Tr\!\left(O^{ 2}_{i}\right)$.
\end{proof}

Using the previous bound on the variance, we have that a number of copies 
\begin{align}
    N = O\!\left(\varepsilon^{-2}\log\left(2M/\delta\right)\underset{i\in [m]}{\mathrm{max}}\left[\Tr\!\left(O^{ 2}_{i}\right)\right]\right)
\end{align} 
suffices to estimate, for each $i\in [m]$, $\Tr( O_i\rho)$ up to precision $\varepsilon$ and with success probability at least $1-\delta$. 
Therefore, if we want to estimate observables with bounded $2$-norm, we have that the protocol is sample-efficient. 
However, it is important to emphasize that sampling efficiency does not imply that the protocol described is computationally efficient; in fact, to guarantee it, we must also be able to compute the unbiased estimator in Eq.\eqref{eq:estimator} in a time-efficient manner.
This is known for observables that have a particular structure such as stabilizer states (i.e., states constructed by the application of a Clifford circuit to a computational basis state); in this case, one can take advantage of the fact that computing the overlap between two stabilizer states can be done efficiently in classical computational time \cite{Aaronson_2004}.
See \cite{Huang_2020} for more details and also for checking the interesting case in which the uniform distribution over the unitary set formed by the tensor product of single-qubit Clifford gates is considered, which is better known as \textquote{random Pauli basis}. In this case, it turns out that the shadow protocol is both sample and time efficient for local observables like Paulis supported on a constant number of qubits. 
To explore additional examples of unitary distributions, one may refer for example to \cite{wan2022matchgate}, which discusses the utilization of the \emph{fermionic Gaussian unitaries} distribution, or to \cite{Hu2023,bertoni2023shallow}, where unitary distributions generated by local random quantum circuits are investigated. Additionally, \cite{bu2022classical} analyzes Pauli-invariant unitary ensembles. 

\section{Acknowledgments}
I thank Dax Enshan Koh, Hakop Pashayan, Philippe Faist, Johannes Jakob Meyer, Janek Denzler, Lorenzo Leone, Ellen Derbishire, Alexander Nietner, Francesco Anna Mele and Carlos Bravo-Prieto for insightful discussions and feedback. I am also grateful to the QML journal club of Jens Eisert's group for providing me the motivation to write this tutorial.

\DeclareFieldFormat{doi}{\mkbibacro{DOI}\addcolon\space\ifhyperref
	{\href{http://dx.doi.org/#1}{\nolinkurl{#1}}}
	{\nolinkurl{#1}}}

\printbibliography

\end{document}